\documentclass[12pt]{article}
\usepackage[utf8]{inputenc}
\usepackage{parskip}
\usepackage{setspace}
\usepackage{authblk}
\usepackage{amsmath,amssymb,amsthm}
\usepackage{mathtools}
\usepackage{bbm}
\usepackage{geometry}
\usepackage{hyperref}
\usepackage{xr-hyper}
\usepackage{bm}
\usepackage{graphicx}
\usepackage{epstopdf} 
\usepackage[labelsep=period,font=small]{caption}
\usepackage[font=footnotesize]{subcaption}
\makeatletter
\renewcommand\p@subfigure{\thefigure}
\makeatother
\usepackage[section]{placeins}
\usepackage{microtype}
\usepackage{natbib}
\usepackage{threeparttable}
\usepackage{booktabs}
\hypersetup{
    colorlinks=true,
    linkcolor=blue,
    citecolor=blue,
    urlcolor=blue
}
\usepackage{float}

\newtheorem{definition}{Definition}[section]
\newtheorem{proposition}{Proposition}[section]
\newtheorem{theorem}{Theorem}[section]
\newtheorem{lemma}{Lemma}[section]
\newtheorem{assumption}{Assumption}[section]

\newcommand{\Var}{\mathrm{Var}}

\allowdisplaybreaks[4]

\title{Optimal Life Insurance Decision in Mean-Variance DC Management with Mortality Improvements}

\author[$1$]{Yueman Feng}
\author[$2$]{Wenyuan Li}
\author[$3$]{Mengyi Xu}
\author[$4$]{Pengyu Wei}

\affil[$1$]{Department of Statistics and Actuarial Science, The University of Hong Kong}
\affil[$2$]{Department of Statistics and Actuarial Science, The University of Hong Kong}
\affil[$3$]{School of Risk and Actuarial Studies, UNSW Business School, UNSW Sydney}
\affil[$4$]{Division of Banking and Finance, Nanyang Business School, Nanyang Technological University}
\date{\today}

\begin{document}

\maketitle

\begin{abstract}
This paper studies the investment and insurance strategies of defined-contribution (DC) pension plans under the mean-variance framework. We consider a stochastic environment with time-varying interest rates, contributions, and mortality risk. The DC plan members are allowed to decide their bond and stock allocations, as well as their life insurance coverage. Adopting the martingale approach, we derive the closed-form optimal strategies and the mean-variance efficient frontier. Further numerical analysis investigates how mortality improvements affect investment and insurance decisions, as well as the sensitivity of the optimal decision to market parameters. Our analysis suggests that longevity raises expectations of future contributions, allowing pension members to adopt a less risky investment strategy. Meanwhile, insurance strategy shifts toward early adulthood to protect the high value of future income and decreases significantly at later ages. Moreover, we conduct sensitivity analyses on the target expected wealth, market price of risk, and contribution growth. These findings provide practical guidance for pension members on investment and offer insights for the design of DC pension plans.

\vspace{0.5cm}
\end{abstract}

\pagebreak

\section{Introduction}
\label{sec:introduction}

Longevity and population aging have been reshaping pension systems around the world. Occupational pension plans are also shifting from defined benefit (DB) plans to defined contribution (DC) plans to adapt to these changes. According to \cite{oecd2025pension}, in 2024, assets in DC plans increased by 11.4\% across OECD countries and other jurisdictions covered by the OECD's Global Pension Statistics. This significantly outperformed the 4.0\% increase in DB plans. Moreover, the proportion of DB plans in total pension plan assets dropped from 39.7\% in 2014 to only 31.9\% at the end of 2024. This evidence highlights the importance of DC plans in global pension markets.

Unlike the guaranteed retirement income provided by DB plans, DC plans do not provide such coverage. The benefits DC members receive depend on their contributions and investment outcomes. They are typically required to make their own investment decisions across decades, which exposes them to various risk sources, including inflation, interest rates, income uncertainty, and mortality risk. The long investment horizon and rapidly changing capital markets complicate decision-making and raise concerns about how to provide adequate retirement savings for DC members.

Various studies have contributed to DC pension management. The literature can be categorized into two groups: utility maximization and mean-variance (MV) optimization. For utility maximization, \citet{boulier2001protecteddc} consider a DC pension management problem with a retirement guarantee under a stochastic interest-rate model. They show that the optimal fund can be decomposed into three parts: a loan related to the contribution, a contingent claim linked to a guarantee, and a hedge fund. \cite{battocchio2004stochasticpension} further investigates the salary risk and inflation risk in the DC pension management. They emphasize that adding the inflation risk makes risk-free assets risky, and ignoring such risk leads to significant investment losses. For pertinent work, we refer to \cite{han2012dcinflation}, \cite{guan2014optimal}, \cite{chen2015regimeswitching}, \cite{tang2018asset}, \cite{dong2020optimal}, \cite{chen2023portfolio}, \cite{li2024dcmortality}, \cite{gao2025bayesian}, \cite{huang2025portfolio}, \cite{wang2025equilibrium}, and \cite{zhu2025robust}. For MV framework, \cite{hojgaard2007mean} first study the optimal DC allocation under the MV objective and compare the MV strategy with the target-based strategy and lifestyle strategy. \cite{yao2013meanvarianceinflation} studies how the inflation risk influences the MV pension member. They obtain an explicit efficient frontier using Lagrange dual theory and conclude that the pension member needs to take more risk given the expected return in a high-inflation economy. For further studies under the MV framework, we refer to \cite{he2013returnpremium}, \cite{yao2014stochasticincome}, \cite{wu2015gmmvmortality}, \cite{sun2016jumpdiffusion}, \cite{menoncin2017targetbased}, \cite{bian2018regimeswitching}. In the existing literature, although pre-retirement death risk has long been acknowledged, the MV literature typically addresses it through a return-of-premium clause. Life insurance is not considered as a control variable, and pre-retirement death risk is not fully hedged.

This research gap has clear institutional and economic concerns. At the practical level, some DC plans, such as Australia's superannuation system, have already incorporated life insurance as an investment choice. This compulsory DC pension scheme directly deducts life cover premiums from members' accumulated balances, implying that savings and protection are managed within a single vehicle rather than kept separate \citep{asic2026insurance}. Moreover, for the 401(k) pension fund in the U.S., life insurance purchase is also allowed to protect the pension members from pre-retirement death \citep{TreasReg1401}. At the theoretical level, mortality risk affects both the expected present value of future labor income and the demand for life insurance. If this risk is managed only passively through the return-of-premium clause or exogenously determined death benefits, pension members' hedging demand for pre-retirement death will be systematically underestimated. Moreover, mortality is not a static factor. Longevity improvements reshape the survival probabilities and change the human capital value as well as the marginal benefit of insurance. The optimal allocation between assets and protection will be distorted when mortality is not considered as a hedgeable risk.

To address these research gaps, this paper develops a continuous-time mean-variance optimization framework for DC pension members considering pre-retirement death. During the accumulation phase, members face stochastic interest rates, stochastic labor incomes, financial market risks, and mortality risk. The pension member is allowed to purchase a bond, a stock equity, and life insurance to minimize the variance of the terminal benefit given a fixed expected retirement payoff. When the pension member dies before retirement, the beneficiary receives a death benefit consisting of the account balance and life insurance payout. By adopting the martingale approach, we can transform the dynamic stochastic control problem into a static constrained one. We then derive the closed-form optimal investment and insurance strategies and further deduce the mean-variance efficient frontier. Finally, we use a verification theorem to validate the optimality of our solution.

Our result suggests that mortality improvement affects both asset allocation and insurance decisions by altering the expected present value of future labor income. Under our fully spanned income setting, mortality improvement increases future labor income, leading to greater exposure to fixed-income assets and reducing the need for equity investment. On the insurance side, demand shifts toward early life to protect the increased accumulation of future income. Meanwhile, insurance demand decreases more sharply at later stages due to the short time window left. Our sensitivity analysis also suggests that optimal strategies depend on multiple model parameters, such as the target expected wealth, market price of risk, and contribution growth.

Based on the framework and analysis, this paper makes three contributions to the existing DC pension literature. We first introduce life insurance as an asset-allocation option within mean-variance DC management. This addresses the gap in which premature death risk has been treated passively without a hedging option. Second, we further incorporate mortality improvement into the DC mean-variance framework and find how longevity affects individuals' investment and insurance decisions by adjusting their expectations of accumulated labor income. This directly relates to the current regulations in DC plans. For example, the group life insurance coverage provided by Australia's superannuation system should be calibrated with explicit reference to cohort-specific mortality improvement. Ignoring longevity trends can result in systematic bias in coverage. Finally, we characterize how the target wealth level and key economic drivers, including market price of risk and contribution growth, shape individuals' optimal investment and insurance strategies. This provides an operational framework for pension managers to adjust the balance between equity exposure and insurance coverage as market conditions and member characteristics evolve.

The remainder of this paper is organized as follows. Section~\ref{sec:economic-setting} introduces the economic and model settings used in this study. Section~\ref{sec:optimization} constructs the MV optimization problem and derives the optimal strategy and efficient frontier. We also provide a verification theorem for the optimal solution at the end of this section. Section~\ref{sec:numerical-analysis} examines how mortality improvement affects optimal asset allocation and insurance demand. We also conduct a sensitivity analysis on key economic drivers. Finally, Section~\ref{sec:conclusion} concludes the paper.

\section{Economic setting}
\label{sec:economic-setting}

\subsection{Financial market}
\label{subsec:financial-market}

We consider a financial market similar to that presented in \citet{munk2010dynamic}. Let $(\Omega, \mathcal{F}, \mathbb{P})$ be a complete filtered probability space. The flow of information is represented by the filtration $\mathcal{F} := \{\mathcal{F}_t\}_{t \in [0,T]}$, which is generated by a three-dimensional vector of independent standard Brownian motions $Z_t = (Z_{r,t}, Z_{S,t}, 
Z_{C,t})^\top$.

We assume the instantaneous risk-free interest rate, $r_t$, is governed by the mean-reverting Vasicek process
\begin{equation}
\label{eq:short-rate}
dr_t = \kappa(\bar{r} - r_t) \, dt - \bar{\sigma}_r^\top dZ_t,
\end{equation}
where $\kappa > 0$ dictates the speed of mean reversion towards the long-term level $\bar{r} > 0$. The interest rate volatility vector is defined as $\bar{\sigma}_r = (\sigma_r, 0, 0)^\top$, implying that the short rate is exclusively driven by the first Brownian motion, $Z_{r,t}$.

The financial market provides the investor with two tradable risky assets: a zero-coupon bond and an equity index. The price dynamics of the bond, $B_t$, and the stock index, $S_t$, are described by the following stochastic differential equations, respectively
\begin{equation*}
\begin{aligned}
dB_t &= B_t (r_t + \sigma_1^\top \Lambda) \, dt + B_t \sigma_1^\top dZ_t, \\
dS_t &= S_t (r_t + \sigma_2^\top \Lambda) \, dt + S_t \sigma_2^\top dZ_t,
\end{aligned}
\end{equation*}
where $\Lambda = (\lambda_r, \lambda_S, 0)^\top$ denotes the market price of risk. We denote the time-to-maturity of the zero-coupon bond by $\tau_B>0$. The volatility vectors for the bond and stock are respectively given by $\sigma_1^\top = (\sigma_r \beta(\tau_B), 0, 0)$ and $\sigma_2^\top = (\sigma_S \rho_{SB}, \sigma_S \sqrt{1-\rho_{SB}^2}, 0)$. Here, $\beta(\tau)=(1-e^{-\kappa\tau})/\kappa$ captures the interest rate sensitivity of a zero-coupon bond with remaining maturity $\tau$, $\sigma_S$ is the stock return volatility, and $\rho_{SB}$ represents the correlation between the equity and bond.

The investor continuously adjusts their portfolio, with the absolute amounts of wealth allocated to the bond and the stock denoted by the vector $\pi_t = (\pi_{B,t}, \pi_{S,t})^\top$. The remaining wealth is held in the risk-free cash account earning the short rate $r_t$.

To preclude arbitrage opportunities within this market, we fix a strictly positive pricing kernel, $\phi_t$, which follows
\begin{equation}
\label{eq:pricing-kernel-sde}
\frac{d\phi_t}{\phi_t} = -r_t \, dt - \Lambda^\top dZ_t, \quad \phi_0 = 1.
\end{equation}
Applying It\^o's lemma yields the explicit exponential representation for the pricing kernel
\begin{equation*}
\phi_t = \exp\left(
- \int_0^t r_s \, ds
- \int_0^t \Lambda^\top dZ_s
- \frac{1}{2} \int_0^t \|\Lambda\|^2 ds
\right).
\end{equation*}

\subsection{Spanned contribution dynamics}
\label{subsec:contribution}

In addition to financial investments, the account receives a stochastic stream of contributions, $C_t$, which follows a geometric Brownian motion with time-varying drift:
\begin{equation}
\label{eq:contribution-sde}
dC_t = C_t \mu_C(t) \, dt + C_t \sigma_3^\top dZ_t,
\end{equation}
where the deterministic drift function $\mu_C(t) = \xi_0(x_0+t) + \xi_1$ captures the expected growth rate with $x_0$ denoting the individual's entry age. The general contribution volatility vector is formulated as
\[
\sigma_3^\top = (\sigma_C \rho_{CB}, \sigma_C \hat{\rho}_{CS}, \sigma_C \sqrt{1-\rho_{CB}^2 - \hat{\rho}_{CS}^2}).
\]
Here, $\sigma_C$ is the contribution volatility, $\rho_{CB}$ is its loading on the bond or interest rate shock, and $\hat{\rho}_{CS}$ is its loading on the adjusted stock-specific shock after eliminating the bond shock; the hat distinguishes this adjusted loading from a raw stock-contribution correlation.
Consequently, the comprehensive $3 \times 3$ volatility matrix governing the entire economic system can be simplified to a lower triangular form
\begin{equation*}
\Sigma =
\begin{pmatrix}
\sigma_1^\top \\
\sigma_2^\top \\
\sigma_3^\top
\end{pmatrix}
=
\begin{pmatrix}
\sigma_r \beta(\tau_B) & 0 & 0 \\
\sigma_S \rho_{SB} & \sigma_S \sqrt{1-\rho_{SB}^2} & 0 \\
\sigma_C \rho_{CB} & \sigma_C \hat{\rho}_{CS} & \sigma_C \sqrt{1-\rho_{CB}^2 - \hat{\rho}_{CS}^2}
\end{pmatrix}.
\end{equation*}
Let $\mathcal{G}:=\{\mathcal{G}_t\}_{0 \le t \le T}$ denote the market-spanned filtration generated by $(Z_{r,t}, Z_{S,t})$.
To ensure market completeness, we follow \citet{munk2010dynamic} and impose the \emph{no unspanned income risk} condition:
\begin{equation}
\label{eq:no-unspanned-risk}
\rho_{CB}^2 + \hat{\rho}_{CS}^2 = 1.
\end{equation}
Economically, this condition eliminates the idiosyncratic shock component in the contribution process. It implies that the contribution risk is perfectly spanned by the tradable financial assets (the bond and the stock). Under \eqref{eq:no-unspanned-risk}, the third component of $\sigma_3$ vanishes, so the contribution process is adapted to $\mathcal{G}$. 

To summarize the risk exposures of the tradable assets compactly, we define the financial market volatility matrix, $\Sigma_F$, as:
\begin{equation}
\label{eq:sigma-F}
\Sigma_F = \begin{pmatrix} \sigma_1^\top \\ \sigma_2^\top \end{pmatrix}
= \begin{pmatrix}
\sigma_r \beta(\tau_B) & 0 & 0 \\
\sigma_S \rho_{SB} & \sigma_S \sqrt{1-\rho_{SB}^2} & 0
\end{pmatrix}.
\end{equation}

\subsection{Mortality}
\label{subsec:mortality}

Let $T_x \ge 0$ denote the remaining lifetime of an individual aged $x$. We assume $T_x$ is a random variable strictly independent of the financial market filtration $\mathcal{F}$.

Define the survival and death probabilities:
\begin{eqnarray*}
	_{t}p_x = \mathbb{P}[T_x>t], \quad 
	_{t}q_x = \mathbb{P}[T_x\leq t] = 1- {_{t}p_x}, \quad 
	\lim \limits_{t\rightarrow \infty} {_{t}p_x} =0, \quad 
 \lim \limits_{t\rightarrow \infty} {_{t}q_x} =1,
\end{eqnarray*}
and the instantaneous force of mortality $\mu_{x+t} = -\dfrac{d}{dt} \ln ({}_t p_x)$, so that
\[
{}_t p_x = \exp\left( - \int_0^t \mu_{x+s} ds \right), \quad
{}_t q_x = \int_0^t {}_s p_x \mu_{x+s} ds.
\]



\subsection{Wealth process}
\label{subsec:wealth-process}

The individual enters the DC pension plan at age $x_0$ at time 0 and retires at time $T$. During the accumulation period, the individual dynamically allocates his or her wealth between the zero-coupon bond and the stock index. In particular, the individual uses the rolling bond strategy for purchasing bonds as outlined in \citet{munk2010dynamic}. Let the vector $\pi_t$ denote the absolute amount of wealth allocated to these financial instruments at time $t$, with the remaining balance held in the risk-free cash account introduced before.

Additionally, the individual can continuously purchase life insurance to hedge against pre-retirement death risk. Let $I_t$ denote the continuous insurance-premium control conditional on survival. If the individual passes away at time $T_x = t$ prior to the retirement date $T$, the designated beneficiary will receive the insurance face value, $I_t/\mu_{x+t}$, alongside the accumulated wealth balance.

Denote the initial wealth by $W_0$. Incorporating the stochastic contribution $C_t$, the DC account balance $W_t$ then evolves according to the following stochastic differential equation
\begin{equation}
\label{eq:wealth-sde-original}
dW_t = ( r_t W_t + \pi_t^\top \Sigma_F \Lambda ) \, dt + \pi_t^\top \Sigma_F dZ_t + (C_t - I_t) \, dt,
\end{equation}
for $0 \le t < T$. By rearranging the terms, we can equivalently express the dynamics of the wealth process as
\begin{equation}
\label{eq:wealth-sde-M}
dW_t =
\left[ (r_t + \mu_{x+t}) W_t + \pi_t^\top \Sigma_F \Lambda \right] dt
+ \pi_t^\top \Sigma_F dZ_t
+ C_t \, dt - \mu_{x+t} M_t \, dt,
\end{equation}
where $M_t$ is the bequest value, equal to the account balance plus the death benefit delivered to the beneficiary if the individual dies at time $T_x = t < T$
\begin{equation}
\label{eq:bequest-M}
M_t = W_{t} + \frac{I_t}{\mu_{x+t}}.
\end{equation}

\section{Optimization problem}
\label{sec:optimization}

\subsection{Problem formulation and admissible strategies}
\label{subsec:problem-formulation}

To ensure the stochastic control problem is mathematically well-posed, we restrict the agent's decisions to a specific class of admissible strategies.

\begin{definition}[Admissible strategies]
\label{def:admissible}
The investment and insurance strategy pair $(\pi, I)$ is admissible, denoted by $(\pi, I) \in \mathcal{A}(0, T)$, if and only if it satisfies the following conditions:
\begin{enumerate}
    \item The portfolio process $\pi_t \in \mathbb{R}^2$, defined for $0 \le t < T$, and the insurance premium process $I_t \in \mathbb{R}$, defined for $0 \le t < T$, are progressively measurable with respect to the filtration $\mathcal{F} = \{\mathcal{F}_t\}_{0 \le t \le T}$;
    \item The wealth process driven by Equation~\eqref{eq:wealth-sde-original} on $0 \le t < T$ admits a unique strong solution $W_t$, with the initial wealth $W_0 = w_0$; the actually realized pension account is the stopped process $W_{t \wedge T_x}$;
    \item The control variables satisfy the square-integrability conditions:
    \begin{equation*}
    \mathbb{E}\left[ \int_0^T \|\pi_t\|^2 dt \right] < \infty, \quad
    \mathbb{E}\left[ \int_0^T I_t^2 dt \right] < \infty;
    \end{equation*}
    \item The resulting wealth process $W_t$ satisfies the uniform square-integrability condition:
    \begin{equation*}
    \mathbb{E}\left[ \sup_{t \in [0,T]} |W_t|^2 \right] < \infty.
    \end{equation*}
\end{enumerate}
\end{definition}

Let $Y_{T\wedge T_x}$ denote the payout of the DC pension plan, then it satisfies
\begin{equation}
Y_{T\wedge T_x} = W_T \mathbbm{1}_{\{T_x > T\}} + M_{T_x} \mathbbm{1}_{\{T_x \leq T\}}.
\end{equation}
The pension member seeks to achieve a predetermined expected return while minimizing financial risk
\begin{align}\label{eq:obj1}
	\min \Var(Y_{T\wedge T_x}), \text{ s.t. } \mathbb{E}[Y_{T\wedge T_x}] = \mathcal{K},
\end{align}
where $ \mathcal{K} $ is the expected required return. It is obvious that $\Var(Y_{T\wedge T_x})=\mathbb{E}[Y^2_{T\wedge T_x}]-(\mathbb{E}[Y_{T\wedge T_x}])^2=\mathbb{E}[Y^2_{T\wedge T_x}]-\mathcal{K}^2$ under the constraint $\mathbb{E}[Y_{T\wedge T_x}] = \mathcal{K}$. Thus, the optimization problem \eqref{eq:obj1} can be simplied to 
\begin{align}\label{eq:obj2}
	\min \mathbb{E}[Y^2_{T\wedge T_x}], \text{ s.t. } \mathbb{E}[Y_{T\wedge T_x}] = \mathcal{K}.
\end{align}

Let the expected target be $\mathcal{K} > 0$, then at $t=0$, the target expectation constraint is given by
\begin{equation*}
\label{eq:dynamic-constraint}
 \mathbb{E}[Y_{T\wedge T_x}]={_T p_x} \mathbb{E}[W_T] + \int_0^T {}_t p_x \mu_{x+t} \mathbb{E}[M_t] \, dt = \mathcal{K}.
\end{equation*}
Furthermore, the second moment can be expressed as 
\begin{equation*}
\mathbb{E}[Y_T^2] = {}_T p_x \mathbb{E}[W_T^2] + \int_0^T {}_t p_x \mu_{x+t} \mathbb{E}[M_t^2] \, dt.
\end{equation*}

Finally, the pension member's optimization problem \eqref{eq:obj2} can be rewritten as
\begin{align}
\label{eq:dynamic-objective}
\min_{(\pi, I) \in \mathcal{A}(0,T)} \quad & {}_T p_x \mathbb{E}[W_T^2] + \int_0^T {}_t p_x \mu_{x+t} \mathbb{E}[M_t^2] \, dt, \\
\label{eq:dynamic-constraint}
\text{subject to} \quad & {}_T p_x \mathbb{E}[W_T] + \int_0^T {}_t p_x \mu_{x+t} \mathbb{E}[M_t] \, dt = \mathcal{K}.
\end{align}

\subsection{The martingale approach and static optimization}
\label{subsec:martingale-approach}

We employ the martingale method to convert the dynamic stochastic control problem into a static optimization problem. The following two propositions establish the necessary and sufficient conditions for the equivalence between the dynamic wealth process and the static budget constraint.

\begin{proposition}
\label{prop:static-budget-necessity}
For any admissible strategy $(\pi, I) \in \mathcal{A}(0,T)$,
\begin{equation}
\label{eq:budget-necessity}
\mathbb{E}\left[
\int_0^T {}_t p_x \phi_t (\mu_{x+t} M_t - C_t) \, dt
+ {}_T p_x \phi_T W_T
\right] = w_0.
\end{equation}
\end{proposition}

\begin{proof}
    See Appendix~\ref{app:proof-prop-static-necessity}.
\end{proof}

\begin{proposition}
\label{prop:static-budget-sufficiency}
Assume additionally that $\Sigma_F \, \Sigma_F^\top \succ 0$, where $\succ$ means positive definite. For any progressively measurable process $M = \{M_t : 0 \le t \le T\}$ adapted to the market-spanned filtration $\mathcal{G}$ and satisfying $\mathbb{E}\left[ \int_0^T |M_t|^2 dt \right] < \infty$, and any $\mathcal{G}_T$-measurable random variable $\eta$ with $\mathbb{E}[\phi_T^2 \eta^2] < \infty$ such that
\begin{equation}
\label{eq:budget-sufficiency}
\mathbb{E}\left[
\int_0^T {}_t p_x \phi_t (\mu_{x+t} M_t - C_t) \, dt
+ {}_T p_x \phi_T \eta
\right] = w_0,
\end{equation}
there exists a portfolio process $\pi = \{\pi_t : 0 \le t < T\}$ and an associated wealth process $W$ with $W_T = \eta$ such that, upon defining $I_t = \mu_{x+t}(M_t - W_t)$, the pair $(\pi, I)$ belongs to $\mathcal{A}(0,T)$.
\end{proposition} 
\begin{proof}
    See Appendix~\ref{app:proof-prop-static-sufficiency}.
\end{proof}

\begin{lemma}
\label{lem:conditional-expectations}
For any time $t \in [0, T]$ and $u \ge t$, let $\tau = u - t$. The conditional expectations under measure $\mathbb{P}$ are given by:
\begin{align}
\label{eq:Et-phiu}
\mathbb{E}_t[\phi_u] &= \phi_t \exp\left[f_1(\tau) + f_2(\tau) r_t\right],\\
\mathbb{E}_t[\phi_u^2] &= \phi_t^2 \exp\left[f_3(\tau) + f_4(\tau) r_t\right], \notag \\
\mathbb{E}_t[\phi_u C_u] &= \phi_t C_t \exp\left[f_5(t, \tau) + f_6(\tau) r_t\right], \notag
\end{align}
where the deterministic functions $f_1(\tau), \ldots, f_4(\tau), f_6(\tau)$, together with $f_5(t,\tau)$, satisfy the following closed-form solutions with initial conditions $f_i(0) = 0$ for $i = 1,2,3,4,6$ and $f_5(t,0) = 0$:
\begin{align*}
f_1(\tau) &=
(\kappa\bar r+\bar\sigma_r^\top\Lambda)\left(\dfrac{1-e^{-\kappa\tau}}{\kappa^2}-\dfrac{\tau}{\kappa}\right)
+ \dfrac{\bar\sigma_r^\top\bar\sigma_r}{2}\left(\dfrac{\tau}{\kappa^2}-\dfrac{3}{2\kappa^3}
+ \dfrac{2e^{-\kappa\tau}}{\kappa^3}
- \dfrac{e^{-2\kappa\tau}}{2\kappa^3}\right),\\[0.5em]
f_2(\tau) &= \dfrac{1}{\kappa} (e^{-\kappa\tau} - 1),\\[0.5em]
f_3(\tau) &= \begin{aligned}[t]
&\left( \Lambda^\top\Lambda - 2\bar r - \dfrac{4\bar\sigma_r^\top\Lambda}{\kappa}
+ \dfrac{2\bar\sigma_r^\top\bar\sigma_r}{\kappa^2} \right)\tau\\
&\quad + \dfrac{2}{\kappa}\left( \bar r + \dfrac{2\bar\sigma_r^\top\Lambda}{\kappa}
- \dfrac{2\bar\sigma_r^\top\bar\sigma_r}{\kappa^2} \right)(1-e^{-\kappa\tau})
- \dfrac{\bar\sigma_r^\top\bar\sigma_r}{\kappa^3}(e^{-2\kappa\tau}-1),
\end{aligned}\\[0.5em]
f_4(\tau) &= \dfrac{2}{\kappa}(e^{-\kappa\tau}-1),\\[0.5em]
f_5(t,\tau) &=
\big(\xi_1-\sigma_3^\top\Lambda\big)\tau
- \dfrac{\Lambda^\top\bar\sigma_r-\sigma_3^\top\bar\sigma_r}{\kappa}\tau
- \bar r\tau
+ \dfrac{\bar\sigma_r^\top\bar\sigma_r}{2\kappa^2}\tau
- \dfrac{3\,\bar\sigma_r^\top\bar\sigma_r}{4\kappa^3}\\
&\quad
+ \dfrac{\bar\sigma_r^\top\bar\sigma_r}{\kappa^3}e^{-\kappa\tau}
- \dfrac{\bar\sigma_r^\top\bar\sigma_r}{4\kappa^3}e^{-2\kappa\tau}
+ \dfrac{\Lambda^\top\bar\sigma_r-\sigma_3^\top\bar\sigma_r}{\kappa^2}(1-e^{-\kappa\tau})\\
&\quad
+ \dfrac{\bar r}{\kappa}(1-e^{-\kappa\tau})
+ \xi_0\!\left[(x_0+t)\tau+\dfrac{\tau^2}{2}\right],\\[0.5em]
f_6(\tau) &= \dfrac{1}{\kappa}(e^{-\kappa\tau}-1).
\end{align*}
\end{lemma}

\begin{proof}
    See Appendix~\ref{app:proof-lem-conditional-expectations}.
\end{proof}

Based on Lemma~\ref{lem:conditional-expectations}, we define the auxiliary aggregated functions $A(t, r_t)$, $B(t, r_t)$, and $H(t, r_t)$ to compactly express the pension member's conditional expectations
\begin{equation}
\label{eq:A-B-H-def}
\begin{aligned}
A(t, r_t) &=
\int_t^T {}_{u-t}p_{x+t} \mu_{x+u}
e^{f_1(u-t) + f_2(u-t)r_t} \, du
+ {}_{T-t}p_{x+t}
e^{f_1(T-t) + f_2(T-t)r_t}, \\
B(t, r_t) &=
\int_t^T {}_{u-t}p_{x+t} \mu_{x+u}
e^{f_3(u-t) + f_4(u-t)r_t} \, du
+ {}_{T-t}p_{x+t}
e^{f_3(T-t) + f_4(T-t)r_t}, \\
H(t, r_t) &=
\int_t^T {}_{u-t}p_{x+t}
e^{f_5(t, u-t) + f_6(u-t)r_t} \, du.
\end{aligned}
\end{equation}

\begin{assumption}[Non-degeneracy]
\label{assump:nondegeneracy}
Assume that
\[
\Sigma_F \, \Sigma_F^\top \succ 0,
\]
and that $A(0, r_0)$, $B(0, r_0)$, and $H(0, r_0)$ are finite with
\[
B(0, r_0) - A(0, r_0)^2 > 0.
\]
\end{assumption}

Propositions~\ref{prop:static-budget-necessity} and~\ref{prop:static-budget-sufficiency} allow for the exclusion of the dynamic control processes. The dynamic problem is thereby reduced to minimizing the objective \eqref{eq:dynamic-objective} over $W_T$ and $M_t$, subject to the expectation constraint \eqref{eq:dynamic-constraint} and the static budget constraint \eqref{eq:budget-necessity}. By constructing the Lagrangian, we obtain the optimal static solutions.

\begin{proposition}
\label{prop:static-solution}
Under Assumption~\ref{assump:nondegeneracy}, for the static optimization problem, the optimal terminal wealth $W_T^\ast$ and the optimal bequest value $M_t^*$ are given by:
\begin{align}
\label{eq:WT-star}
W_T^* &= \frac{\lambda_1^*}{2} + \frac{\lambda_2^*}{2} \phi_T,\\
\label{eq:Mt-star}
M_t^* &= \frac{\lambda_1^*}{2} + \frac{\lambda_2^*}{2} \phi_t,
\end{align}
where the optimal Lagrange multipliers $\lambda_1^*$ and $\lambda_2^*$ are deterministic constants determined by:
\begin{align}
\label{eq:lambda1-star}
\lambda_1^* &=
2 \cdot \frac{
\mathcal{K} \cdot B(0, r_0)
- \big(w_0 + C_0 H(0, r_0)\big) \cdot A(0, r_0)
}{
B(0, r_0) - A(0, r_0)^2
},\\
\label{eq:lambda2-star}
\lambda_2^* &=
2 \cdot \frac{
\big(w_0 + C_0 H(0, r_0)\big)
- \mathcal{K} \cdot A(0, r_0)
}{
B(0, r_0) - A(0, r_0)^2
}.
\end{align}
\end{proposition}

\begin{proof}
    See Appendix~\ref{app:proof-prop-static-solution}.
\end{proof}

\subsection{Optimal dynamic strategies and the efficient frontier}
\label{subsec:optimal-strategies}

By applying It\^o's lemma to the optimal wealth process and using the auxiliary functions defined above, we obtain the explicit solutions for the dynamic controls.

\begin{theorem}
\label{thm:optimal-strategies}
For the continuous-time optimization problem, the optimal wealth process $W_t^*$ and the optimal insurance strategy $I_t^*$ for $0 \le t < T$, together with the optimal investment strategy $\pi_t^*$ for $0 \le t < T$ whenever $\Sigma_F\Sigma_F^\top$ is invertible, are explicitly given by
\begin{align}
\label{eq:Wt-star-dynamic}
W_t^* &=
\frac{\lambda_1^*}{2} A(t, r_t)
+ \frac{\lambda_2^*}{2} \phi_t B(t, r_t)
- C_t H(t, r_t),\\
\label{eq:pi-star}
\pi_t^* &=
(\Sigma_F \, \Sigma_F^\top)^{-1} \Sigma_F
\Bigg[
- \bar{\sigma}_r \left(
\frac{\lambda_1^*}{2} \frac{\partial A}{\partial r}
+ \frac{\lambda_2^*}{2} \phi_t \frac{\partial B}{\partial r}
- C_t \frac{\partial H}{\partial r}
\right)
- \frac{\lambda_2^*}{2} \phi_t B(t, r_t) \Lambda
- C_t H(t, r_t) \sigma_3
\Bigg],\\
\label{eq:I-star}
I_t^* &=
\mu_{x+t} \left[
\frac{\lambda_1^*}{2} \bigl(1 - A(t, r_t)\bigr)
+ \frac{\lambda_2^*}{2} \phi_t \bigl(1 - B(t, r_t)\bigr)
+ C_t H(t, r_t)
\right],
\end{align}
where the partial derivatives with respect to the short rate $r_t$ are:
\begin{equation*}
\begin{aligned}
\frac{\partial A}{\partial r} &=
\int_t^T {}_{u-t}p_{x+t} \mu_{x+u} f_2(u-t)
e^{f_1(u-t) + f_2(u-t)r_t} \, du
+ {}_{T-t}p_{x+t} f_2(T-t)
e^{f_1(T-t) + f_2(T-t)r_t}, \\
\frac{\partial B}{\partial r} &=
\int_t^T {}_{u-t}p_{x+t} \mu_{x+u} f_4(u-t)
e^{f_3(u-t) + f_4(u-t)r_t} \, du
+ {}_{T-t}p_{x+t} f_4(T-t)
e^{f_3(T-t) + f_4(T-t)r_t}, \\
\frac{\partial H}{\partial r} &=
\int_t^T {}_{u-t}p_{x+t} f_6(u-t)
e^{f_5(t, u-t) + f_6(u-t)r_t} \, du.
\end{aligned}
\end{equation*}
\end{theorem}

\begin{proof}
    See Appendix~\ref{app:proof-thm-optimal-strategies}.
\end{proof}

\begin{theorem}
\label{thm:efficient-frontier}
The mean-variance efficient frontier on the efficient branch
\[
\mathcal{K} \ge \frac{W_0^{\textnormal{eff}}}{A(0, r_0)}
\]
is mathematically expressed as
\begin{equation}
\label{eq:efficient-frontier}
\mathcal{K}
= \frac{W_0^{\textnormal{eff}}}{A(0, r_0)}
+ \frac{\sqrt{B(0, r_0) - A(0, r_0)^2}}{A(0, r_0)} \,
\sigma_{Y_T},
\end{equation}
where $W_0^{\textnormal{eff}} = w_0 + C_0 H(0, r_0)$ defines the agent's effective initial wealth, and $\sigma_{Y_T}$ represents the standard deviation of the mortality-contingent payout.
\end{theorem}

\begin{proof}
    See Appendix~\ref{app:proof-thm-efficient-frontier}.
\end{proof}

\subsection{Verification theorem}
\label{subsec:verification}

This theorem formally proves that the strategy satisfying the first-order conditions indeed minimizes the objective functional within the admissible set.

\begin{theorem}[Verification Theorem]
\label{thm:verification}
Suppose that the candidate optimal strategy $(\pi^*, I^*)$ is given by the equations in Theorem~\ref{thm:optimal-strategies}, and let $(W_T^*, M_t^*)$ be the corresponding optimal terminal wealth and bequest value processes. Assume that $(\pi^*, I^*) \in \mathcal{A}(0, T)$. Then, for any other admissible strategy $(\pi, I) \in \mathcal{A}(0, T)$ generating processes $W_T$ and $M_t$ that satisfy both the expectation constraint \eqref{eq:dynamic-constraint} and the static budget constraint \eqref{eq:budget-necessity}, the following inequality holds:
\begin{equation*}
\mathbb{E}[(Y_T^*)^2] \le \mathbb{E}[(Y_T)^2].
\end{equation*}
Consequently, $(\pi^*, I^*)$ is the globally optimal strategy for the mean-variance optimization problem.
\end{theorem}

\begin{proof}
    See Appendix~\ref{app:proof-thm-verification}.
\end{proof}

\section{Numerical analysis}
\label{sec:numerical-analysis}

\subsection{Model calibration}
\label{subsec:model-calibration}

This section provides a numerical illustration of the optimal strategies and the implied efficient frontier. Baseline parameters for the financial market and the contribution dynamics are not separately estimated in this paper; they are adopted from standard values in the literature, primarily following \citet{munk2010dynamic}.

We consider a pension plan member who enters the plan at the age of $x_0 = 22$ and retires at age 67. The accumulation period of this individual is therefore $T = 45$ years. The initial wealth and initial contribution rate are both normalized to $w_0 = 0$ and $C_0 = 1$, respectively. 

For the financial market, the interest-rate parameters are set to $\kappa=0.50$, $\bar{r}=0.02$, and $\sigma_r=0.02$, with $r_0=\bar r=0.02$. The rolling-bond strategy on a zero-coupon bond is implemented with fixed time-to-maturity $\tau_B=10$ years. Stock-market and contribution volatilities are both set to 20\% ($\sigma_S = 0.20, \sigma_C = 0.20$). To maintain market completeness with respect to background risk, we assume contribution risk is fully spanned by the equity index, i.e., $\rho_{CB} = 0$ and $\hat\rho_{CS} = 1$. The market price of interest-rate risk is set to $\lambda_r = 0$, and the market price of equity risk to $\lambda_S = 0.2$. The growth rate factors of the contribution process are set to $\xi_0 = 0$ and $\xi_1 = 0.02$.

We further set the target expected payout to $\mathcal{K} = 200$. Assume a 10\% contribution rate and a 20-year level retirement income, this target is equivalent to about 49.3\% replacement ratio under the baseline interest rate $\bar r=0.02$. This value is consistent with the replacement rate target reported by OECD, in which the gross replacement rate of full career average earners is about 52\% \citep{oecd2025pag}. The alternative targets $\mathcal{K}=100$ and $\mathcal{K}=300$ correspond to replacement ratios of 24.7\% and 74.0\%, respectively.

Table~\ref{financial_market_estimate_table} reports the baseline parameters used throughout the numerical analysis.


\begin{table}[!htbp]
	\centering
	\begin{threeparttable}
		\caption{Baseline economic and financial parameters.}
		\label{financial_market_estimate_table}
		\begin{tabular}{c r c r c r}
			\toprule
			Parameter     & Value & Parameter     & Value  & Parameter     & Value \\
			\midrule
			\multicolumn{6}{l}{\textit{Interest rate process}}\\
			$\kappa$      & 0.50  & $\bar{r}$     & 0.02   & $\sigma_{r}$  & 0.02  \\
			\multicolumn{6}{l}{\textit{Financial and contribution risks}}\\
			$\sigma_{S}$  & 0.20  & $\sigma_{C}$  & 0.20   & $\rho_{SB}$   & 0.00  \\  
				$\rho_{CB}$   & 0.00  & $\hat\rho_{CS}$ & 1.00 & $\tau_B$      & 10.00 \\
			\multicolumn{6}{l}{\textit{Expected contribution growth and market prices of risk}}\\
			$\xi_0$       & 0.00  & $\xi_1$       & 0.02   &               &       \\
			$\lambda_r$   & 0.00  & $\lambda_S$   & 0.20   &               &       \\
			\multicolumn{6}{l}{\textit{State variables}}\\
			$x_0$         & 22    & $T$           & 45     & $\mathcal{K}$ & 200.0 \\    
			$w_0$         & 0.00  & $C_0$         & 1.00   & $r_0$         & 0.02  \\ 
			\bottomrule
		\end{tabular}
		\begin{tablenotes}[para]
			\small
			Parameters are annualized. Contribution risk is assumed to be fully spanned by the traded assets. Continuous-time simulations use an Euler--Maruyama scheme with 400 time steps and 10{,}000 Monte Carlo paths.
		\end{tablenotes}
	\end{threeparttable}
\end{table}
\FloatBarrier

For the individual mortality rate, we use U.S. mortality data from the Mortality Improvement Model (MIM-2021-v4) published by the Society of Actuaries (SOA). Our benchmark is a non-improvement mortality model fixed in the year 2017, and a comparison with the improvement model is in Section \ref{subsec:mortality-improvement}. Following \citet{forfar1988graduation}, we assume that the force of mortality follows a generalized Gompertz--Makeham $GM(s_1,s_2)$ specification: 
\begin{equation*}
\mu_x 
= \sum_{i=1}^{s_1} a_i (x-22)^{i-1}
+ \exp\left\{\sum_{i=s_1+1}^{s_1+s_2} a_i (x-22)^{i-s_1-1}\right\}, \quad 22 \le x \le 67.
\end{equation*}
We estimate candidate GM models under alternative polynomial orders, considering all $(s_1,s_2)$ combinations in a $3\times 4$ grid, and determine the final model according to coefficient significance and the Bayesian information criterion (BIC) calculated as $\mathrm{BIC}=k\log n-2\ell$, where $\ell$ denotes the maximized log likelihood, $k$ is the number of estimated parameters, and $n$ is the sample size. This selection procedure leads to a $GM(2,3)$ specification for the baseline mortality model. The estimation results are presented in Table~\ref{mortality_estimate_table}.


\begin{table}[!htbp]
	\caption{Estimation results for the force of mortality.}
	\label{mortality_estimate_table}
	\centering
	\begin{tabular}{c r r r}
	\toprule
	Model        &$GM(2,3)$                &BIC                         &$258{,}111.64$               \\
	\midrule
	Parameter    &$a_1$                    &$a_2$                     &$a_3$                       \\
	Value        &$1.385293 \times 10^{-3}$ &$3.961326 \times 10^{-5}$ &$-1.355493 \times 10^{1}$   \\
	
    Parameter    &$a_4$                    &$a_5$                     &                            \\
	Value        &$3.637433 \times 10^{-1}$ &$-3.500351 \times 10^{-3}$&                            \\
    \bottomrule		
	\end{tabular}
\end{table}
\FloatBarrier

\subsection{Baseline results}
\label{subsec:baseline-results}

This subsection illustrates the time trends of optimal wealth, insurance, and portfolio decisions. Under the baseline parameters in Tables~\ref{financial_market_estimate_table} and \ref{mortality_estimate_table}, we simulate the model via Euler--Maruyama discretization (400 time steps) with 10{,}000 Monte Carlo paths. 

Figure~\ref{fig:3D_efficient_frontier} shows how the efficient frontier changes with the market prices of risk $\lambda_r$ and $\lambda_S$, together with the labor-income growth parameter $\xi_0$. When the parameters are fixed, the volatility $\sigma_{Y_T}$ is linear in the expected payoff $\mathbb{E}[Y_T^*]$, consistent with equation~\eqref{eq:efficient-frontier}. Taking the Sharpe ratios $\lambda_r$ and $\lambda_S$ into account, a higher market price of risk generally reduces the risk required to achieve the same target, although the effect of $\lambda_r$ is non-monotonic over the range considered. As $\xi_0$ increases, the present value of future contributions rises, and less risk is needed to achieve the same level of expected terminal payoff.

\begin{figure}[!htbp]
    \centering
    \begin{subfigure}[b]{0.32\textwidth}
        \centering
        \includegraphics[width=\textwidth]{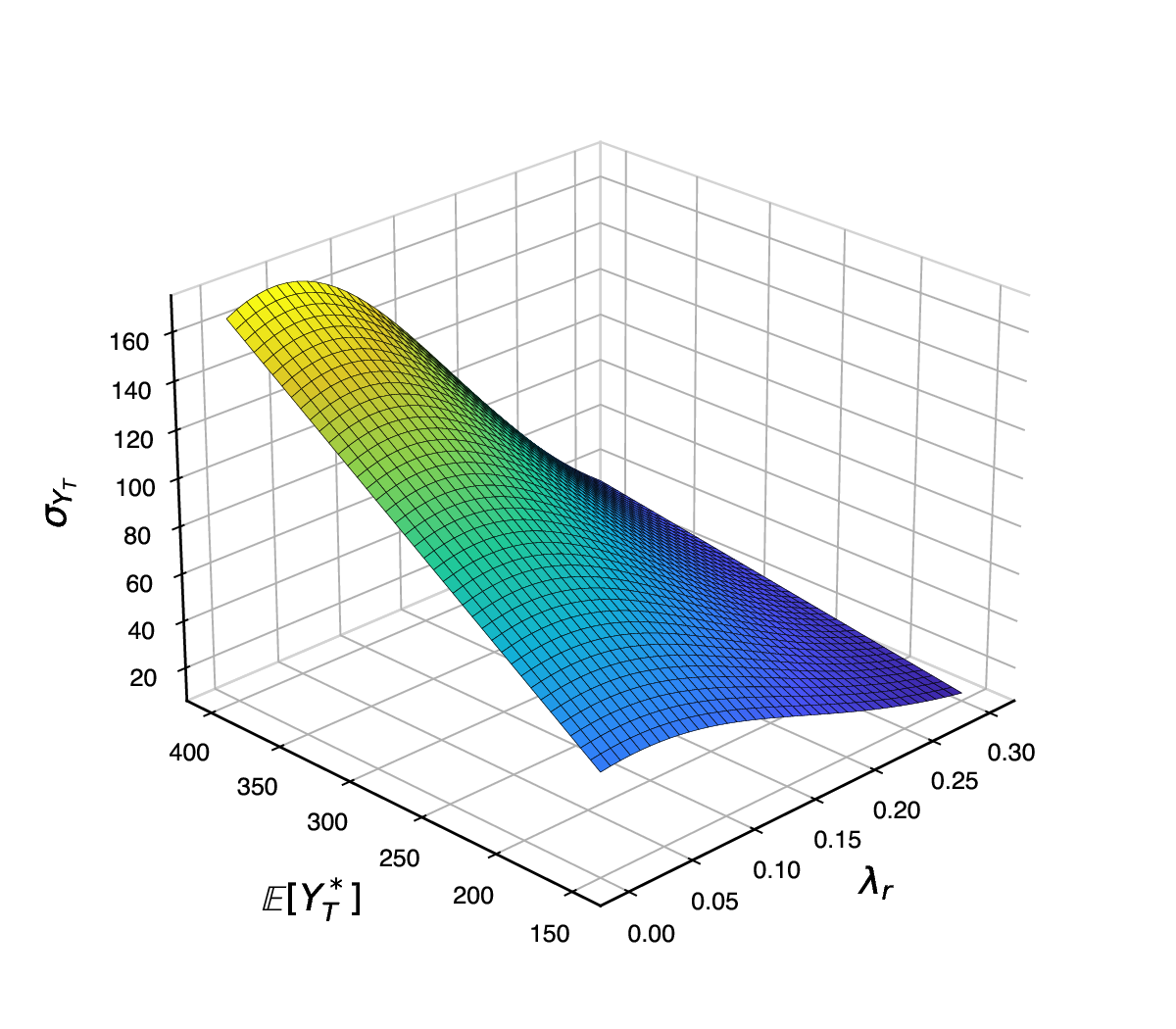}
        \caption{Interest-rate risk price $\lambda_r$}
        \label{fig:3D_efficient_frontier_lambda_r}
    \end{subfigure}
    \hfill
    \begin{subfigure}[b]{0.32\textwidth}
        \centering
        \includegraphics[width=\textwidth]{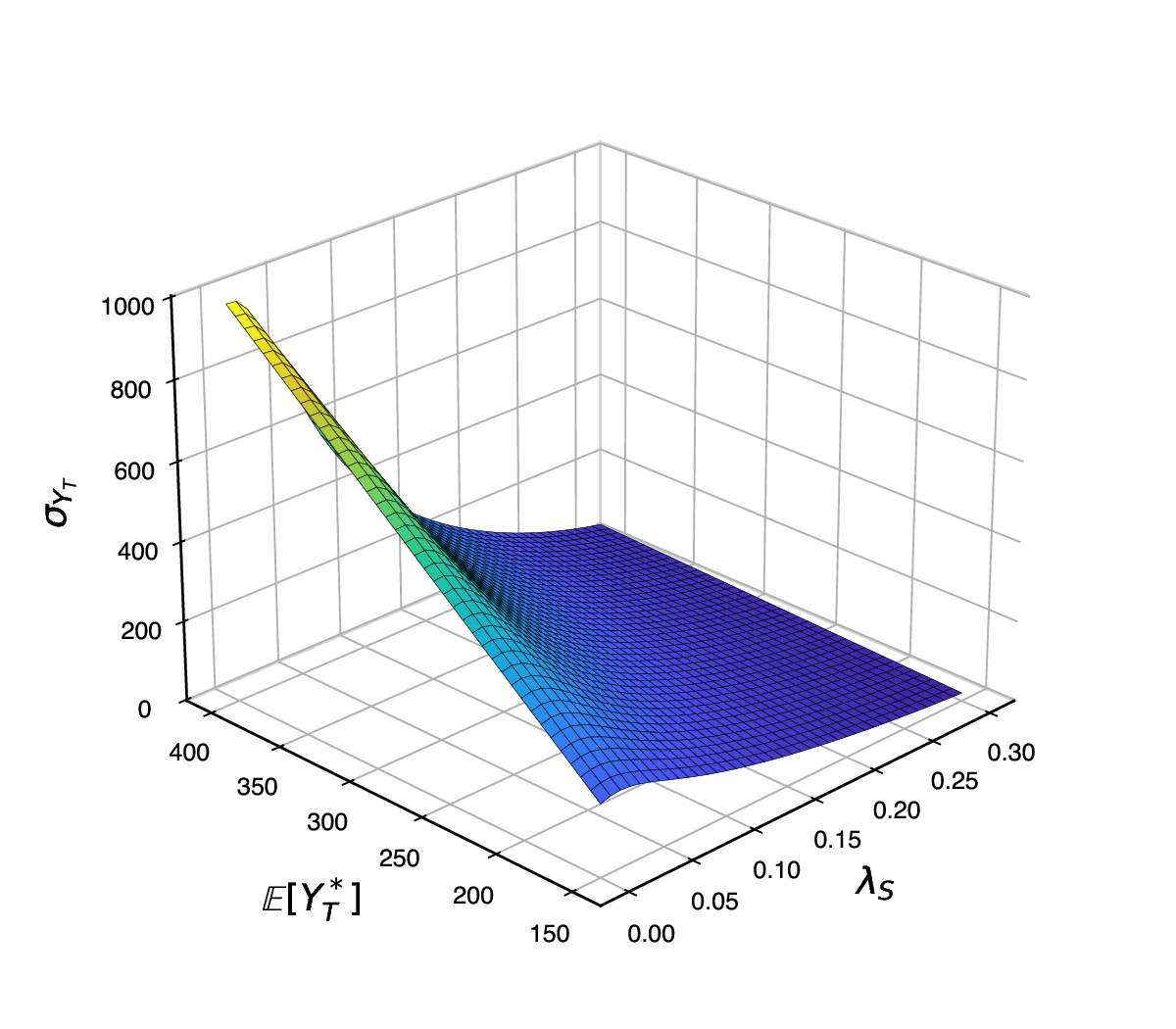}
        \caption{Equity risk price $\lambda_S$}
        \label{fig:3D_efficient_frontier_lambda_S}
    \end{subfigure}
    \hfill
    \begin{subfigure}[b]{0.32\textwidth}
        \centering
        \includegraphics[width=\textwidth]{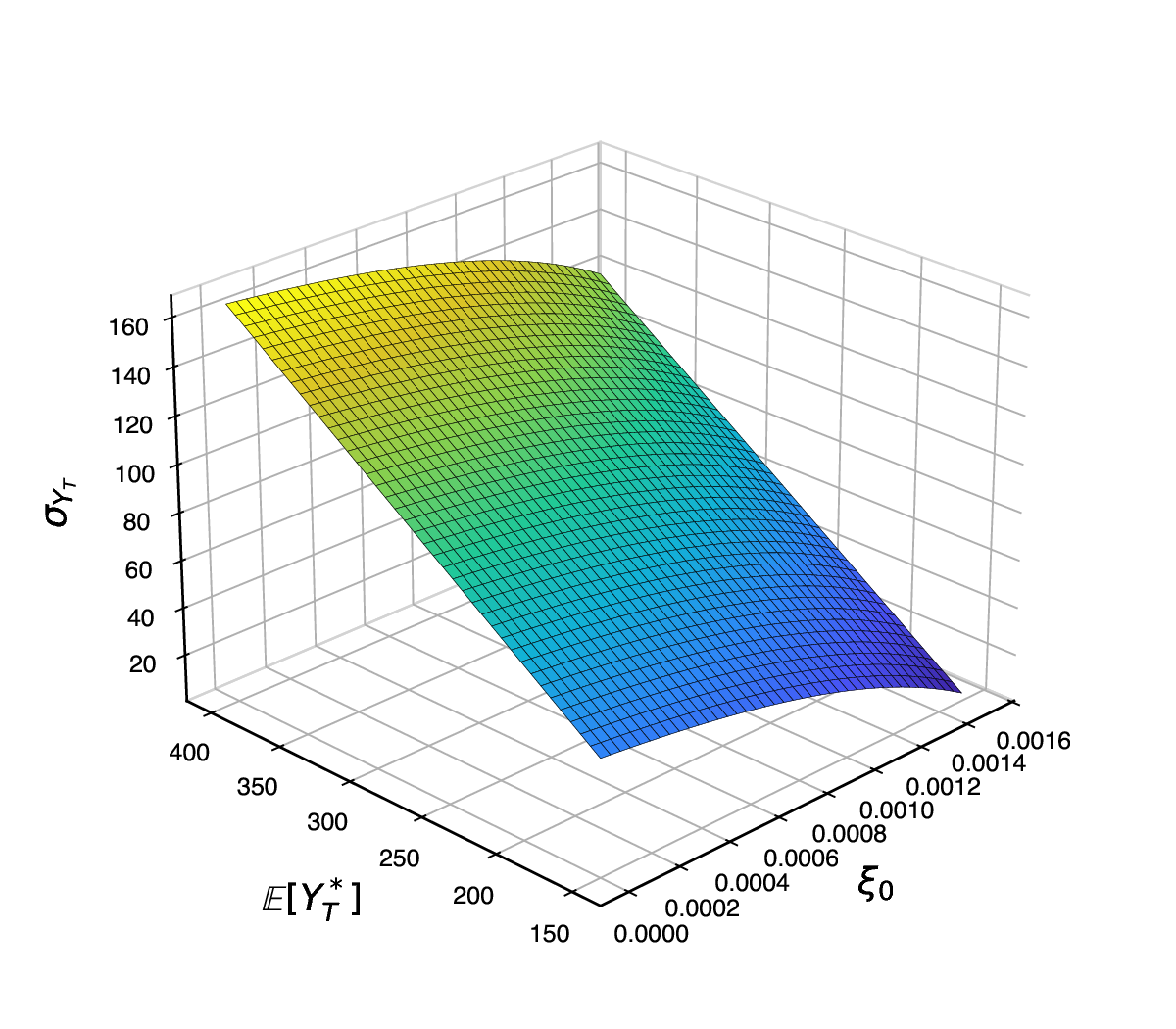}
        \caption{Contribution growth $\xi_0$}
        \label{fig:3D_efficient_frontier_xi_0}
    \end{subfigure}
    \caption{Three-dimensional efficient frontier surfaces with respect to the market prices of risk and the contribution growth parameter.}
    \label{fig:3D_efficient_frontier}
\end{figure}

Figure~\ref{fig:analysis3_prop} reports the expected optimal bond and stock allocation together with the optimal insurance premium. Both dollar amounts and ratios to wealth are reported. As shown in both Figures~\ref{fig:analysis3_piB} and \ref{fig:analysis3_prop_piB}, the pension member tends to short bonds at an early age to purchase the stock index. The bond position gradually shifts from negative to positive as the member deleverages and reduces risk exposure. Near retirement, it declines toward zero because the future contribution stream has almost been exhausted, leaving little remaining interest-rate risk to hedge. Stock position, on the other hand, moves in the almost opposite direction. Figures~\ref{fig:analysis3_piS} and \ref{fig:analysis3_prop_piS} demonstrate such a de-leveraging behavior. Insurance demand shown in Figures~\ref{fig:analysis3_I} and ~\ref{fig:analysis3_prop_I} is consistent with the classical result of a hump-shaped curve. At an early age, the mortality risk is very limited and can almost be ignored. People start paying attention to it in midlife, and demand for life insurance increases. As retirement approaches, premature death becomes less of a concern, pushing $I_t^*$ towards zero.

\begin{figure}[!htbp]
    \centering
    \begin{subfigure}[b]{0.32\textwidth}
        \centering
        \includegraphics[width=\textwidth]{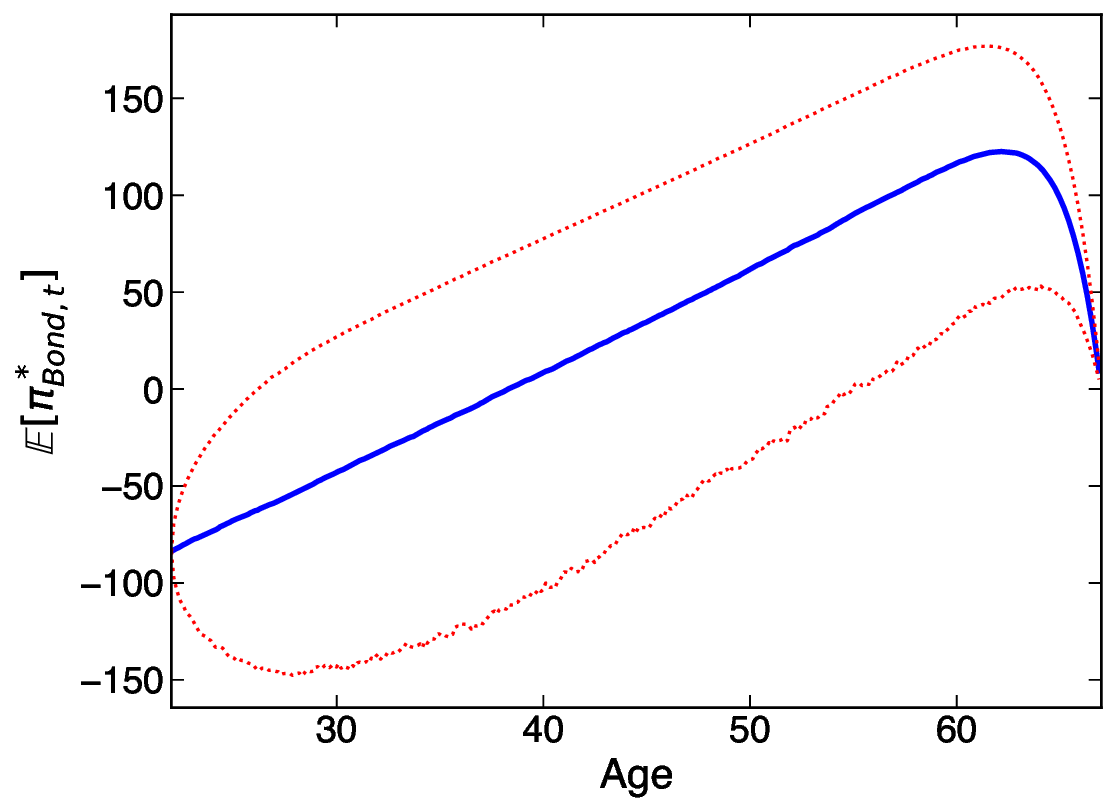}
        \caption{Bond Allocation ($\pi_{Bond}^*$)}
        \label{fig:analysis3_piB}
    \end{subfigure}
    \hfill
    \begin{subfigure}[b]{0.32\textwidth}
        \centering
        \includegraphics[width=\textwidth]{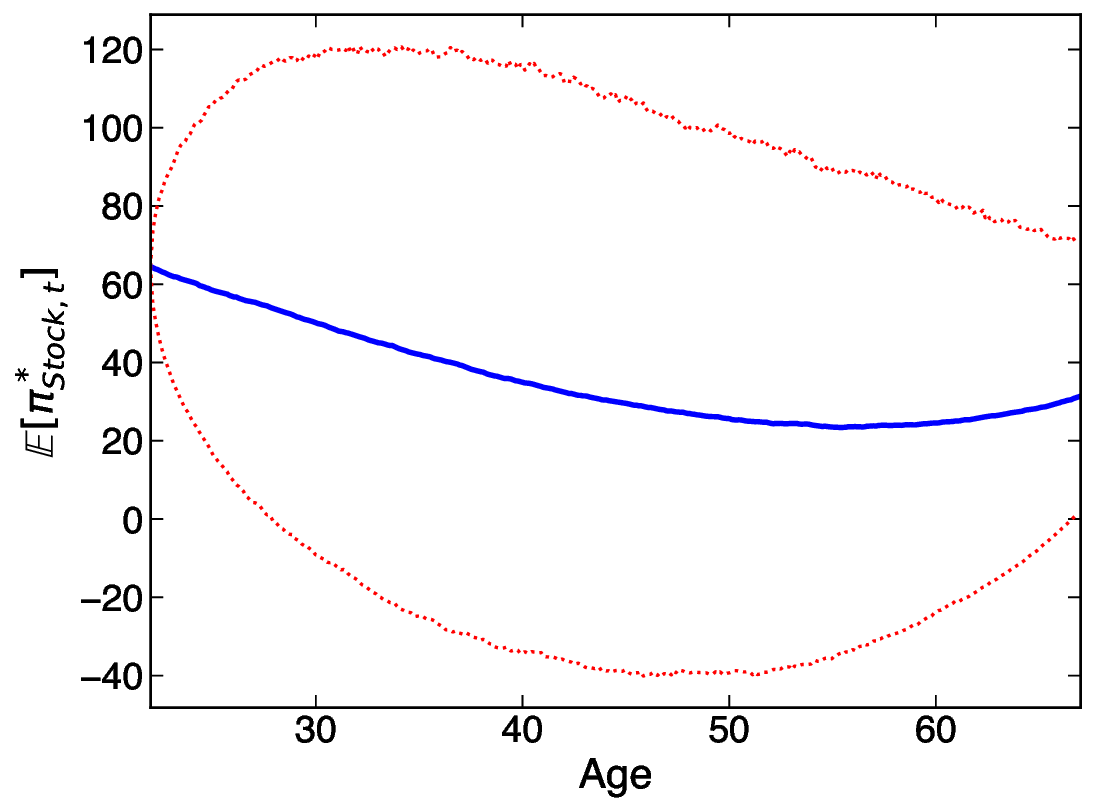}
        \caption{Stock Allocation ($\pi_{Stock}^*$)}
        \label{fig:analysis3_piS}
    \end{subfigure}
    \hfill
    \begin{subfigure}[b]{0.32\textwidth}
        \centering
        \includegraphics[width=\textwidth]{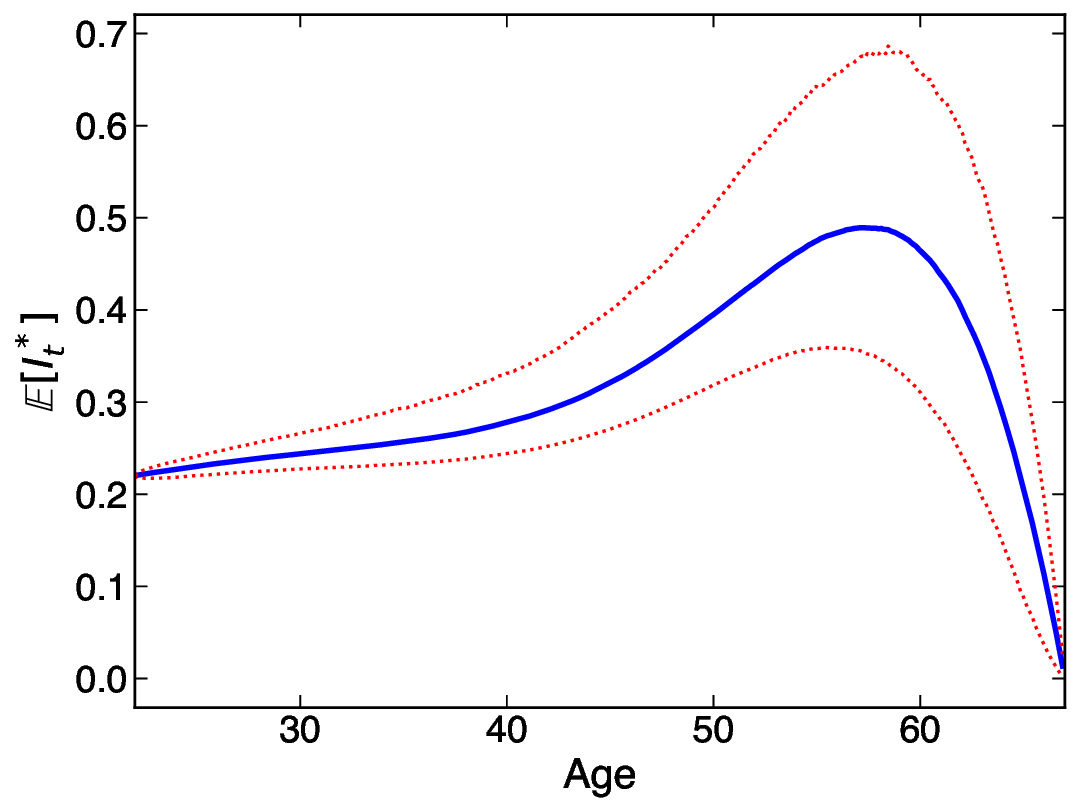}
        \caption{Insurance Premium ($I_t^*$)}
        \label{fig:analysis3_I}
    \end{subfigure}
    
    \vspace{0.6em}
    
    \begin{subfigure}[b]{0.32\textwidth}
        \centering
        \includegraphics[width=\textwidth]{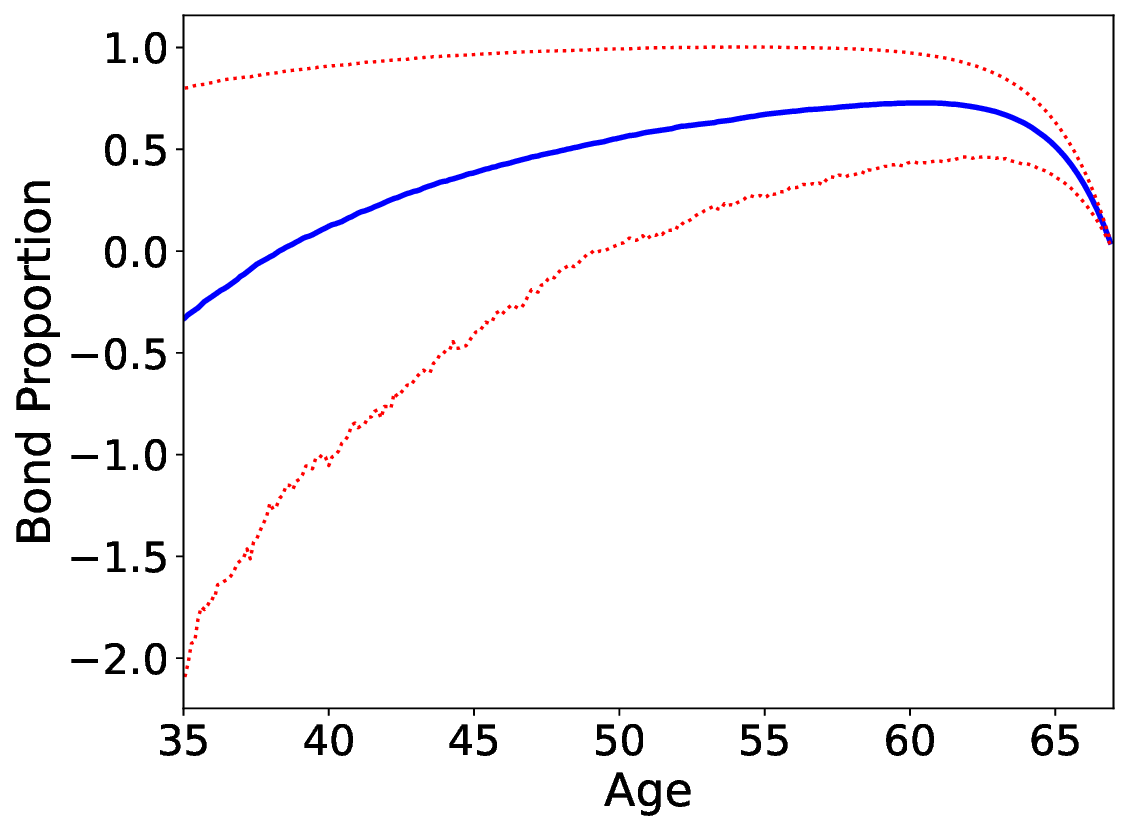}
        \caption{Bond Proportion}
        \label{fig:analysis3_prop_piB}
    \end{subfigure}
    \hfill
    \begin{subfigure}[b]{0.32\textwidth}
        \centering
        \includegraphics[width=\textwidth]{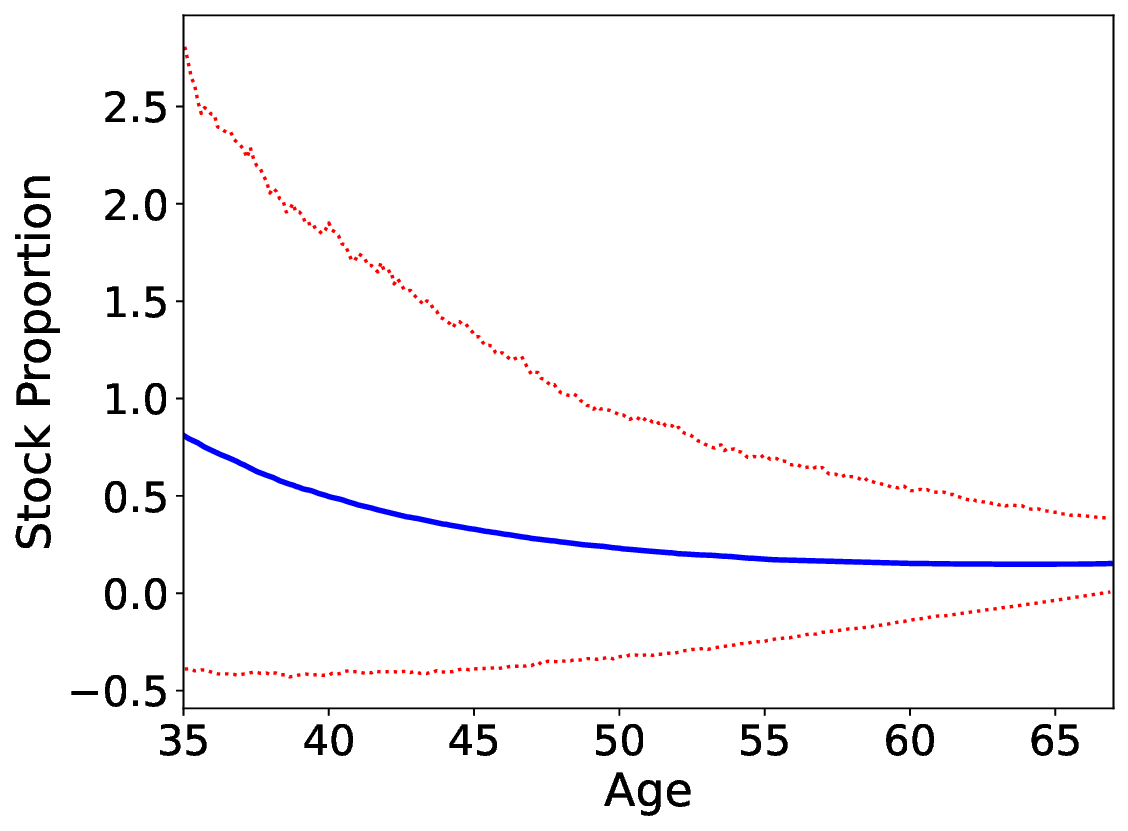}
        \caption{Stock Proportion}
        \label{fig:analysis3_prop_piS}
    \end{subfigure}
    \hfill
    \begin{subfigure}[b]{0.32\textwidth}
        \centering
        \includegraphics[width=\textwidth]{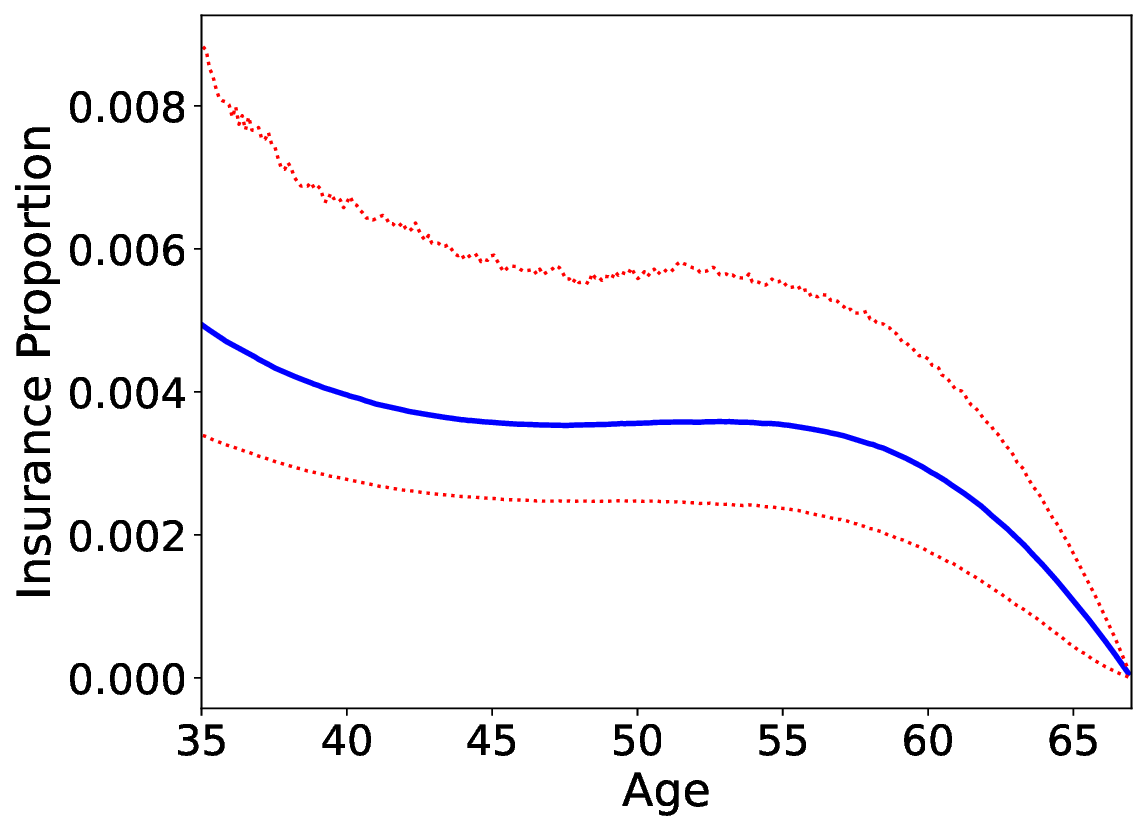}
        \caption{Insurance Proportion}
        \label{fig:analysis3_prop_I}
    \end{subfigure}
    \caption{Detailed simulation results for the amount and proportion.}
    \label{fig:analysis3_prop}
\end{figure}
\FloatBarrier

\subsection{Impact of mortality improvement}
\label{subsec:mortality-improvement}

The Gompertz-Makeham model introduced in Section \ref{subsec:model-calibration} is relatively simple and ignores the fact that mortality rates may change over time. Increased life expectancy and longevity can flatten the force of mortality curve and significantly affect the optimal investment and insurance decisions. To incorporate this trend, we adopt the Mortality Improvement Model (MIM-2021-v4) published by the Society of Actuaries (SOA) for calibration. We consider a male participant of age 22 entering the plan at calendar year 2017. This participant will then retire at age 67 in 2062 according to our setting. Two scenarios are then constructed: 
\begin{enumerate}
    \item \textbf{Static model}: Gompertz-Makeham model parameters are estimated using the mortality rates of the single base year 2017. 
    \item \textbf{Projected model}: Mortality rates are extracted diagonally across the projection table from age 22 in 2017 to age 67 in 2062.
\end{enumerate}

Figure~\ref{fig:analysis4_mu} compares the natural logarithm of force of mortality estimated under both scenarios. At an early stage, the values are similar. However, in the 50s and 60s (calendar years 2040s--2050s), the projected model delivers substantially lower mortality, which is consistent with the anticipated longevity trends.

\begin{figure}[!htbp]
    \centering
    \begin{subfigure}[b]{0.48\textwidth}
        \centering
        \includegraphics[width=\textwidth]{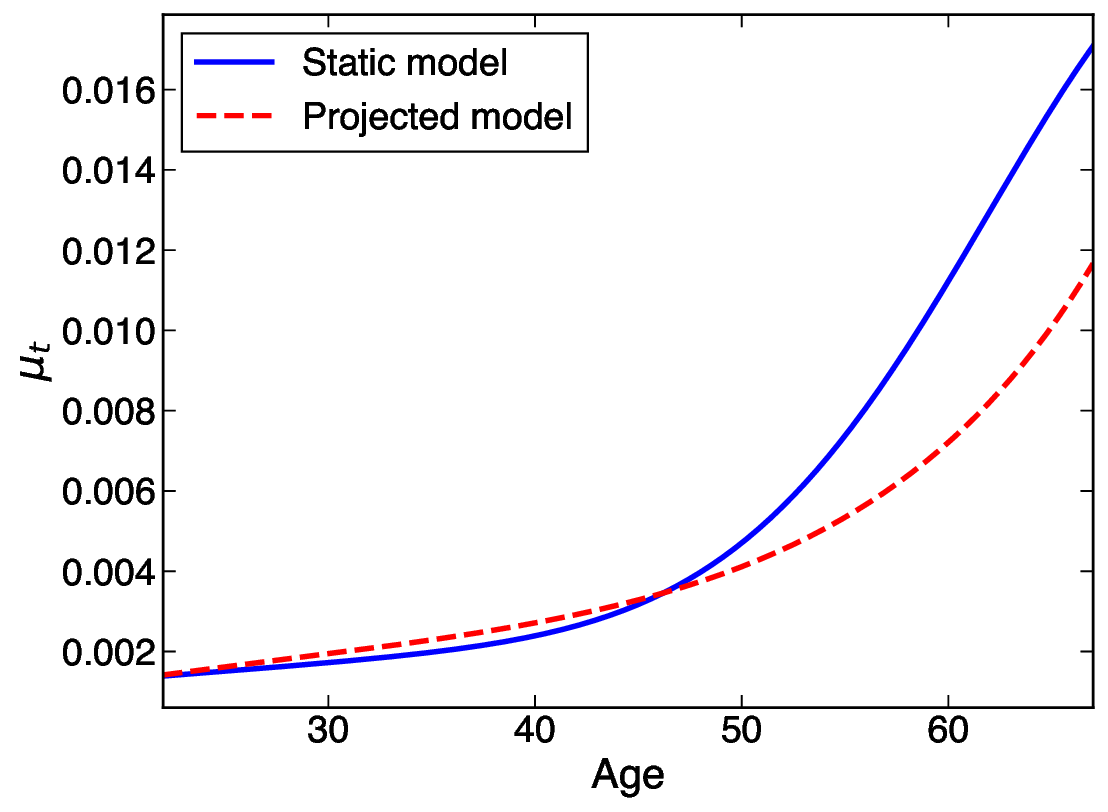}
        \caption{Mortality rates $\mu_t$}
    \end{subfigure}
    \hfill
    \begin{subfigure}[b]{0.48\textwidth}
        \centering
        \includegraphics[width=\textwidth]{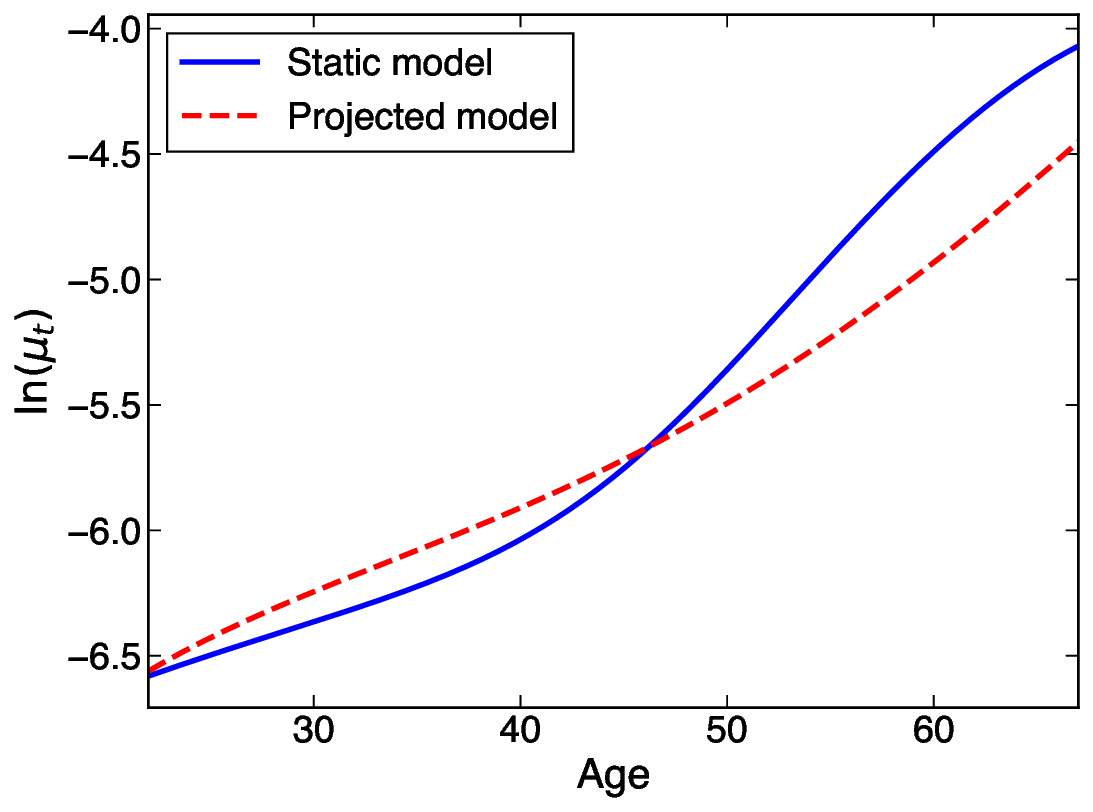}
        \caption{Natural log mortality rates $\ln(\mu_t)$}
    \end{subfigure}
    \caption{Comparison of mortality rates and natural log mortality rates between the static and projected models.}
    \label{fig:analysis4_mu}
\end{figure}
\FloatBarrier

The corresponding optimal strategies are summarized in Figure~\ref{fig:analysis4}. Figures \ref{fig:analysis4_piB} and \ref{fig:analysis4_piS} plot the difference in investment strategy (strategy in the projected model minus that in the static model). We find that mortality improvements overall make pension members more risk-seeking in their investments. Specifically, the pension member tends to short more bonds and purchase more stocks at an early age, buy more bonds and stocks at a middle age, and buy fewer bonds and more stocks as they approach retirement. However, as shown in Figures~\ref{fig:analysis4_prop_piB} and \ref{fig:analysis4_prop_piS}, this investment change is relatively small and doesn't change the proportions allocated between bond and stock. In contrast, the mortality improvement significantly changes the life insurance strategy both in amount and proportion (see Figures \ref{fig:analysis4_I} and \ref{fig:analysis4_prop_I}). In particular, a longer life expectancy leads to a higher expected present value of future income. This pushes the protection needs toward early adulthood, and because of this early shift, mid-to-late-life insurance allocation drops.

\begin{figure}[!htbp]
    \centering
    \begin{subfigure}[b]{0.32\textwidth}
        \centering
        \includegraphics[width=\textwidth]{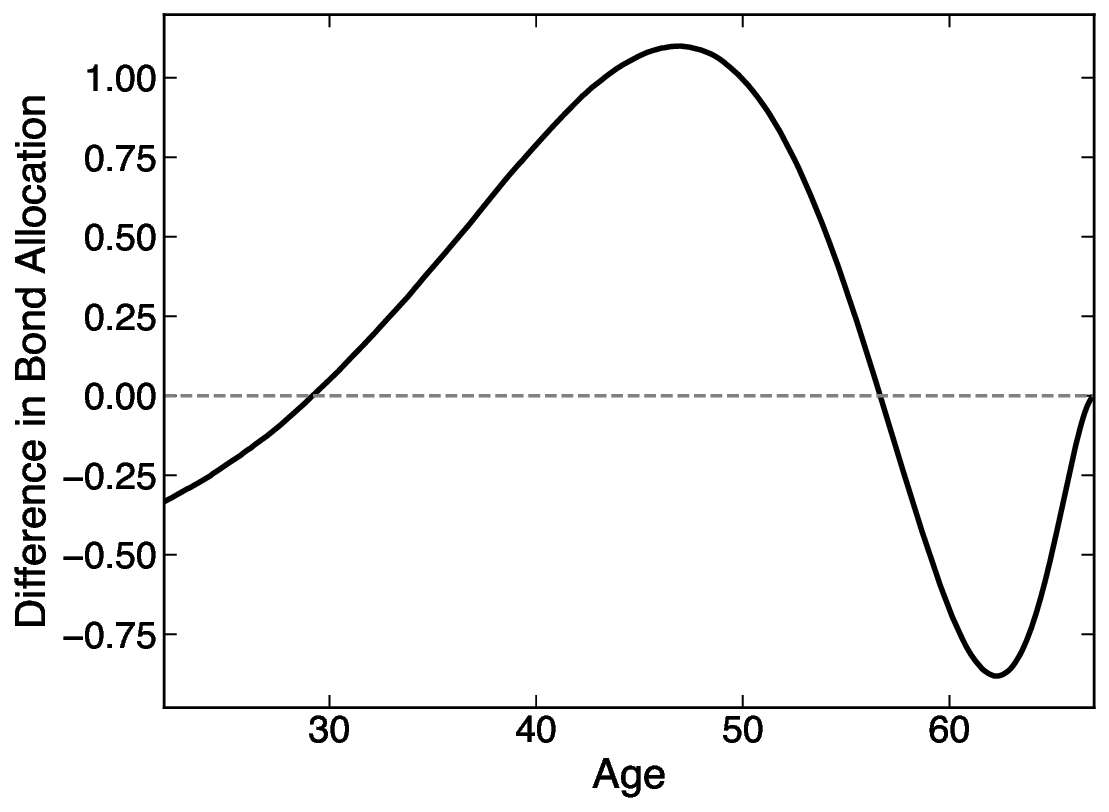}
        \caption{Difference in Bond Allocation}
        \label{fig:analysis4_piB}
    \end{subfigure}
    \hfill
    \begin{subfigure}[b]{0.32\textwidth}
        \centering
        \includegraphics[width=\textwidth]{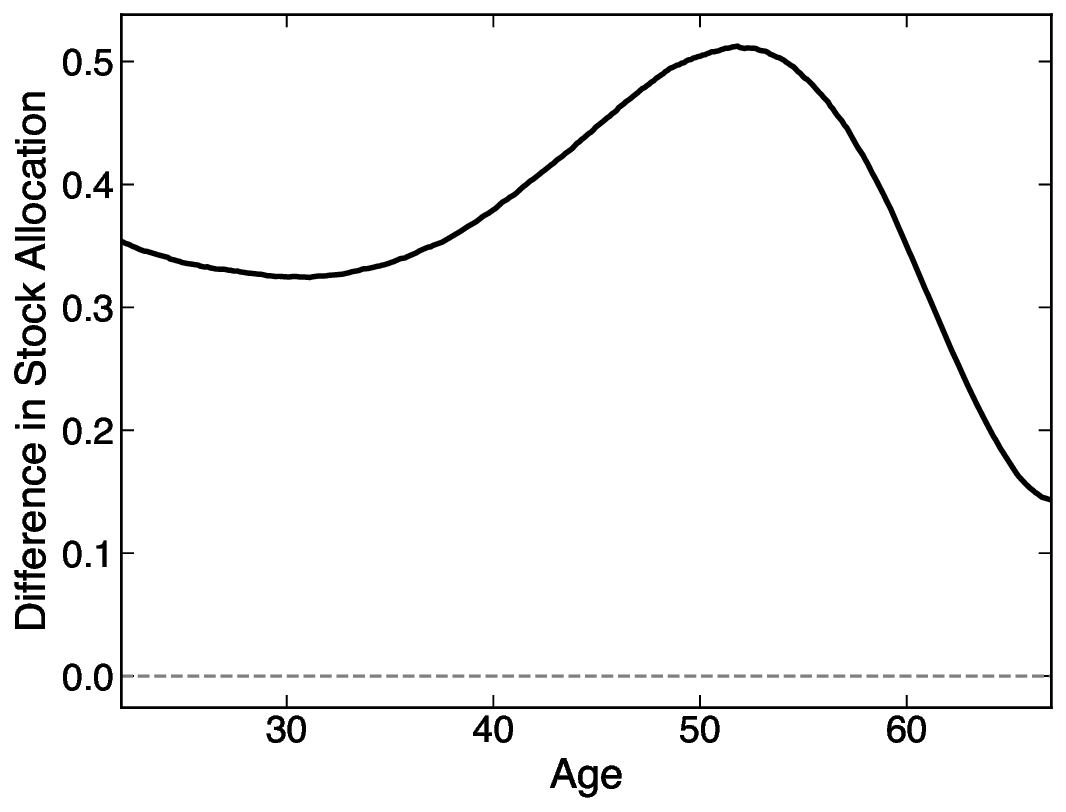}
        \caption{Difference in Stock Allocation}
        \label{fig:analysis4_piS}
    \end{subfigure}
    \hfill
    \begin{subfigure}[b]{0.32\textwidth}
        \centering
        \includegraphics[width=\textwidth]{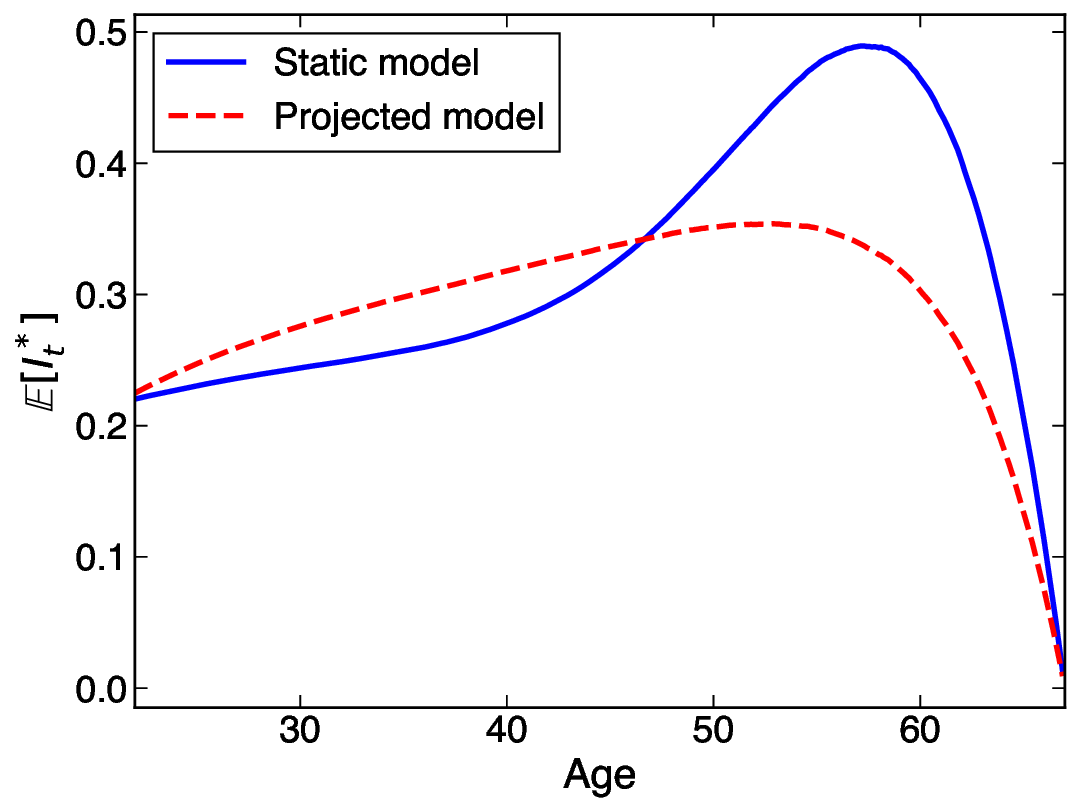}
        \caption{Insurance Premium}
        \label{fig:analysis4_I}
    \end{subfigure}
    
    \vspace{0.6em}
    
    \begin{subfigure}[b]{0.32\textwidth}
        \centering
        \includegraphics[width=\textwidth]{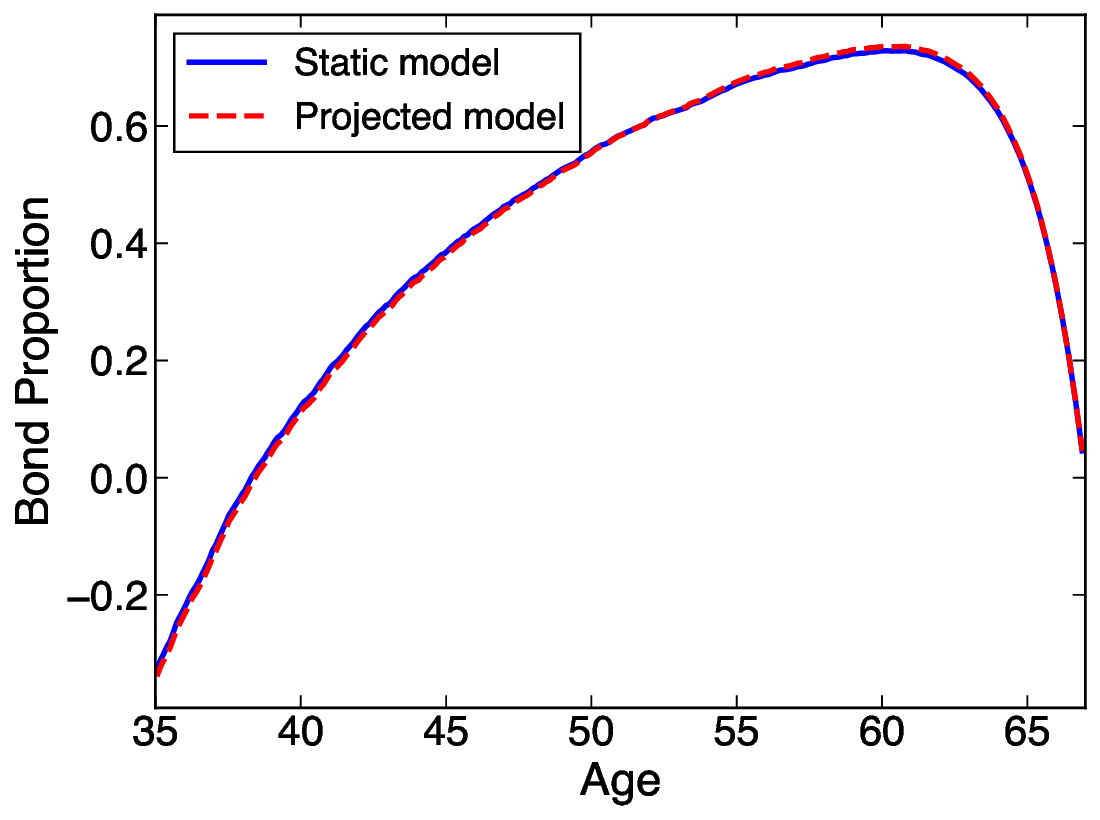}
        \caption{Bond Proportion}
        \label{fig:analysis4_prop_piB}
    \end{subfigure}
    \hfill
    \begin{subfigure}[b]{0.32\textwidth}
        \centering
        \includegraphics[width=\textwidth]{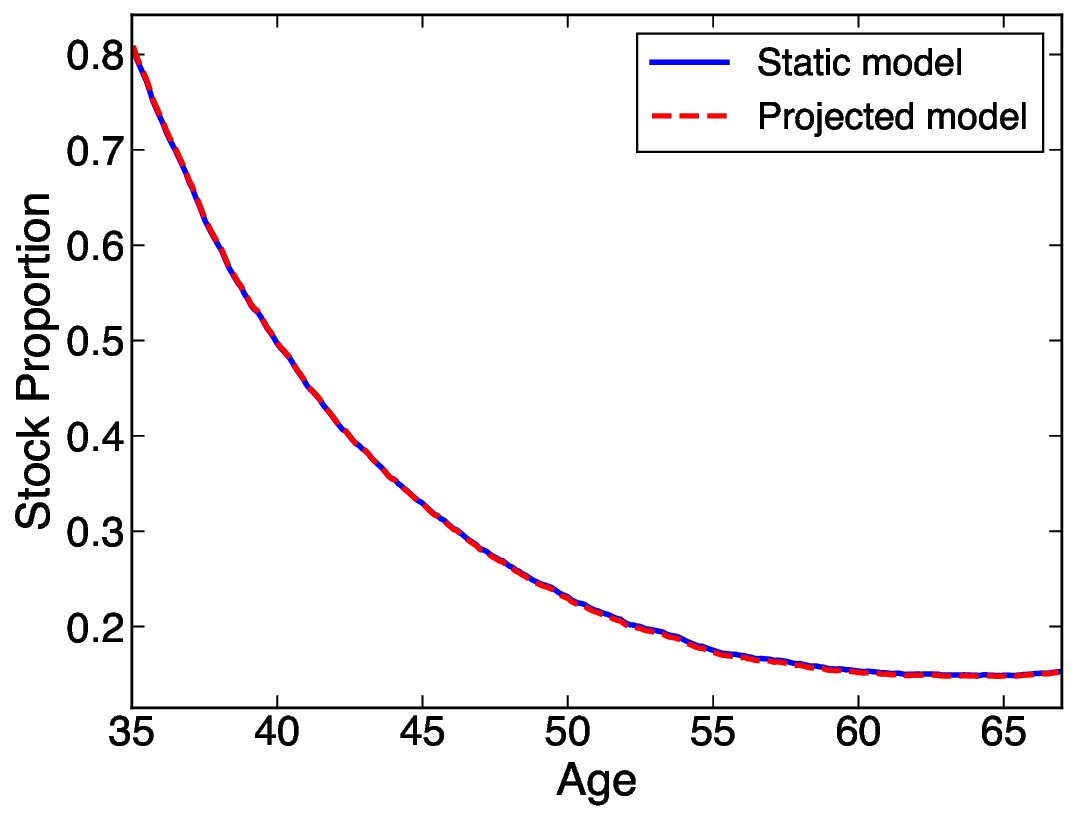}
        \caption{Stock Proportion}
        \label{fig:analysis4_prop_piS}
    \end{subfigure}
    \hfill
    \begin{subfigure}[b]{0.32\textwidth}
        \centering
        \includegraphics[width=\textwidth]{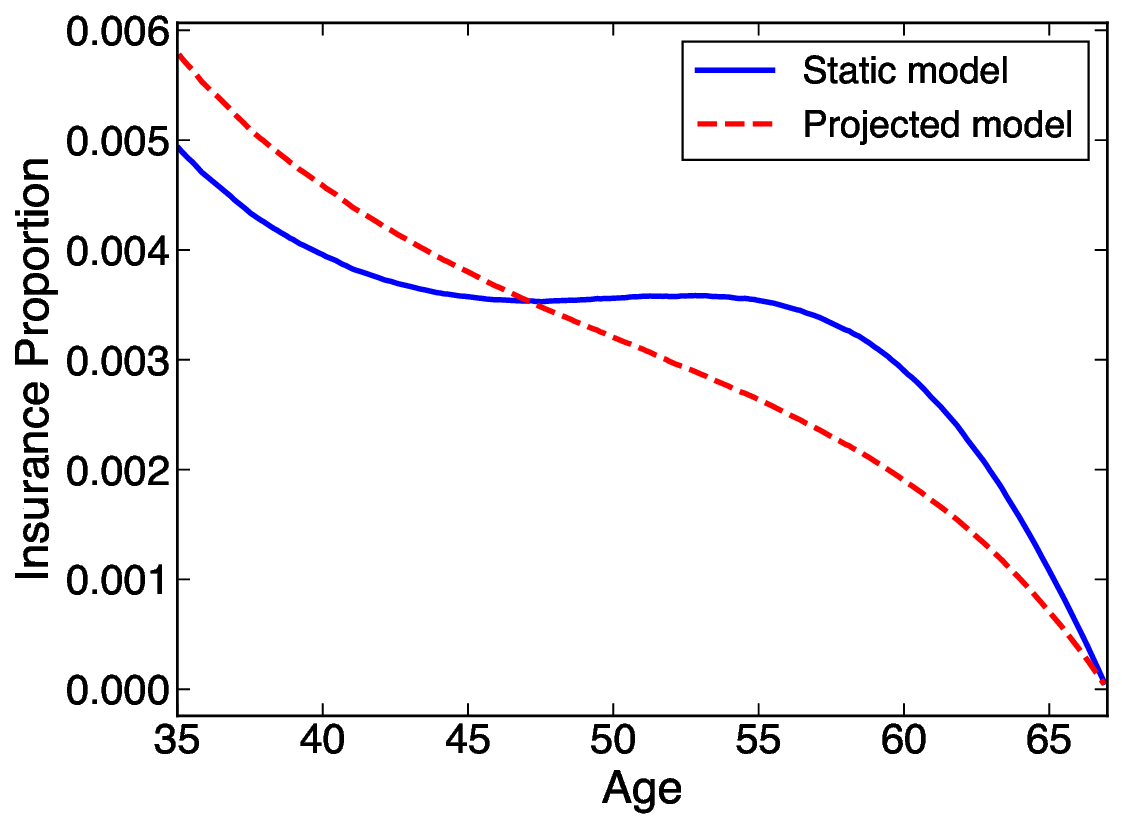}
        \caption{Insurance Proportion}
        \label{fig:analysis4_prop_I}
    \end{subfigure}
    \caption{Static and projected mortality model comparison}
    \label{fig:analysis4}
\end{figure}

\subsection{Sensitivity analysis}
\label{subsec:sensitivity-analysis}

In this section, we conduct a comprehensive sensitivity analysis to examine how key economic and model setting parameters affect optimal strategies $(\pi_t^*, I_t^*)$. In each subsection, we only vary one parameter while holding all others fixed at their baseline values in Table~\ref{financial_market_estimate_table}.

\subsubsection{Impact of target expected wealth ($\mathcal{K}$)}

Figure~\ref{fig:sens_G1_K} suggests that a higher wealth target will motivate pension members to take more risks, or in other words, more risks have to be taken in order to achieve a higher target. At early ages, individuals short more bonds (Figure~\ref{fig:G1_K_piB}) to finance more risky asset positions (Figure~\ref{fig:G1_K_piS}). As retirement approaches, the optimal strategy shifts from risk-taking to a more conservative strategy to secure current wealth. On the insurance side, even though the dollar amount of insurance premium paid is higher for a larger \(\mathcal{K}\) (Figure~\ref{fig:G1_K_I}), the relative proportion to wealth is actually smaller since the wealth level is higher under a larger \(\mathcal{K}\) (Figure~\ref{fig:G1_K_prop_I}).

\begin{figure}[!htbp]
    \centering
    \begin{subfigure}[b]{0.32\textwidth}
        \centering
        \includegraphics[width=\textwidth]{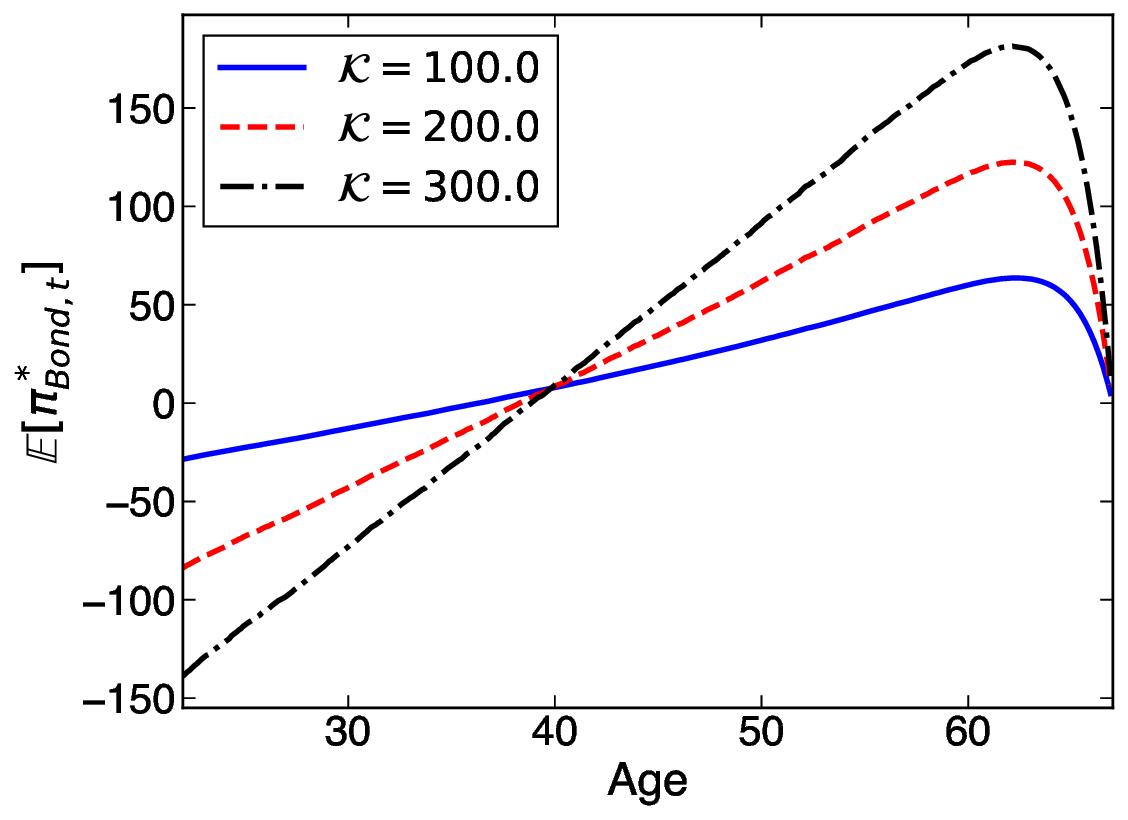}
        \caption{Bond Allocation ($\pi_{Bond}^*$)}
        \label{fig:G1_K_piB}
    \end{subfigure}
    \hfill
    \begin{subfigure}[b]{0.32\textwidth}
        \centering
        \includegraphics[width=\textwidth]{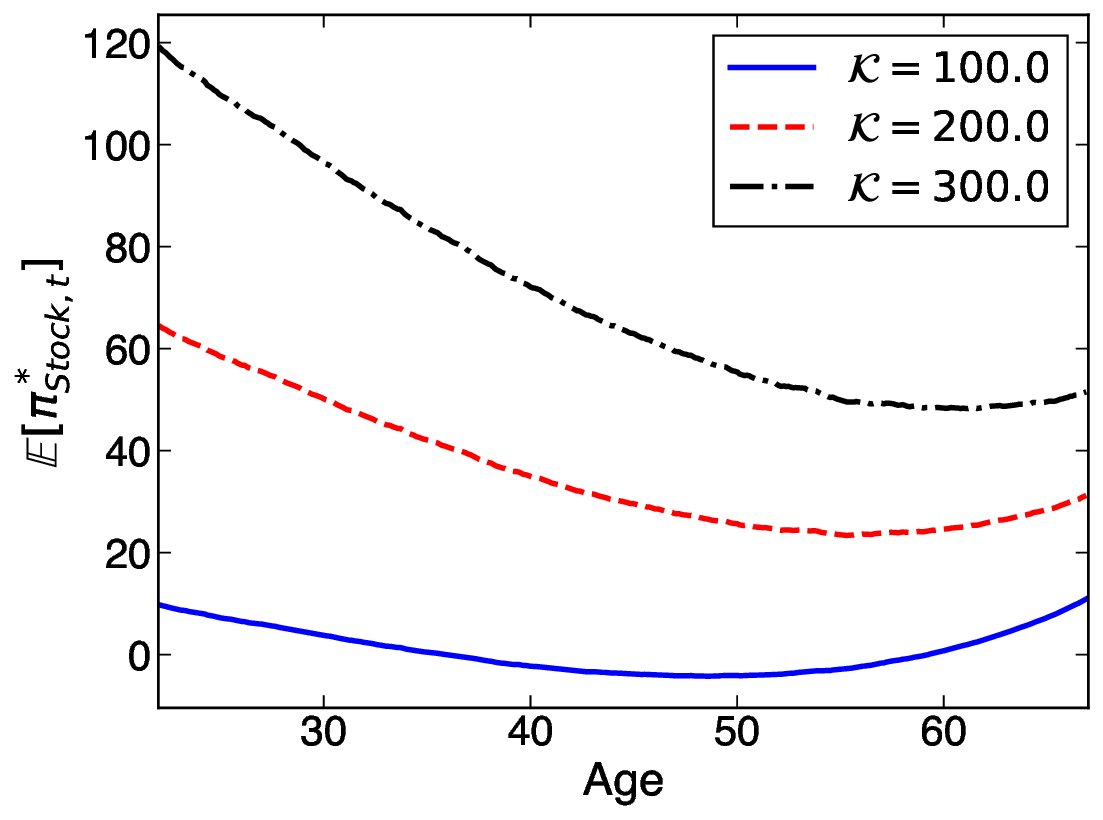}
        \caption{Stock Allocation ($\pi_{Stock}^*$)}
        \label{fig:G1_K_piS}
    \end{subfigure}
    \hfill
    \begin{subfigure}[b]{0.32\textwidth}
        \centering
        \includegraphics[width=\textwidth]{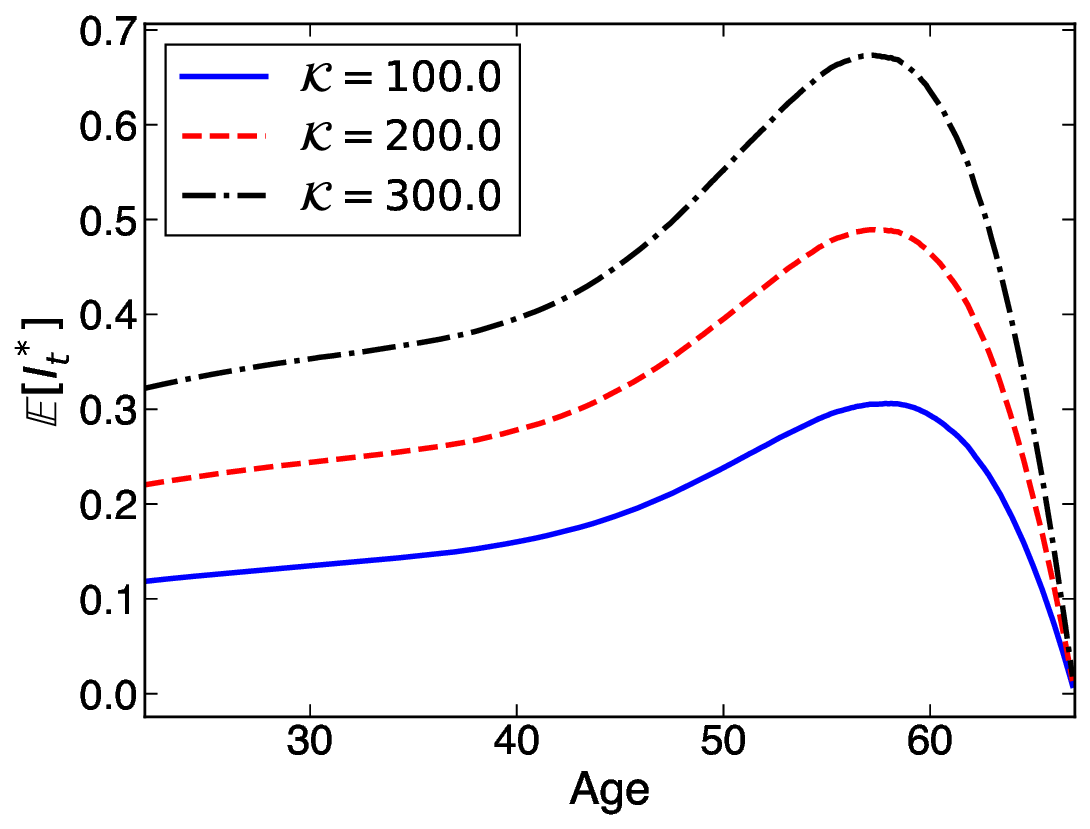}
        \caption{Insurance Premium ($I_t^*$)}
        \label{fig:G1_K_I}
    \end{subfigure}
    
    \vspace{0.6em}
    
    \begin{subfigure}[b]{0.32\textwidth}
        \centering
        \includegraphics[width=\textwidth]{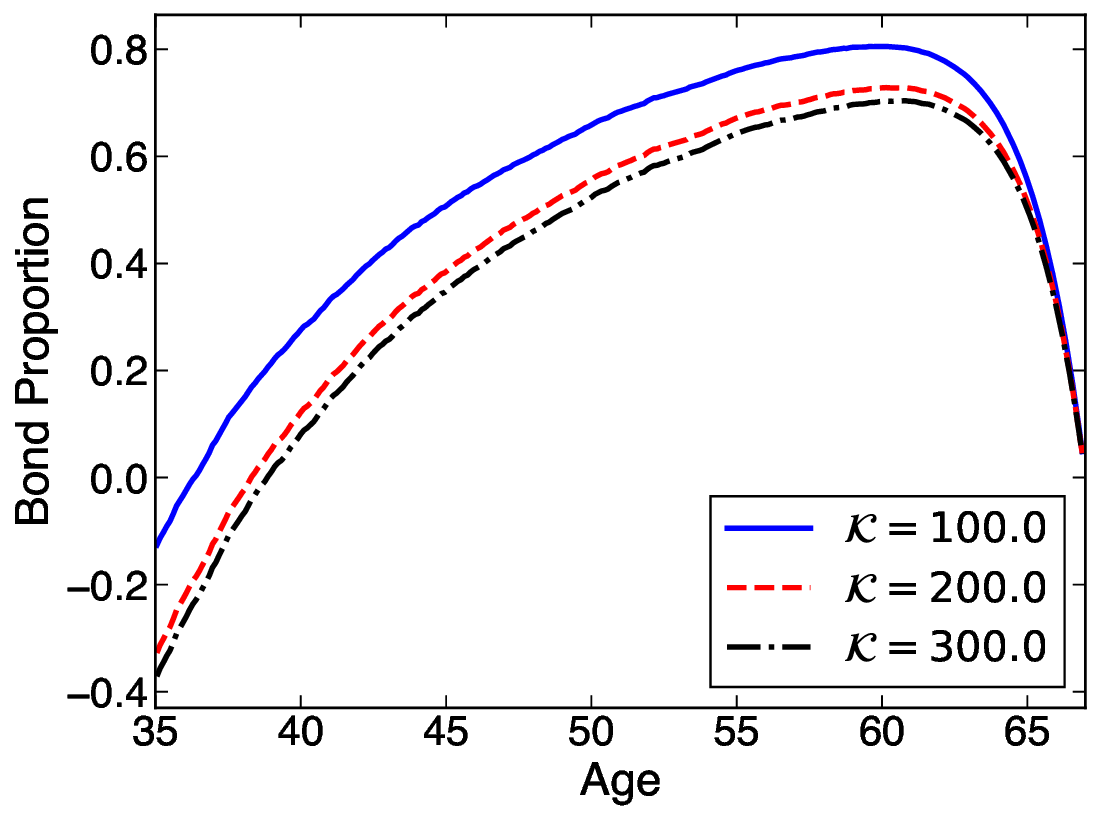}
        \caption{Bond Proportion}
        \label{fig:G1_K_prop_piB}
    \end{subfigure}
    \hfill
    \begin{subfigure}[b]{0.32\textwidth}
        \centering
        \includegraphics[width=\textwidth]{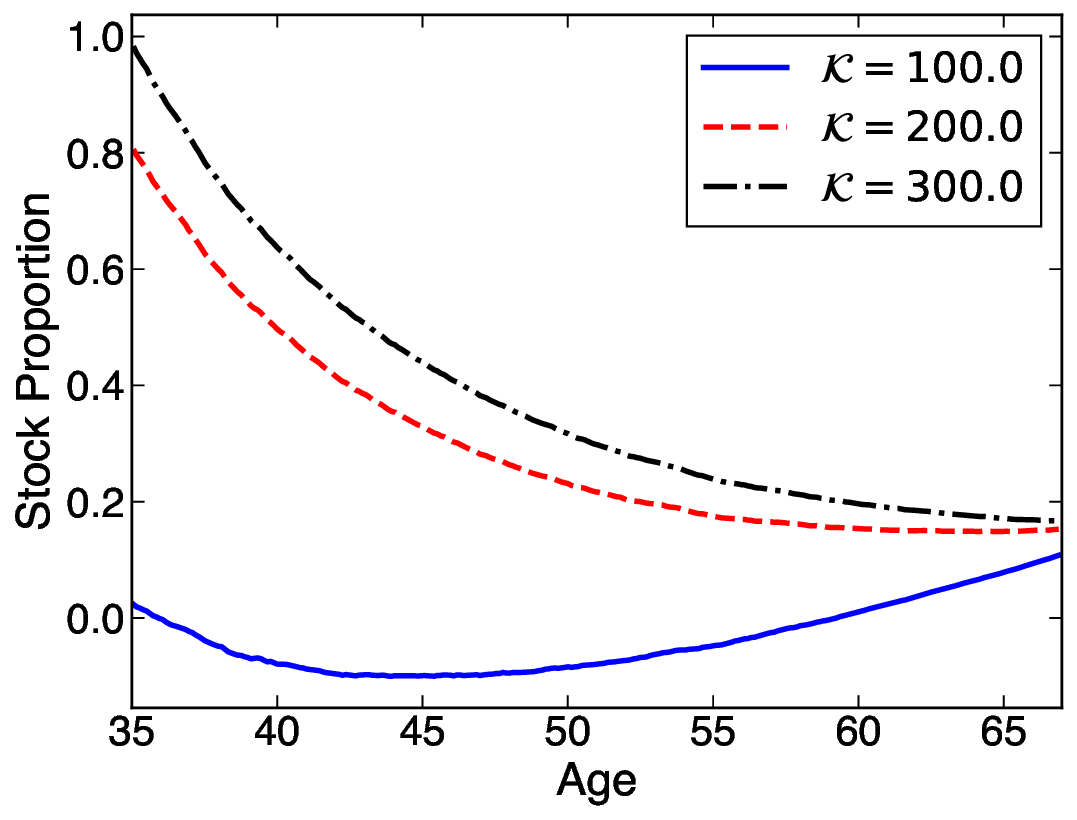}
        \caption{Stock Proportion}
        \label{fig:G1_K_prop_piS}
    \end{subfigure}
    \hfill
    \begin{subfigure}[b]{0.32\textwidth}
        \centering
        \includegraphics[width=\textwidth]{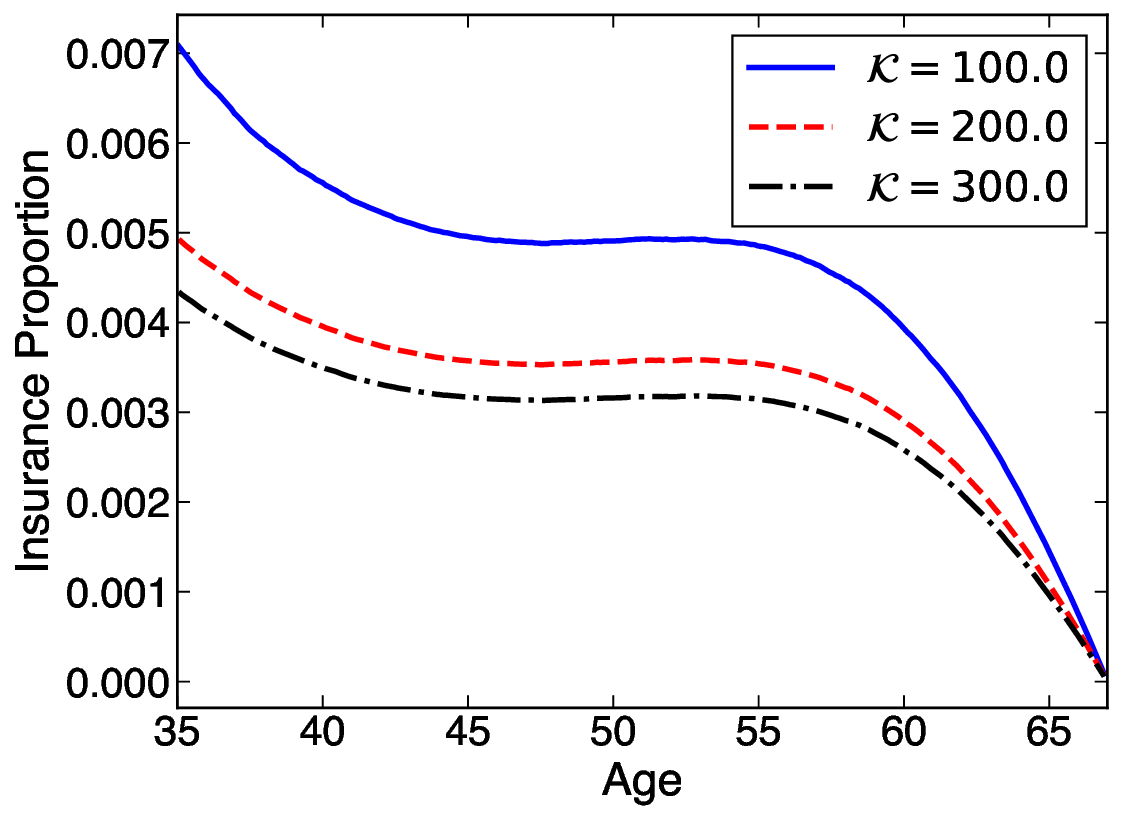}
        \caption{Insurance Proportion}
        \label{fig:G1_K_prop_I}
    \end{subfigure}
    \caption{Sensitivity to the target expected payout $\mathcal{K}$. The top row reports amounts, and the bottom row reports proportions.}
    \label{fig:sens_G1_K}
\end{figure}

\subsubsection{Impact of Sharpe ratios $\lambda_r$ and $\lambda_S$}

As shown in Figure~\ref{fig:sens_G2_lamr}, a higher interest rate Sharpe ratio $\lambda_r$ raises the return on bonds and partially replaces the speculative demand for stocks. Therefore, $\pi_{Bond}^*$ rises and shifts from a short position to a long position at an earlier age (Figure~\ref{fig:G2_lamr_piB}), while $\pi_{Stock}^*$ drops moderately (Figure~\ref{fig:G2_lamr_piS}). The effect on insurance needs is very limited, suggesting that $\lambda_r$ influences the investment strategy more significantly. The proportion-to-wealth plots in Figures~\ref{fig:G2_lamr_prop_piB}--\ref{fig:G2_lamr_prop_I} also support this.

\begin{figure}[!htbp]
    \centering
    \begin{subfigure}[b]{0.32\textwidth}
        \centering
        \includegraphics[width=\textwidth]{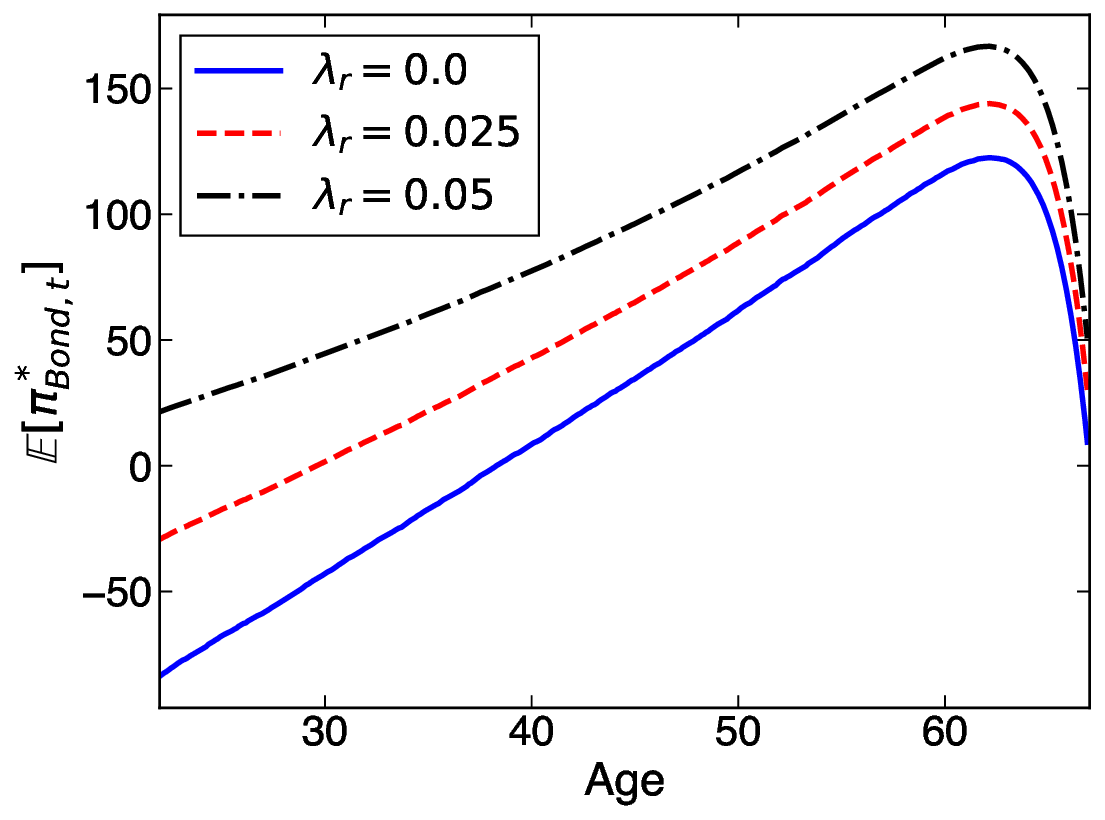}
        \caption{Bond Allocation ($\pi_{Bond}^*$)}
        \label{fig:G2_lamr_piB}
    \end{subfigure}
    \hfill
    \begin{subfigure}[b]{0.32\textwidth}
        \centering
        \includegraphics[width=\textwidth]{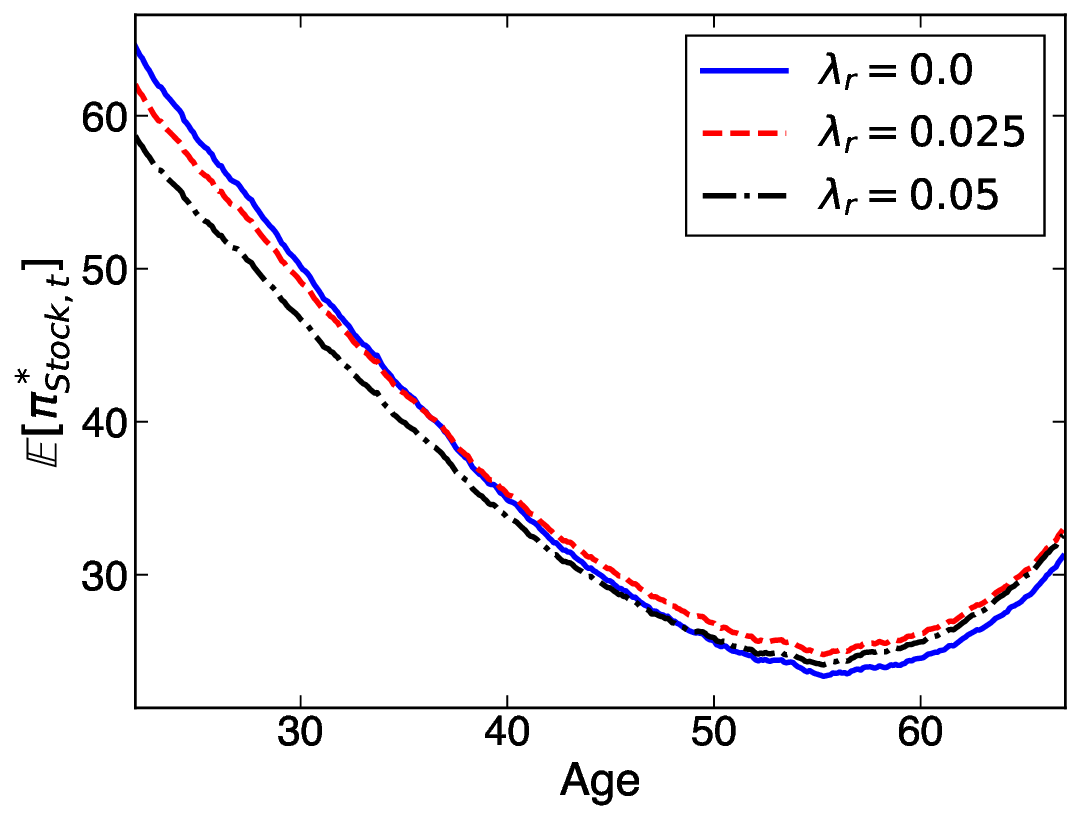}
        \caption{Stock Allocation ($\pi_{Stock}^*$)}
        \label{fig:G2_lamr_piS}
    \end{subfigure}
    \hfill
    \begin{subfigure}[b]{0.32\textwidth}
        \centering
        \includegraphics[width=\textwidth]{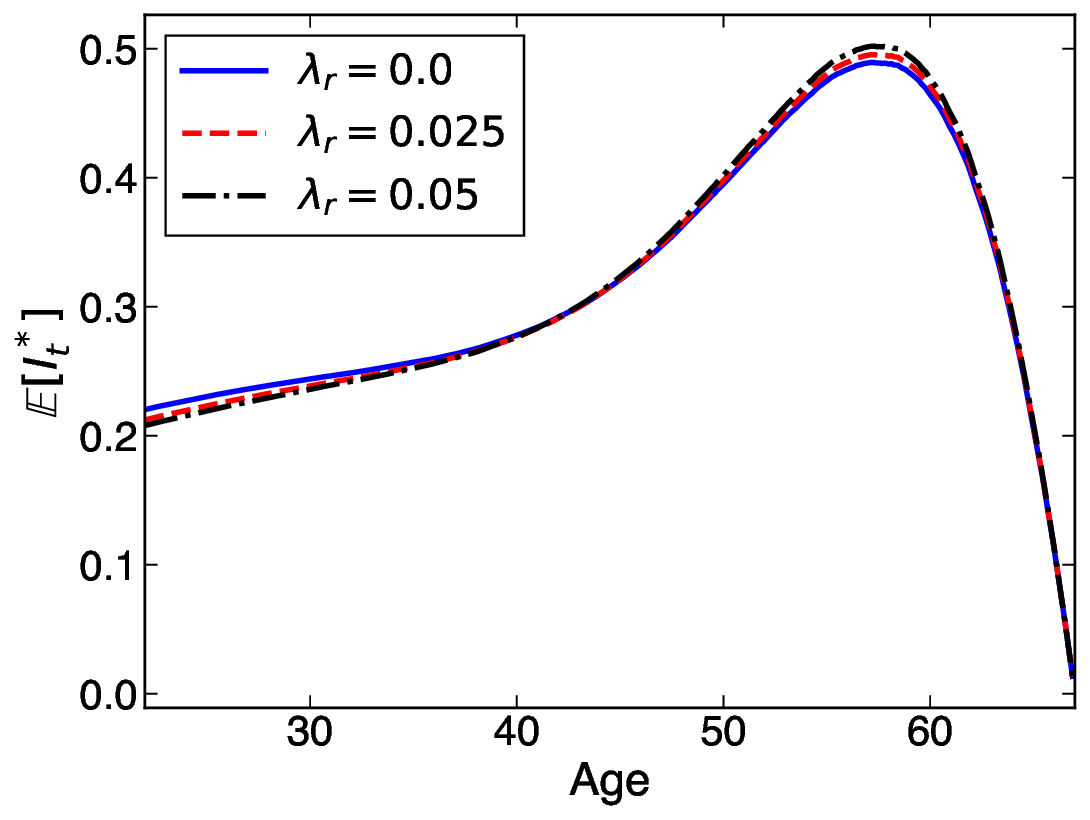}
        \caption{Insurance Premium ($I_t^*$)}
        \label{fig:G2_lamr_I}
    \end{subfigure}
    
    \vspace{0.6em}
    
    \begin{subfigure}[b]{0.32\textwidth}
        \centering
        \includegraphics[width=\textwidth]{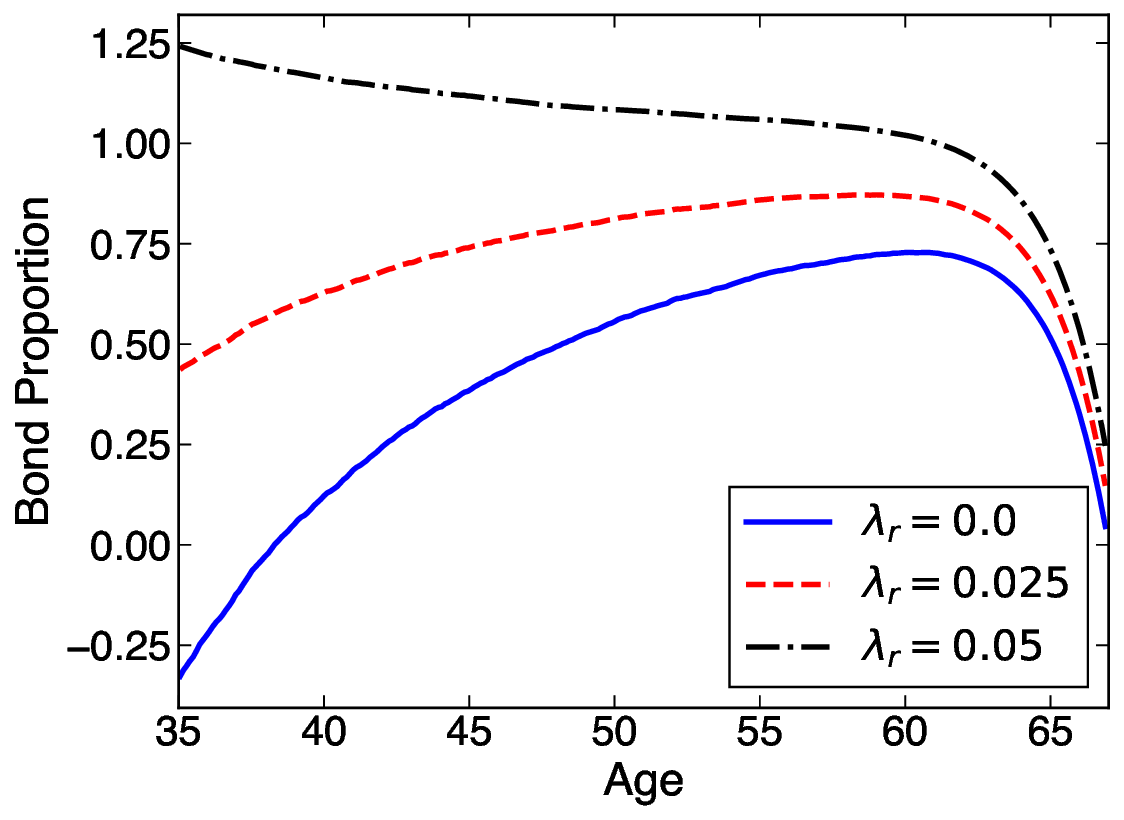}
        \caption{Bond Proportion}
        \label{fig:G2_lamr_prop_piB}
    \end{subfigure}
    \hfill
    \begin{subfigure}[b]{0.32\textwidth}
        \centering
        \includegraphics[width=\textwidth]{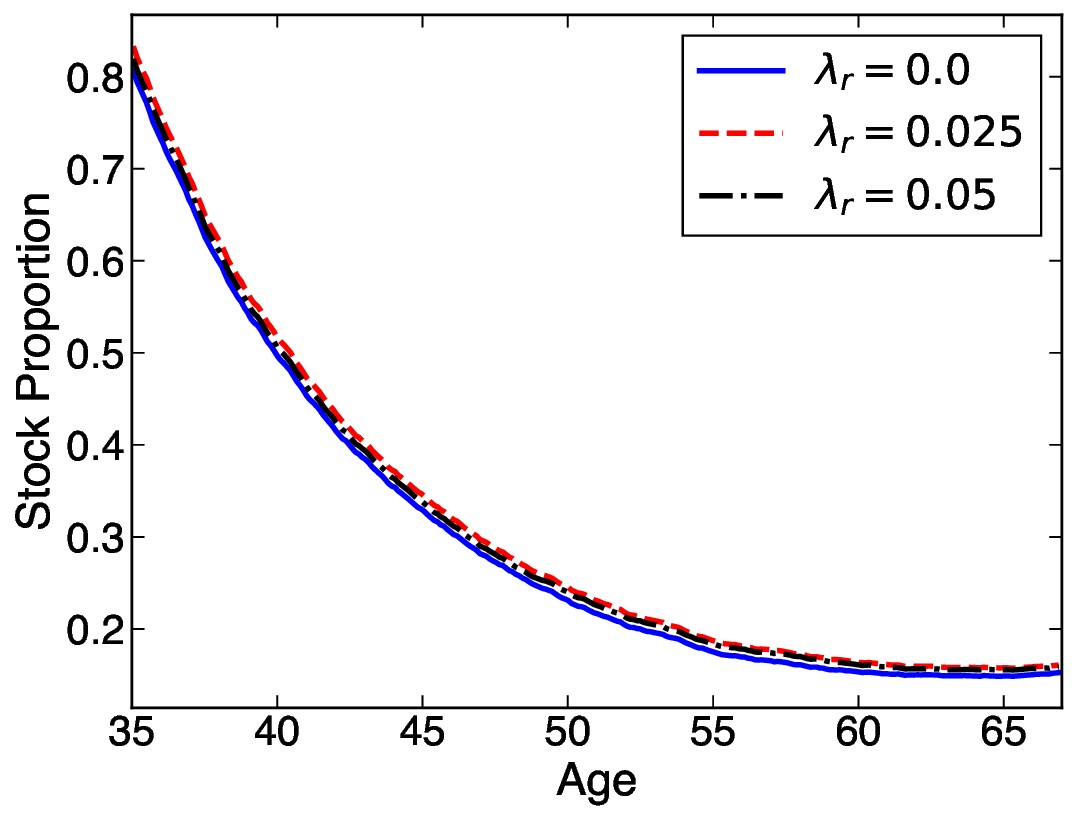}
        \caption{Stock Proportion}
        \label{fig:G2_lamr_prop_piS}
    \end{subfigure}
    \hfill
    \begin{subfigure}[b]{0.32\textwidth}
        \centering
        \includegraphics[width=\textwidth]{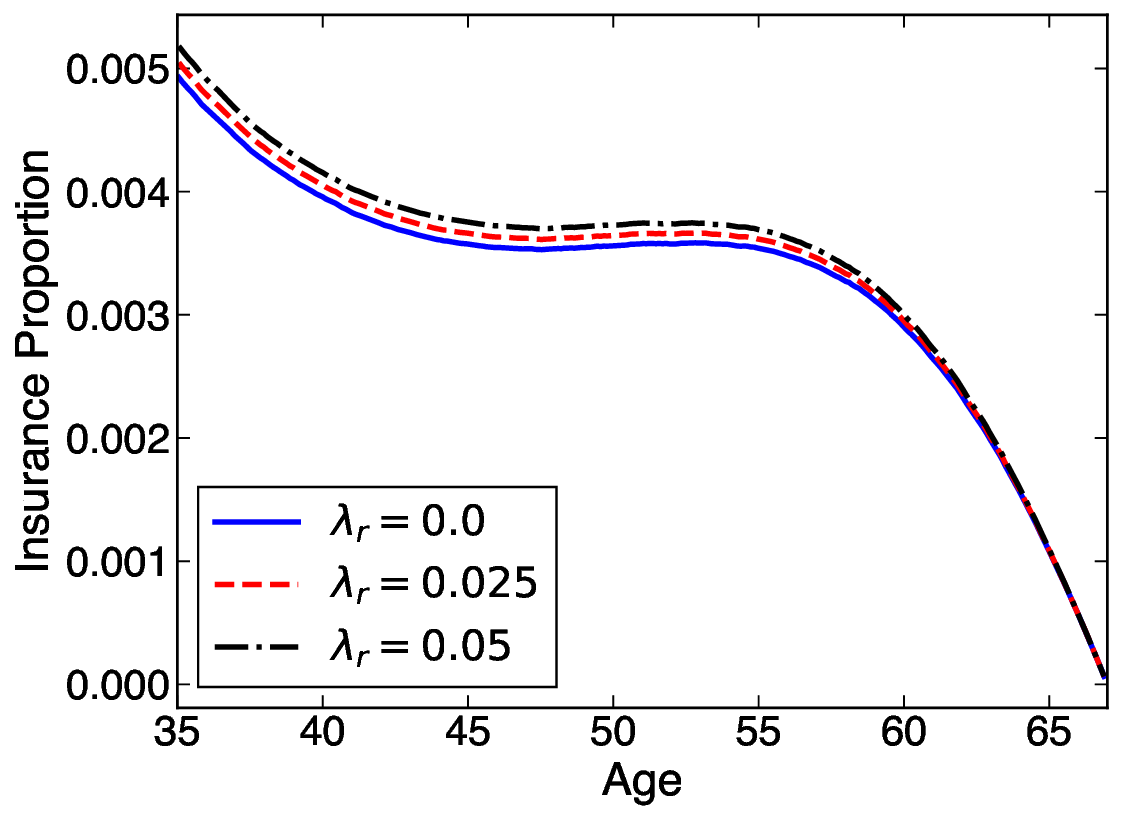}
        \caption{Insurance Proportion}
        \label{fig:G2_lamr_prop_I}
    \end{subfigure}
    \caption{Sensitivity to the interest-rate market price of risk $\lambda_r$. The top row reports amounts, and the bottom row reports proportions.}
    \label{fig:sens_G2_lamr}
\end{figure}

\begin{figure}[!htbp]
    \centering
    \begin{subfigure}[b]{0.32\textwidth}
        \centering
        \includegraphics[width=\textwidth]{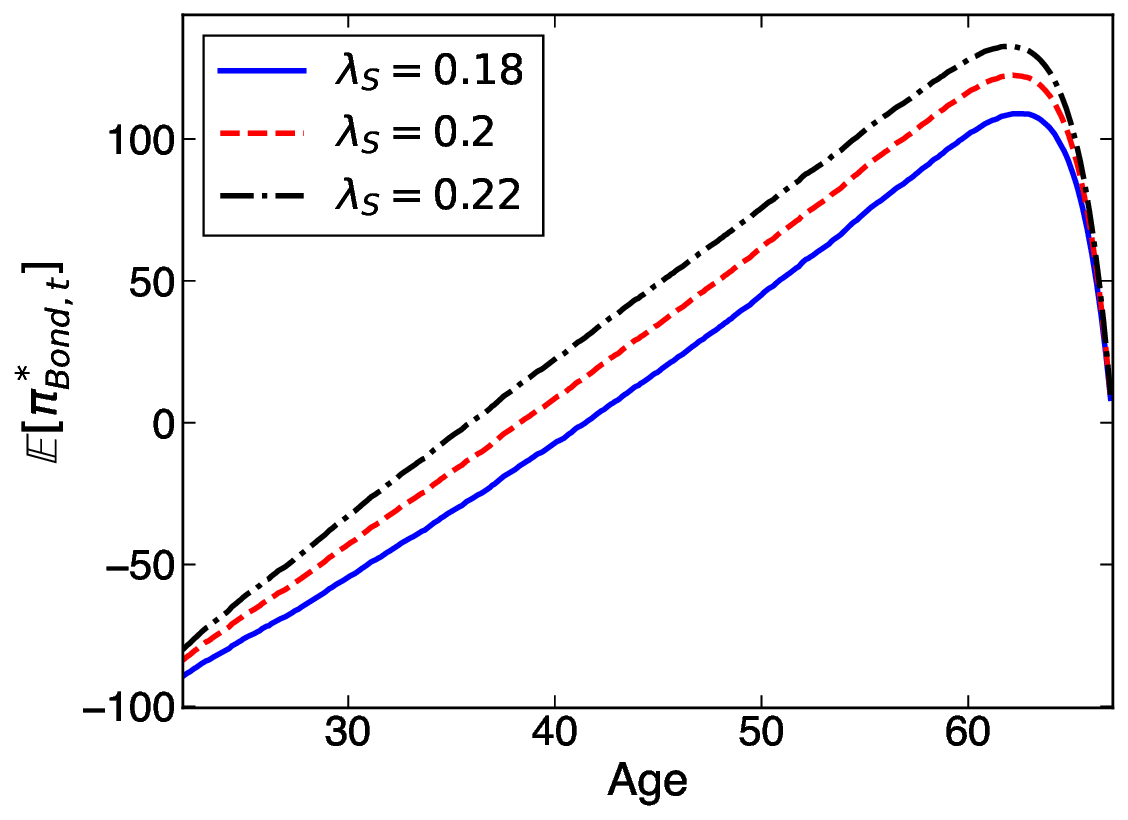}
        \caption{Bond Allocation ($\pi_{Bond}^*$)}
        \label{fig:G3_lamS_piB}
    \end{subfigure}
    \hfill
    \begin{subfigure}[b]{0.32\textwidth}
        \centering
        \includegraphics[width=\textwidth]{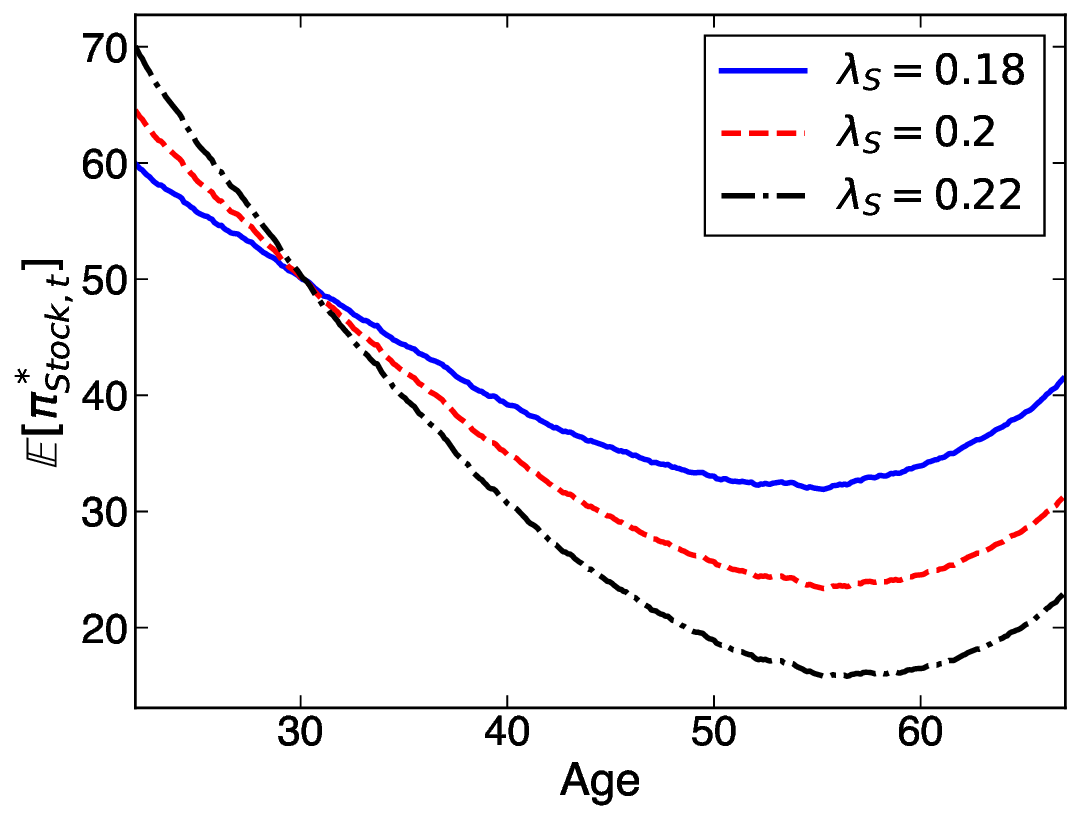}
        \caption{Stock Allocation ($\pi_{Stock}^*$)}
        \label{fig:G3_lamS_piS}
    \end{subfigure}
    \hfill
    \begin{subfigure}[b]{0.32\textwidth}
        \centering
        \includegraphics[width=\textwidth]{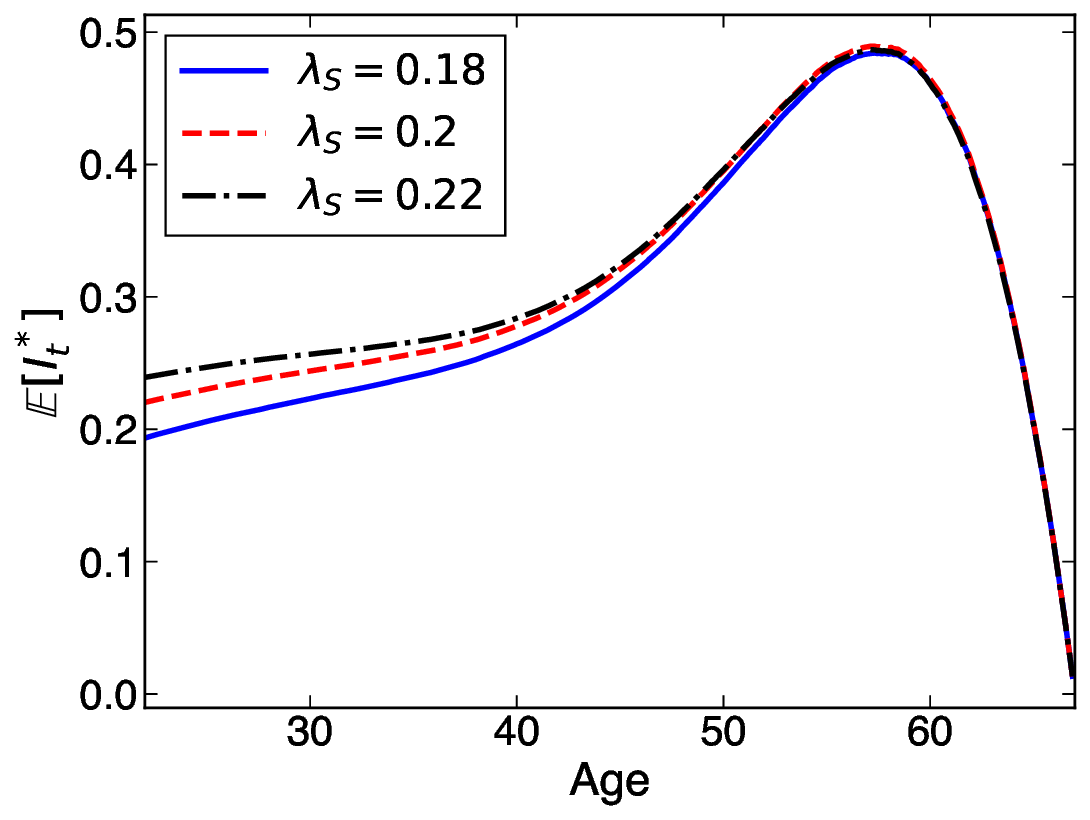}
        \caption{Insurance Premium ($I_t^*$)}
        \label{fig:G3_lamS_I}
    \end{subfigure}
    
    \vspace{0.6em}
    
    \begin{subfigure}[b]{0.32\textwidth}
        \centering
        \includegraphics[width=\textwidth]{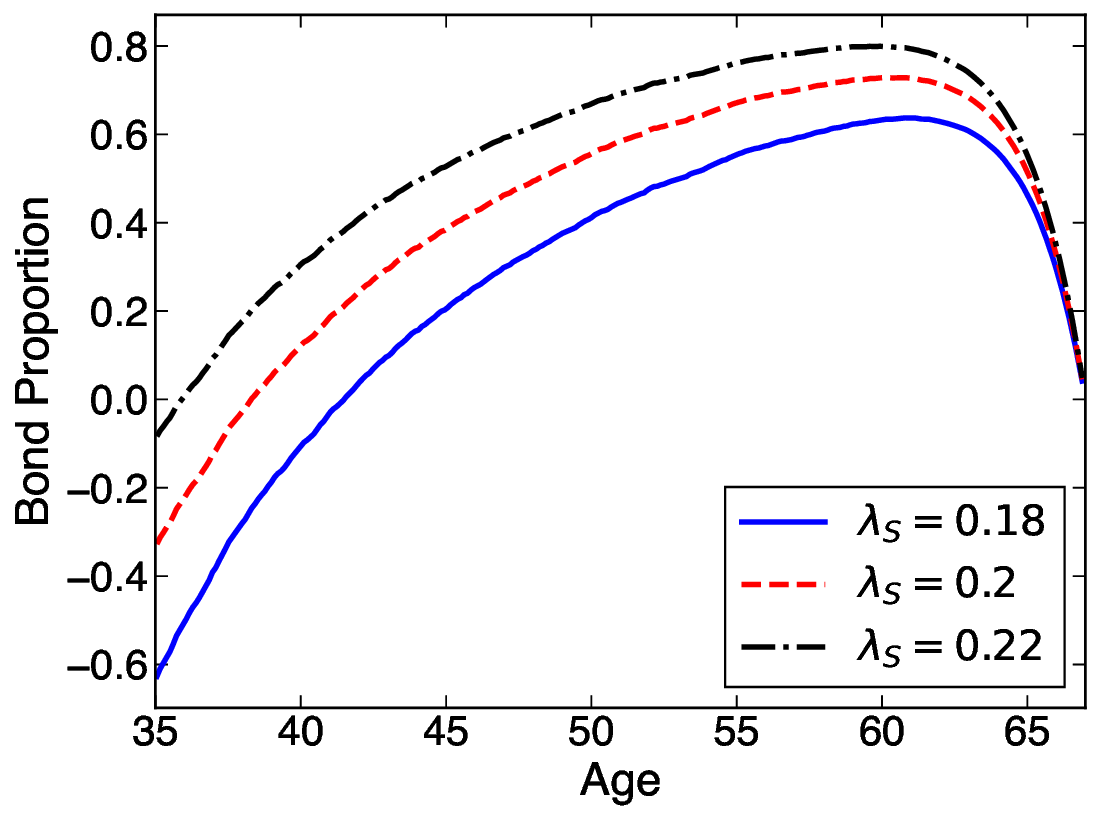}
        \caption{Bond Proportion}
        \label{fig:G3_lamS_prop_piB}
    \end{subfigure}
    \hfill
    \begin{subfigure}[b]{0.32\textwidth}
        \centering
        \includegraphics[width=\textwidth]{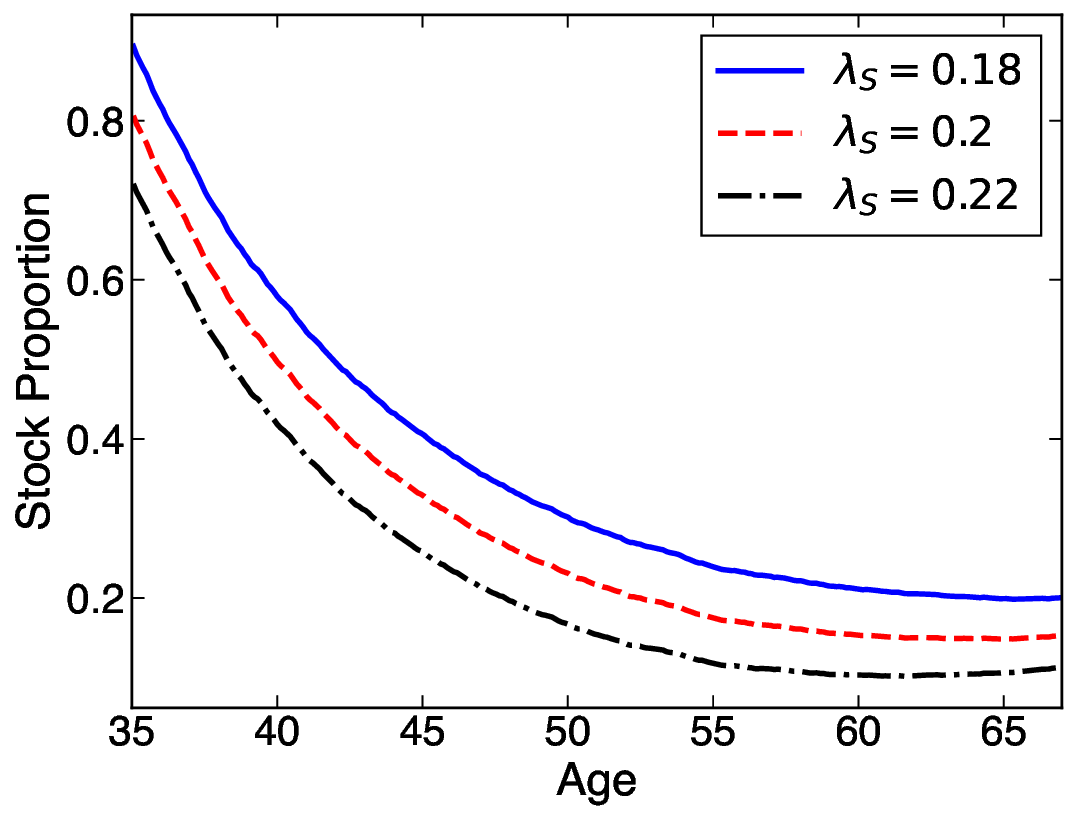}
        \caption{Stock Proportion}
        \label{fig:G3_lamS_prop_piS}
    \end{subfigure}
    \hfill
    \begin{subfigure}[b]{0.32\textwidth}
        \centering
        \includegraphics[width=\textwidth]{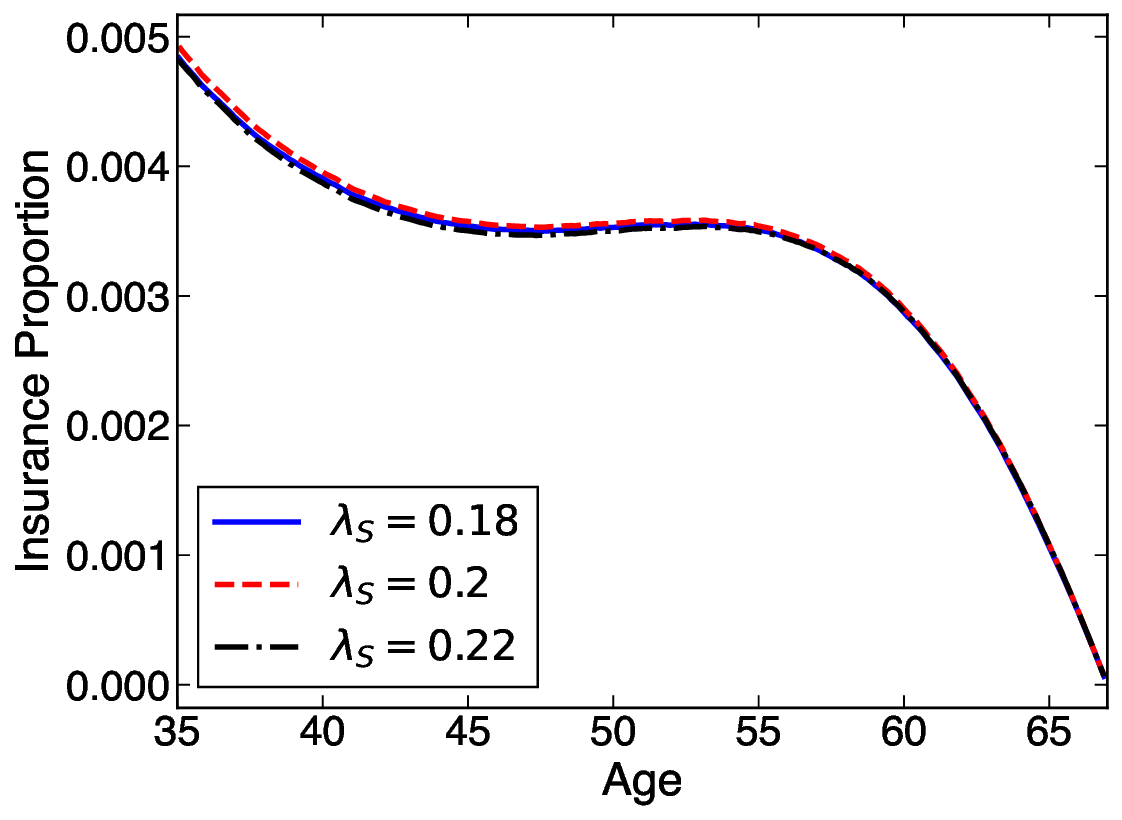}
        \caption{Insurance Proportion}
        \label{fig:G3_lamS_prop_I}
    \end{subfigure}
    \caption{Sensitivity to the equity market price of risk $\lambda_S$. The top row reports amounts, and the bottom row reports proportions.}
    \label{fig:sens_G3_lamS}
\end{figure}

A higher stock Sharpe ratio $\lambda_S$ raises the expected return of stocks per position. Since our retirement target is fixed, a higher per-position return on stocks means individuals can achieve the same total return with smaller positions. Therefore, $\pi_{Bond}^*$ is larger and shifts to long positions at an earlier stage because the financing demand for bonds is reduced, as shown in both Figure~\ref{fig:G3_lamS_piB} and Figure~\ref{fig:G3_lamS_prop_piB}. Stock positions drop accordingly, as in Figure~\ref{fig:G3_lamS_piS} and Figure~\ref{fig:G3_lamS_prop_piS}. The effect of $\lambda_S$ on life insurance is limited, compared to the more significant impact on investment strategy.

\subsubsection{Impact of labor income dynamics ($\xi_0$ and $\sigma_C$)}

A higher labor income drift parameter $\xi_0$ leads to faster wage growth. Since the individual has a fixed wealth target independent of income increases, this change in expected contribution covers a larger proportion of the wealth target. Under the mean-variance framework, this leads to a reduced position in risky assets in order to minimize variance. Therefore, as shown in Figure~\ref{fig:G4_xi0_piS}, $\pi_{Stock}^*$ drops significantly. Higher expected income increases the potential loss from premature death, and more insurance coverage is needed to protect the more valuable human capital. Insurance demand rises accordingly, as shown in both Figures~\ref{fig:G4_xi0_I} and ~\ref{fig:G4_xi0_prop_I}.

\begin{figure}[!htbp]
    \centering
    \begin{subfigure}[b]{0.32\textwidth}
        \centering
        \includegraphics[width=\textwidth]{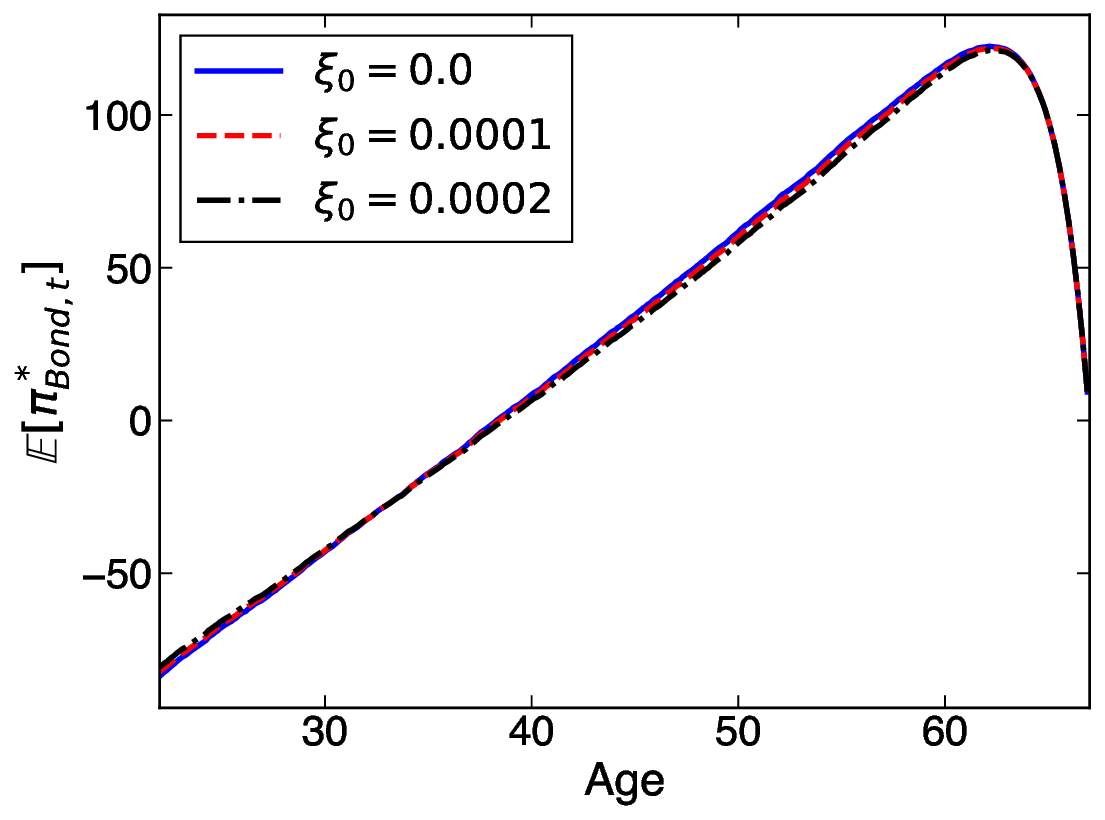}
        \caption{Bond Allocation ($\pi_{Bond}^*$)}
        \label{fig:G4_xi0_piB}
    \end{subfigure}
    \hfill
    \begin{subfigure}[b]{0.32\textwidth}
        \centering
        \includegraphics[width=\textwidth]{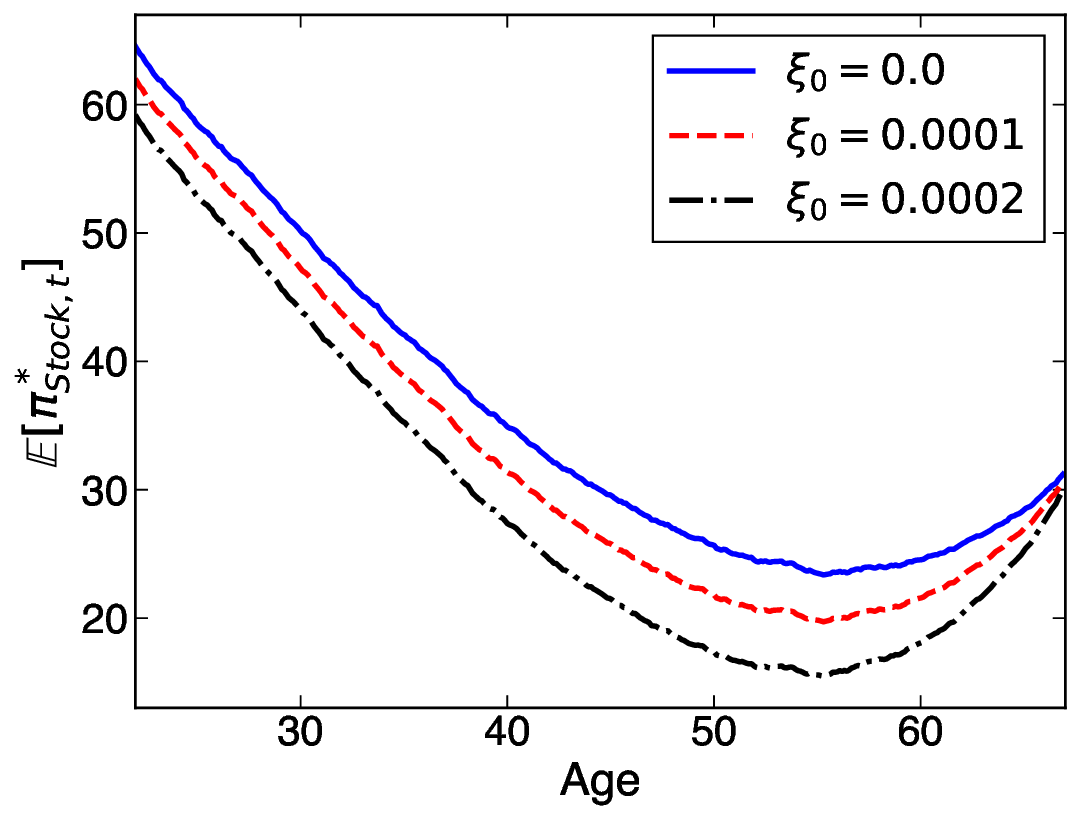}
        \caption{Stock Allocation ($\pi_{Stock}^*$)}
        \label{fig:G4_xi0_piS}
    \end{subfigure}
    \hfill
    \begin{subfigure}[b]{0.32\textwidth}
        \centering
        \includegraphics[width=\textwidth]{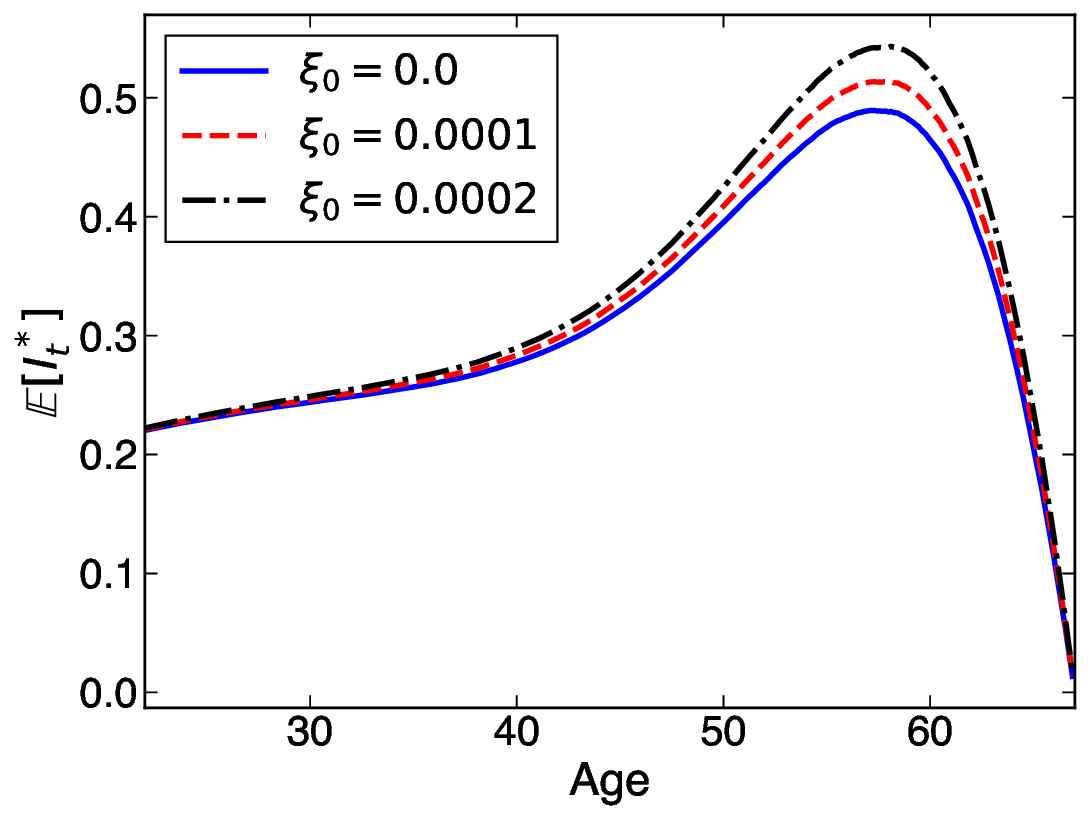}
        \caption{Insurance Premium ($I_t^*$)}
        \label{fig:G4_xi0_I}
    \end{subfigure}
    \begin{subfigure}[b]{0.32\textwidth}
        \centering
        \includegraphics[width=\textwidth]{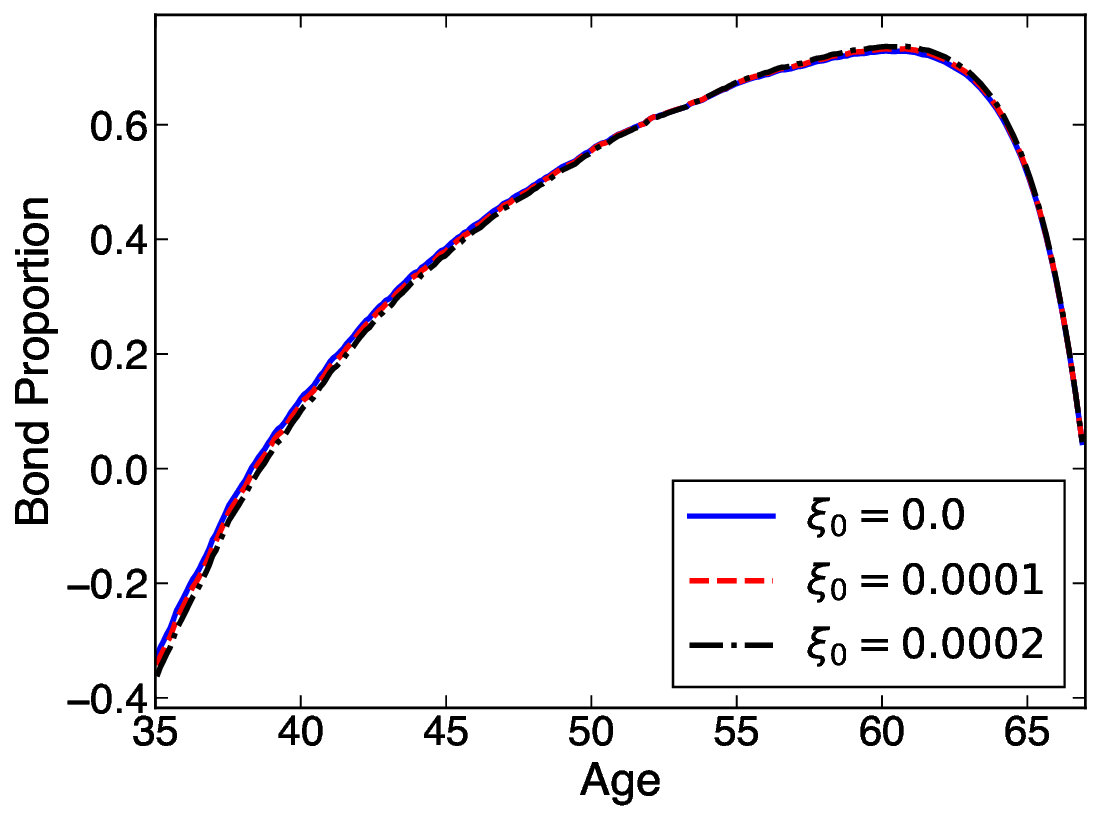}
        \caption{Bond Proportion}
        \label{fig:G4_xi0_prop_piB}
    \end{subfigure}
    \hfill
    \begin{subfigure}[b]{0.32\textwidth}
        \centering
        \includegraphics[width=\textwidth]{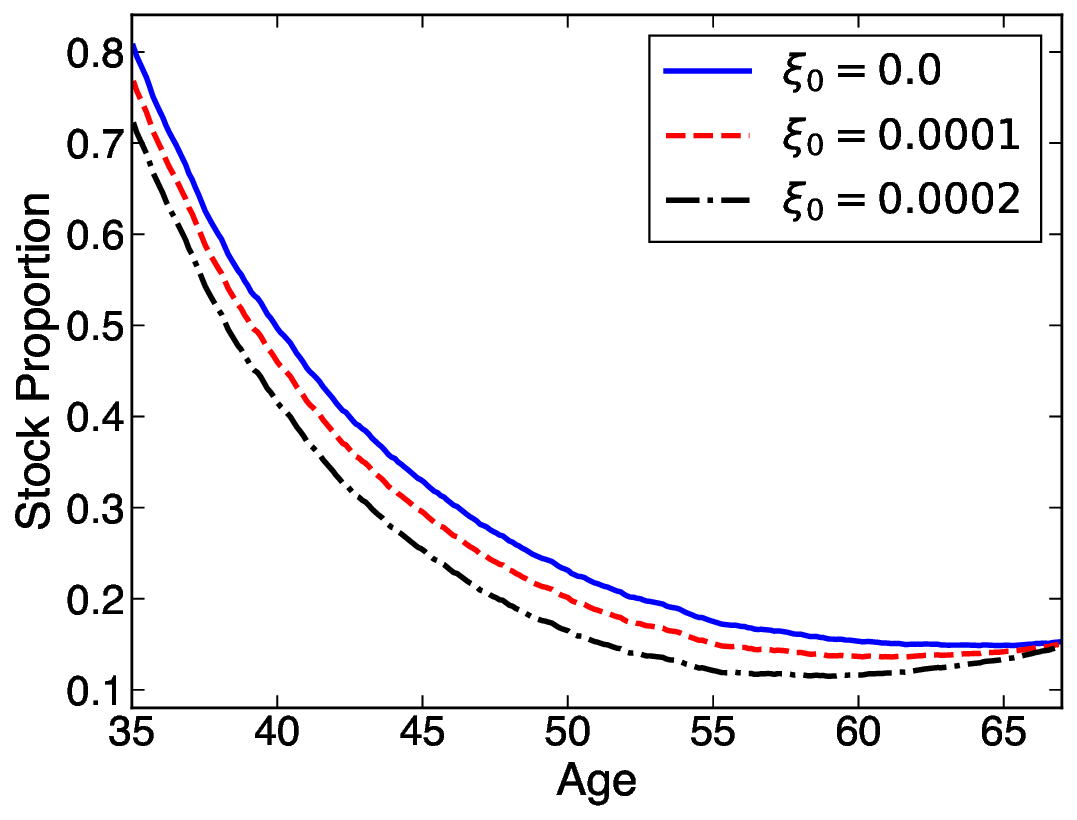}
        \caption{Stock Proportion}
        \label{fig:G4_xi0_prop_piS}
    \end{subfigure}
    \hfill
    \begin{subfigure}[b]{0.32\textwidth}
        \centering
        \includegraphics[width=\textwidth]{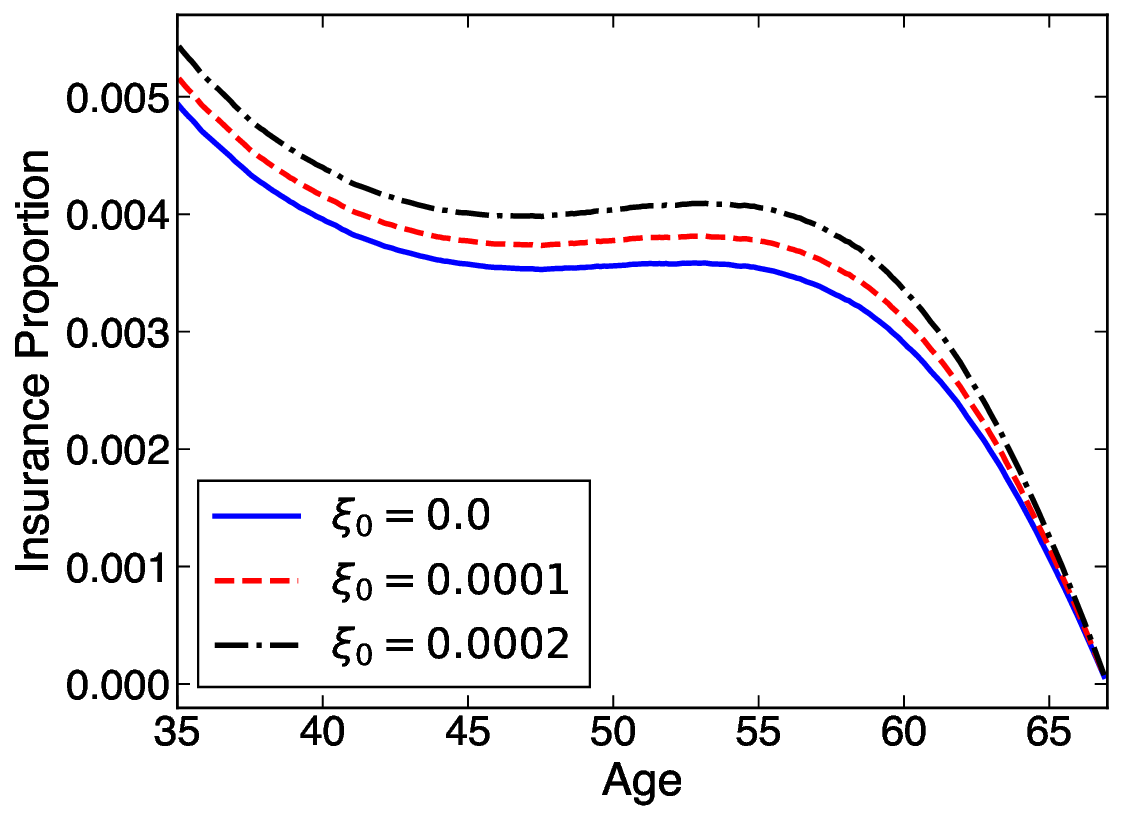}
        \caption{Insurance Proportion}
        \label{fig:G4_xi0_prop_I}
    \end{subfigure}
    \caption{Sensitivity to the labor-income drift parameter $\xi_0$. The top row reports amounts, and the bottom row reports proportions.}
    \label{fig:sens_G4_xi0}
\end{figure}
\vspace{-0.1in}
\begin{figure}[!htbp]
    \centering
    \begin{subfigure}[b]{0.32\textwidth}
        \centering
        \includegraphics[width=\textwidth]{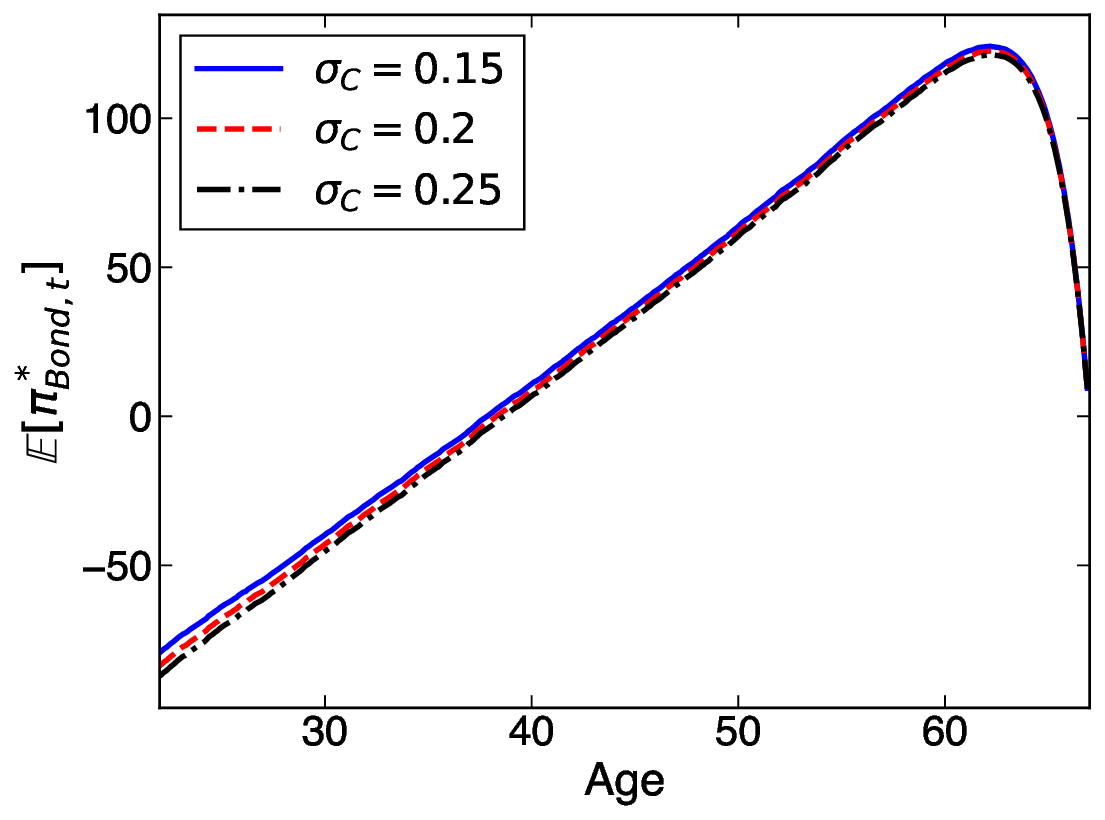}
        \caption{Bond Allocation ($\pi_{Bond}^*$)}
        \label{fig:G4_sigC_piB}
    \end{subfigure}
    \hfill
    \begin{subfigure}[b]{0.32\textwidth}
        \centering
        \includegraphics[width=\textwidth]{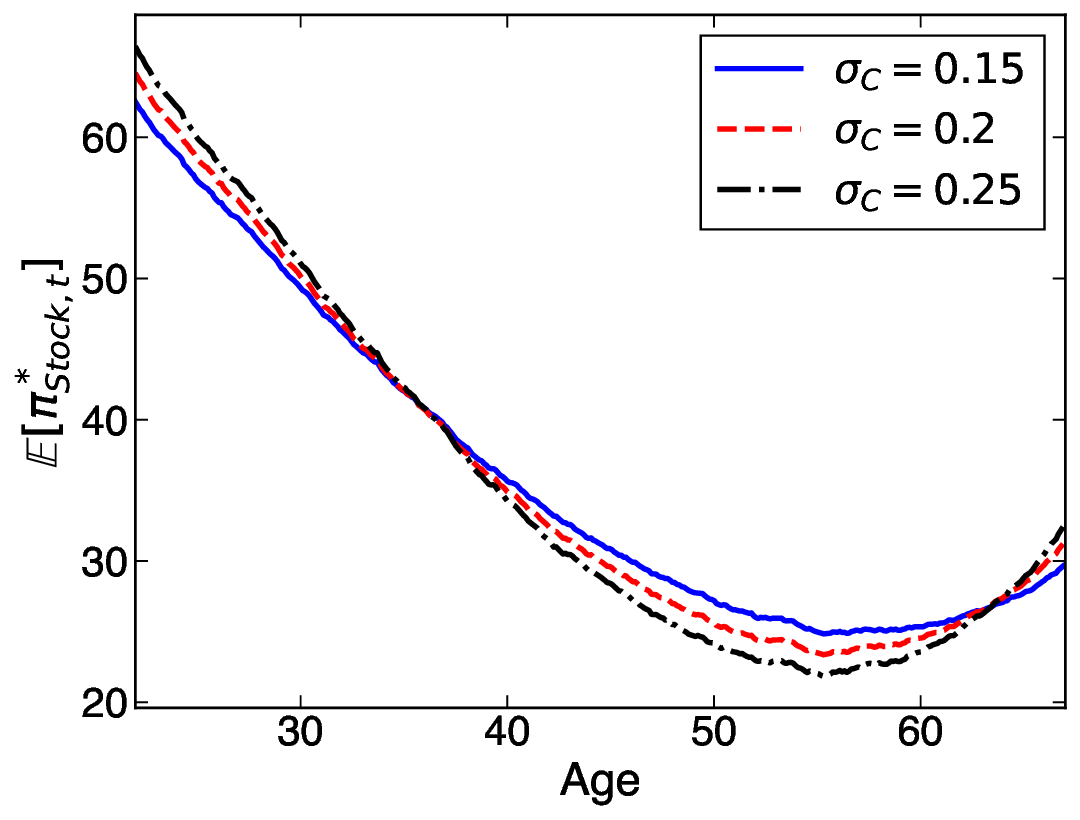}
        \caption{Stock Allocation ($\pi_{Stock}^*$)}
        \label{fig:G4_sigC_piS}
    \end{subfigure}
    \hfill
    \begin{subfigure}[b]{0.32\textwidth}
        \centering
        \includegraphics[width=\textwidth]{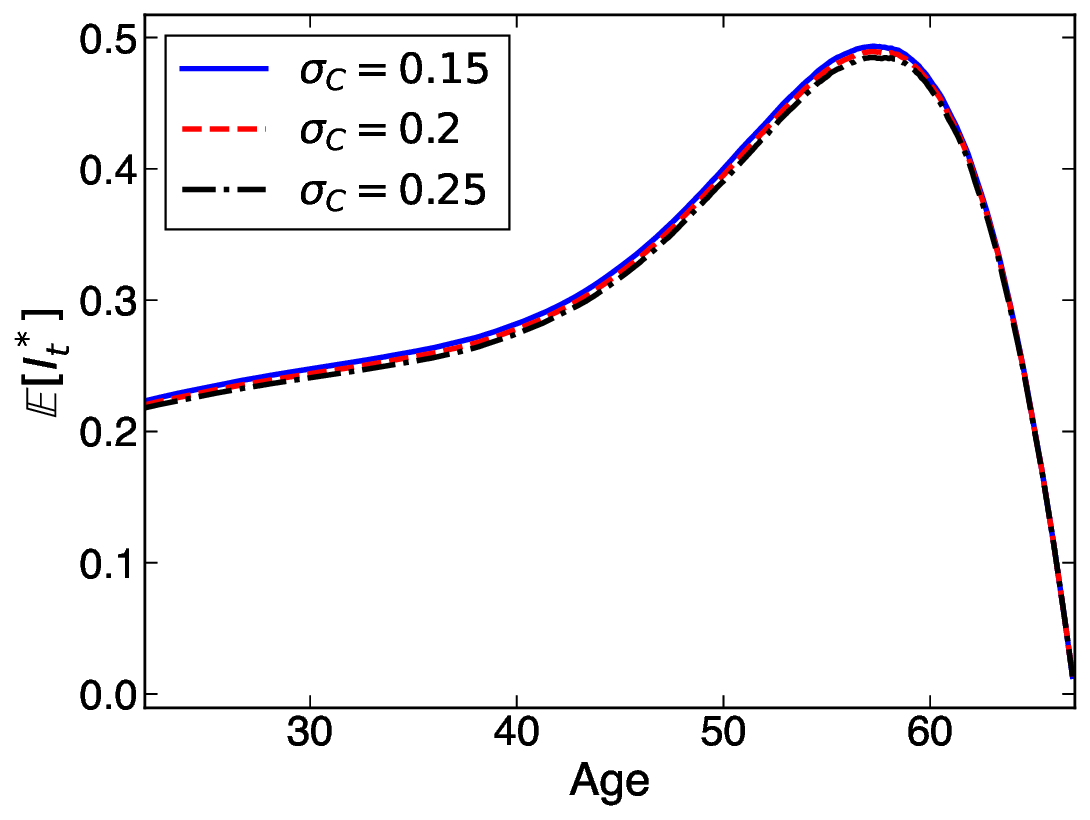}
        \caption{Insurance Premium ($I_t^*$)}
        \label{fig:G4_sigC_I}
    \end{subfigure}
    
    \begin{subfigure}[b]{0.32\textwidth}
        \centering
        \includegraphics[width=\textwidth]{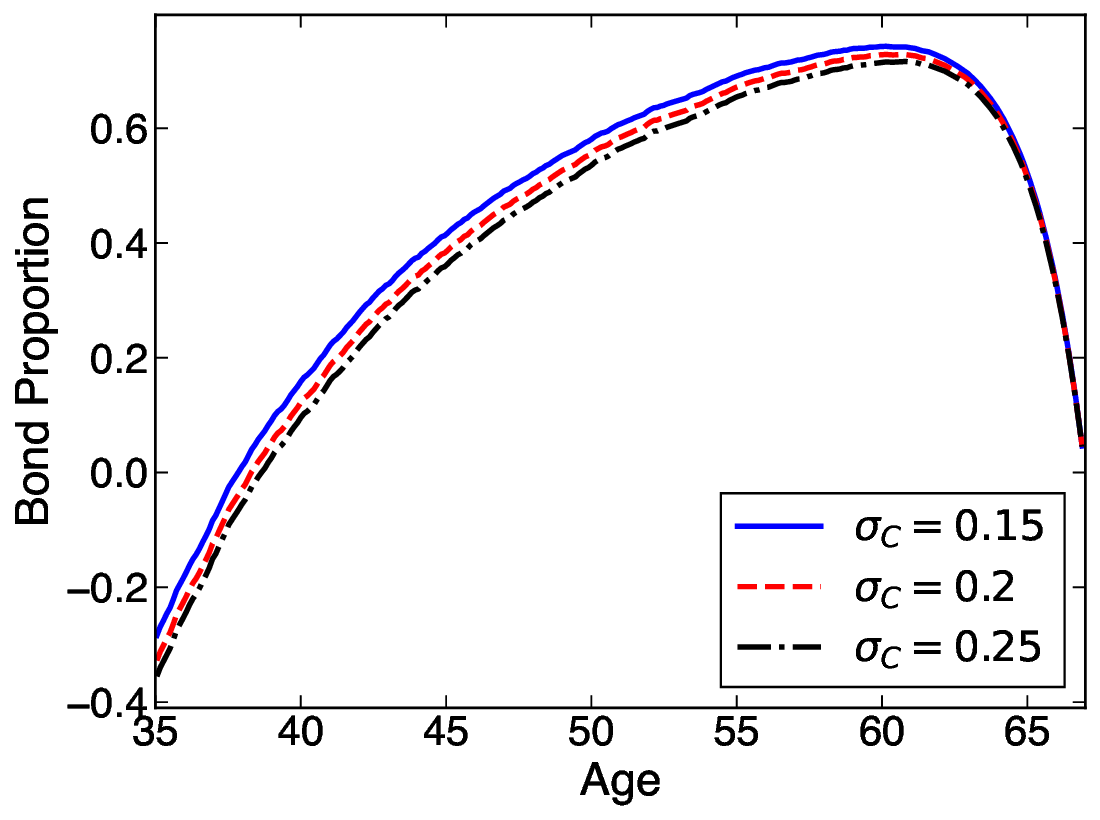}
        \caption{Bond Proportion}
        \label{fig:G4_sigC_prop_piB}
    \end{subfigure}
    \hfill
    \begin{subfigure}[b]{0.32\textwidth}
        \centering
        \includegraphics[width=\textwidth]{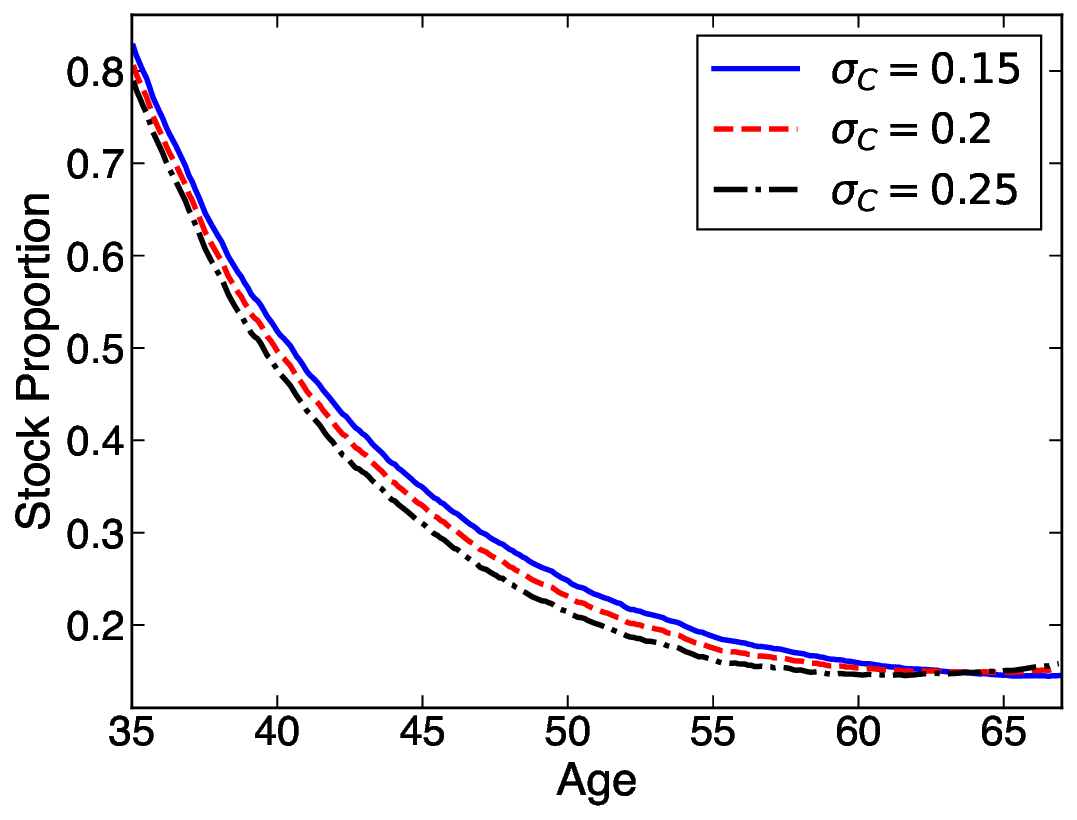}
        \caption{Stock Proportion}
        \label{fig:G4_sigC_prop_piS}
    \end{subfigure}
    \hfill
    \begin{subfigure}[b]{0.32\textwidth}
        \centering
        \includegraphics[width=\textwidth]{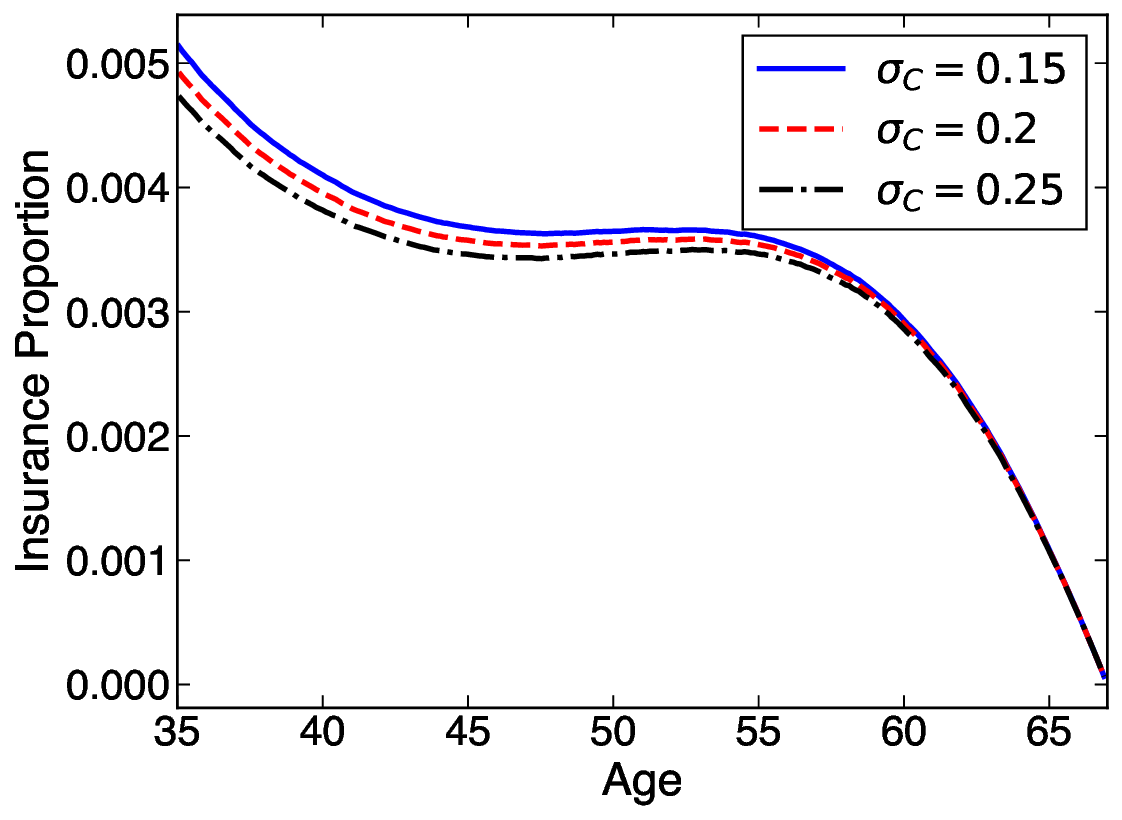}
        \caption{Insurance Proportion}
        \label{fig:G4_sigC_prop_I}
    \end{subfigure}
    \caption{Sensitivity to labor-income volatility $\sigma_C$. The top row reports amounts, and the bottom row reports proportions.}
    \label{fig:sens_G4_sigC}
\end{figure}

Figure~\ref{fig:sens_G4_sigC} examines how changes in labor income volatility $\sigma_C$ affect investment decisions. We find that income uncertainty reduces all allocations to bonds, stocks, and life insurance. Although the stock amount increases a little in the pension member's early age, the proportion invested in the stock overall decreases with the income volatility. Lastly, a less stable future income also reduces the value of human capital, making insurance protection less attractive.

\section{Conclusion}
\label{sec:conclusion}
This paper studies the pre-retirement death in a continuous-time mean-variance DC pension management problem. Pension participants face stochastic interest rates, stochastic contributions, and mortality risks. To hedge such risks, they are allowed to invest in bonds and stocks and purchase life insurance for protection. Using the martingale approach, we transform this dynamic control problem into a static one and derive closed-form optimal terminal wealth, bequest, dynamic portfolio allocations, and insurance strategies. We further prove the global optimality of our solution with a verification theorem. Our numerical analysis suggests a de-risking pattern with time in optimal financial allocation, while the optimal insurance premium is hump-shaped. This is consistent with classical results. Furthermore, we show that mortality improvement associated with the longevity trend reshapes financial and actuarial decisions by increasing the present value of future labor income. The insurance strategy is particularly sensitive to this mortality change compared with bonds and stocks. These findings provide theoretical guidance for DC pension design considering pre-retirement death and also offer practical insights for pension participants to optimize their retirement payoff.

\subsection*{Acknowledgments}
This paper is in memory of Professor Pengyu Wei, a true scholar in Actuarial Science and Mathematical Finance.


\bibliographystyle{apalike}
\bibliography{ref}

\vspace{0.3in}
\noindent Yueman Feng (Corresponding Author)\\
Department of Statistics and Actuarial Science\\
The University of Hong Kong\\
Pokfulam Road, Hong Kong SAR, China\\
Email: fengym@connect.hku.hk\\

\noindent Wenyuan Li\\
Department of Statistics and Actuarial Science\\
The University of Hong Kong\\
Pokfulam Road, Hong Kong SAR, China\\
Email: wylsaas@hku.hk\\

\noindent Mengyi Xu\\
School of Risk and Actuarial Studies\\
UNSW Business School\\
UNSW Sydney, Australia\\
Email: m.xu@unsw.edu.au\\

\noindent Pengyu Wei (In Memoriam)\\
Division of Banking and Finance\\
Nanyang Business School\\
Nanyang Technological University, Singapore\\
E-mail: pengyu.wei@ntu.edu.sg

\appendix

\counterwithin{figure}{section}
\counterwithin{table}{section}
\counterwithin{equation}{section}


\section{Proof of Proposition~\ref{prop:static-budget-necessity}}
\label{app:proof-prop-static-necessity}

\begin{proof}[Proof of Proposition~\ref{prop:static-budget-necessity}]
Define the mortality-adjusted pricing kernel
\[
\Xi_t = {}_t p_x \phi_t .
\]
Since $d({}_t p_x)=-{}_t p_x\mu_{x+t}\,dt$ and $\phi_t$ satisfies \eqref{eq:pricing-kernel-sde}, we have
\begin{equation}
\label{eq:mortality-adjusted-kernel-dynamics}
d\Xi_t
= -\Xi_t(r_t+\mu_{x+t})\,dt-\Xi_t\Lambda^\top dZ_t .
\end{equation}
Applying It\^o's product rule to $\Xi_t W_t$ and using the wealth dynamics \eqref{eq:wealth-sde-M}, the drift terms involving $r_t+\mu_{x+t}$ and $\pi_t^\top\Sigma_F\Lambda$ cancel, yielding
\begin{equation}
\label{eq:discounted-wealth-product}
d(\Xi_t W_t)
= \Xi_t(C_t-\mu_{x+t}M_t)\,dt
+ \Xi_t\bigl(\pi_t^\top\Sigma_F-W_t\Lambda^\top\bigr)dZ_t .
\end{equation}
By the Definition~\ref{def:admissible}, together with the finite-moment property of the pricing kernel over the finite horizon, integrating over $[0,T]$ and taking expectations gives
\begin{equation}
\label{eq:budget-necessity-appendix}
\mathbb{E}\left[ {}_T p_x \phi_T W_T \right]
- \mathbb{E}\left[
\int_0^T {}_s p_x \phi_s (C_s - \mu_{x+s} M_s) \, ds
\right] = w_0.
\end{equation}
Rearranging the terms yields the static budget constraint \eqref{eq:budget-necessity}, completing the necessity proof.
\end{proof}

\section{Proof of Proposition~\ref{prop:static-budget-sufficiency}}
\label{app:proof-prop-static-sufficiency}

\begin{proof}[Proof of Proposition~\ref{prop:static-budget-sufficiency}]
Define the mortality-adjusted pricing kernel
\[
\Xi_t = {}_t p_x \phi_t.
\]

For a given $\mathcal{G}_T$-measurable random variable $\eta$ and an intermediate process $M_t$ satisfying the static constraint \eqref{eq:budget-sufficiency}, define the conditional budget process
\begin{equation}
\label{eq:N-process}
N_t = \mathbb{E} \left[
\int_0^T \Xi_s (\mu_{x+s} M_s - C_s) \, ds
+ \Xi_T \eta \,\Big|\, \mathcal{G}_t
\right].
\end{equation}
The static budget constraint implies $N_0 = w_0$. Under the stated square-integrability assumptions, $N_t$ is a square-integrable $\mathbb{P}$-martingale on the market-spanned filtration. By the Martingale Representation Theorem, there exists a square-integrable predictable process $\vartheta_t$ such that
\begin{equation}
\label{eq:martingale-representation}
N_t = w_0 + \int_0^t \vartheta_s^\top dZ_s.
\end{equation}
Define the candidate wealth process by
\begin{equation}
\label{eq:replicating-wealth}
W_t = \frac{1}{\Xi_t} \left[
N_t - \int_0^t \Xi_s (\mu_{x+s} M_s - C_s) \, ds
\right].
\end{equation}
By construction, $W_T=\eta$. Applying It\^o's formula to \eqref{eq:replicating-wealth}, using \eqref{eq:mortality-adjusted-kernel-dynamics}, and matching the diffusion term in \eqref{eq:wealth-sde-M}, we choose
\begin{equation}
\label{eq:replicating-pi}
\pi_t = (\Sigma_F \Sigma_F^\top)^{-1} \Sigma_F \left(
\frac{\vartheta_t}{\Xi_t} + \Lambda W_t
\right).
\end{equation}
The vector $\vartheta_t/\Xi_t+\Lambda W_t$ has zero third component, so it lies in the row space of $\Sigma_F$ under the no-unspanned-risk setting. The existence of the inverse $(\Sigma_F \Sigma_F^\top)^{-1}$ is guaranteed by the strictly positive definite constant financial market volatility matrix in \eqref{eq:sigma-F}. 

With this choice of $\pi_t$, the dynamics of $W_t$ coincide with \eqref{eq:wealth-sde-M}. Finally, setting $I_t=\mu_{x+t}(M_t-W_t)$ gives the bequest identity \eqref{eq:bequest-M}. The square-integrability assumptions and the martingale representation ensure that the constructed controls satisfy the admissibility requirements, so the static target is dynamically attainable.
\end{proof}

\section{Proof of Lemma~\ref{lem:conditional-expectations}}
\label{app:proof-lem-conditional-expectations}

\begin{proof}[Proof of Lemma~\ref{lem:conditional-expectations}]
We demonstrate the derivation for the first conditional expectation \eqref{eq:Et-phiu}. Write $\tau = u - t$. By the definition of conditional expectation, the process
\[
\widetilde{N}_t = \mathbb{E}_t[\phi_u]
\]
must be a $\mathbb{P}$-martingale. We conjecture the solution takes the exponential affine form
\[
\widetilde{N}_t = \phi_t F(t, r_t),
\quad
F_1(t, r_t) = \exp[f_1(\tau) + f_2(\tau) r_t].
\]
Recall the dynamics of the pricing kernel and the short rate from \eqref{eq:pricing-kernel-sde} and \eqref{eq:short-rate}:
\begin{align}
\frac{d\phi_t}{\phi_t} &= - r_t \, dt - \Lambda^\top dZ_t, \quad \phi_0 = 1,\\
dr_t &= \kappa(\bar{r} - r_t) \, dt - \bar{\sigma}_r^\top dZ_t.
\end{align}
Applying It\^o's lemma to $F_1(t, r_t)$ yields:
\begin{equation}
\label{eq:dF}
dF_1(t,r_t) = F_1(t,r_t)\Big[
-(f_1'(\tau)+f_2'(\tau)r_t)
+ \kappa(\bar r-r_t)f_2(\tau)
+ \tfrac{1}{2}f_2(\tau)^2\|\bar\sigma_r\|^2
\Big]dt
- F_1(t,r_t)f_2(\tau)\bar\sigma_r^\top dZ_t.
\end{equation}
Using the product rule for It\^o processes, the dynamics of $\widetilde{N}_t = \phi_t F_1(t, r_t)$ are given by:
\begin{equation}
\label{eq:dNF}
\begin{aligned}
d(\phi_tF_1(t,r_t))
&= \phi_tF_1(t,r_t)\Big[
-(f_1'(\tau)+f_2'(\tau)r_t)
+ \kappa(\bar r-r_t)f_2(\tau)
+ \tfrac{1}{2}f_2(\tau)^2\bar{\sigma}_r^\top \bar{\sigma}_r
- r_t \\
&\quad + f_2(\tau)\Lambda^\top\bar\sigma_r\Big]dt - \phi_tF_1(t,r_t)\big(\Lambda+f_2(\tau)\bar\sigma_r\big)^\top dZ_t.
\end{aligned}
\end{equation}
For $\widetilde{N}_t$ to be a martingale, the drift term must be identically zero for all $r_t$. This yields the following system of ODEs:
\begin{equation}
\label{eq:ODE-f1-f2}
\begin{cases}
f_2'(\tau)+\kappa f_2(\tau)+1=0, & f_2(0)=0, \\[2pt]
f_1'(\tau)=\kappa\bar r f_2(\tau)
+ \tfrac{1}{2}f_2(\tau)^2\bar\sigma_r^\top\bar\sigma_r
+ f_2(\tau)\bar\sigma_r^\top\Lambda, & f_1(0)=0.
\end{cases}
\end{equation}
Solving these ODEs gives the closed-form expressions for $f_1(\tau)$ and $f_2(\tau)$.

Following a similar procedure for the second expectation, we conjecture
\[
\mathbb{E}_t[\phi_u^2] = \phi_t^2 F_2(t, r_t)
\quad \text{with} \quad
F_2(t, r_t) = \exp[f_3(\tau) + f_4(\tau) r_t].
\]
Applying It\^o's lemma to $\phi_t^2 F_2(t,r_t)$ and setting the drift to zero leads to:
\begin{equation}
\label{eq:ODE-f3-f4}
\begin{cases}
f_4'(\tau)+\kappa f_4(\tau)+2=0, & f_4(0)=0, \\[2pt]
f_3'(\tau)=\kappa\bar r f_4(\tau)
+ \tfrac{1}{2}f_4(\tau)^2\bar\sigma_r^\top\bar\sigma_r
+ \Lambda^\top\Lambda
+ 2f_4(\tau)\bar\sigma_r^\top\Lambda, & f_3(0)=0.
\end{cases}
\end{equation}
The solutions to this system yield $f_3(\tau)$ and $f_4(\tau)$.

Finally, for the third expectation involving the contribution process $C_t$, governed by \eqref{eq:contribution-sde}, we conjecture
\[
\mathbb{E}_t[\phi_u C_u] = \phi_t C_t F_3(t, r_t)
\quad \text{where} \quad
F_3(t, r_t) = \exp[f_5(t,\tau) + f_6(\tau) r_t].
\]
The product rule applied to $\phi_t C_t F_3(t, r_t)$ implies the drift must vanish, generating the following ODEs:
\begin{equation}
\label{eq:ODE-f5-f6}
\begin{cases}
f_6'(\tau)+\kappa f_6(\tau)+1=0,  \\[2pt]
\partial_\tau f_5(t,\tau)=\kappa\bar r f_6(\tau)
+ \tfrac{1}{2} f_6(\tau)^2\bar\sigma_r^\top\bar\sigma_r
+ \xi_0(x+t)+\xi_1-\sigma_3^\top\Lambda
- f_6(\tau)\sigma_3^\top\bar\sigma_r
+ f_6(\tau)\Lambda^\top\bar\sigma_r, \\
f_5(t,0)=f_6(0)=0.
\end{cases}
\end{equation}
Integrating these equations produces the closed-form functions for $f_5(t,\tau)$ and $f_6(\tau)$, thereby completing the proof.
\end{proof}

\section{Proof of Proposition~\ref{prop:static-solution}}
\label{app:proof-prop-static-solution}

\begin{proof}[Proof of Proposition~\ref{prop:static-solution}]
The Lagrangian for the constrained optimization problem \eqref{eq:dynamic-objective}--\eqref{eq:dynamic-constraint} subject to \eqref{eq:budget-necessity} is defined as:
\begin{equation}
\label{eq:Lagrangian}
\begin{aligned}
\mathcal{L}
&= {}_T p_x \mathbb{E}[W_T^2]
+ \int_0^T {}_t p_x \mu_{x+t} \mathbb{E}[M_t^2] \, dt  - \lambda_1 \left(
{}_T p_x \mathbb{E}[W_T]
+ \int_0^T {}_t p_x \mu_{x+t} \mathbb{E}[M_t] \, dt
- \mathcal{K}
\right) \\
&\quad - \lambda_2 \left(
\mathbb{E}\left[
\int_0^T {}_t p_x \phi_t (\mu_{x+t} M_t - C_t) \, dt
+ {}_T p_x \phi_T W_T
\right] - w_0
\right),
\end{aligned}
\end{equation}
where $\lambda_1, \lambda_2 \in \mathbb{R}$ are the Lagrange multipliers. Applying the derivative, the first-order conditions (FOCs) with respect to $W_T$ and $M_t$ yield:
\begin{align}
\label{eq:FOC-WT}
\frac{\partial \mathcal{L}}{\partial W_T}
&= {}_T p_x \left( 2 W_T - \lambda_1 - \lambda_2 \phi_T \right) = 0
\quad \Rightarrow \quad
W_T^* = \frac{\lambda_1}{2} + \frac{\lambda_2}{2} \phi_T,\\
\label{eq:FOC-Mt}
\frac{\partial \mathcal{L}}{\partial M_t}
&= {}_t p_x \mu_{x+t}
\left( 2 M_t - \lambda_1 - \lambda_2 \phi_t \right) = 0
\quad \Rightarrow \quad
M_t^* = \frac{\lambda_1}{2} + \frac{\lambda_2}{2} \phi_t.
\end{align}
Substituting $W_T^*$ and $M_t^*$ back into the expectation constraint \eqref{eq:dynamic-constraint} and the static budget constraint \eqref{eq:budget-necessity}, and applying Lemma~\ref{lem:conditional-expectations} through the definitions of $A$, $B$, and $H$ in \eqref{eq:A-B-H-def}, yields a system of linear equations for $\lambda_1$ and $\lambda_2$. Solving this system provides the explicit expressions \eqref{eq:lambda1-star}--\eqref{eq:lambda2-star} for $\lambda_1^*$ and $\lambda_2^*$, which completes the proof.
\end{proof}

\section{Proof of Theorem~\ref{thm:optimal-strategies}}
\label{app:proof-thm-optimal-strategies}

\begin{proof}[Proof of Theorem~\ref{thm:optimal-strategies}]
The static budget constraint \eqref{eq:budget-necessity} implies that at any time $t \in [0, T]$, the conditional expectation satisfies
\begin{equation}
\label{eq:dynamic-budget-conditional}
\mathbb{E}_t \left[
\int_t^T {}_{u-t} p_{x+t} \mu_{x+u}
\frac{\phi_u}{\phi_t} M_u \, du
+ {}_{T-t} p_{x+t} \frac{\phi_T}{\phi_t} W_T
\right]
= W_t + \mathbb{E}_t \left[
\int_t^T {}_{u-t} p_{x+t}
\frac{\phi_u}{\phi_t} C_u \, du
\right].
\end{equation}
Substituting the FOCs \eqref{eq:WT-star}--\eqref{eq:Mt-star} from Proposition~\ref{prop:static-solution} into this conditional budget constraint gives:
\begin{equation}
\label{eq:Wt-star-start}
\begin{aligned}
W_t^*
&= \mathbb{E}_t \Bigg[
\int_t^T {}_{u-t} p_{x+t} \mu_{x+u}
\frac{\phi_u}{\phi_t}
\left(\frac{\lambda_1^*}{2} + \frac{\lambda_2^*}{2} \phi_u\right) \, du
+ {}_{T-t} p_{x+t}
\frac{\phi_T}{\phi_t}
\left(\frac{\lambda_1^*}{2} + \frac{\lambda_2^*}{2} \phi_T\right)
\Bigg] \\
&\quad - \mathbb{E}_t \left[
\int_t^T {}_{u-t} p_{x+t}
\frac{\phi_u}{\phi_t} C_u \, du
\right].
\end{aligned}
\end{equation}
Noting that
\[
\frac{\phi_u}{\phi_t} \phi_u = \phi_t \left(\frac{\phi_u}{\phi_t}\right)^2,
\]
we apply Lemma~\ref{lem:conditional-expectations} to rewrite the equation in terms of $A(t, r_t)$, $B(t, r_t)$, and $H(t, r_t)$ defined in \eqref{eq:A-B-H-def}, which simplifies directly to the explicit form \eqref{eq:Wt-star-dynamic} of $W_t^*$.

To derive $\pi_t^*$, we apply It\^o's formula to
\[
W_t^* = W^*(t, r_t, \phi_t, C_t).
\]
The diffusion term of $dW_t^*$ is:
\begin{equation}
\label{eq:diff-Wt-star}
\text{Diff}(dW_t^*) =
\left[
- \bar{\sigma}_r^\top \frac{\partial W_t^*}{\partial r}
- \phi_t \Lambda^\top \frac{\partial W_t^*}{\partial \phi}
+ C_t \sigma_3^\top \frac{\partial W_t^*}{\partial C}
\right] dZ_t.
\end{equation}
By matching this diffusion term with the continuous wealth dynamics $dW_t$ governed by $\pi_t^\top \Sigma_F dZ_t$ in \eqref{eq:wealth-sde-M}, we obtain
\[
\pi_t^\top \Sigma_F dZ_t = \text{Diff}(dW_t^*).
\]
Multiplying by $\Sigma_F^\top(\Sigma_F \Sigma_F^\top)^{-1}$ yields the optimal investment strategy $\pi_t^*$ in \eqref{eq:pi-star}.

Finally, the relation between the bequest and the wealth process is $M_t = W_t + I_t/\mu_{x+t}$ from \eqref{eq:bequest-M}. Therefore, the optimal insurance premium is
\[
I_t^* = \mu_{x+t} (M_t^* - W_t^*).
\]
Substituting $M_t^*$ from Proposition~\ref{prop:static-solution} and $W_t^*$ from \eqref{eq:Wt-star-dynamic} completes the proof of \eqref{eq:I-star}.
\end{proof}

\section{Proof of Theorem~\ref{thm:efficient-frontier}}
\label{app:proof-thm-efficient-frontier}

\begin{proof}[Proof of Theorem~\ref{thm:efficient-frontier}]
From the expectation constraint \eqref{eq:dynamic-constraint} and the static budget constraint \eqref{eq:budget-necessity}, we obtain
\begin{equation}
\label{eq:EF-constraints}
\mathbb{E}[Y_T] = \frac{\lambda_1^*}{2} + \frac{\lambda_2^*}{2} A(0, r_0) = \mathcal{K}, \quad
\frac{\lambda_1^*}{2} A(0, r_0) + \frac{\lambda_2^*}{2} B(0, r_0) = W_0^{\textnormal{eff}}.
\end{equation}
The expected squared payoff is
\begin{equation}
\label{eq:EF-second-moment}
\mathbb{E}[Y_T^2]
= \left(\frac{\lambda_1^*}{2}\right)^2
+ 2\left(\frac{\lambda_1^*}{2}\right)\left(\frac{\lambda_2^*}{2}\right) A(0, r_0)
+ \left(\frac{\lambda_2^*}{2}\right)^2 B(0, r_0).
\end{equation}
The variance is then:
\begin{equation}
\label{eq:EF-variance}
\text{Var}(Y_T)
= \mathbb{E}[Y_T^2] - (\mathbb{E}[Y_T])^2
= \left(\frac{\lambda_2^*}{2}\right)^2 \bigl(B(0, r_0) - A(0, r_0)^2\bigr).
\end{equation}
Solving the linear constraint system \eqref{eq:EF-constraints} for $\lambda_2^*$ yields
\begin{equation}
\label{eq:lambda2-star-alternative}
\frac{\lambda_2^*}{2}
= \frac{W_0^{\textnormal{eff}} - \mathcal{K} A(0, r_0)}{B(0, r_0) - A(0, r_0)^2}.
\end{equation}
Substituting this expression back into the variance equation \eqref{eq:EF-variance} provides:
\begin{equation}
\label{eq:EF-variance-final}
\text{Var}(Y_T)
= \frac{\bigl(W_0^{\textnormal{eff}} - \mathcal{K} A(0, r_0)\bigr)^2}{B(0, r_0) - A(0, r_0)^2}.
\end{equation}
Taking the square root and rearranging the equation in terms of the expected target $\mathcal{K}$ yields the specified efficient frontier \eqref{eq:efficient-frontier}. This completes the proof.
\end{proof}

\section{Proof of Theorem~\ref{thm:verification}}
\label{app:proof-thm-verification}

\begin{proof}[Proof of Theorem~\ref{thm:verification}]
The argument follows the martingale verification method (see Theorems 3.5 and 6.3 in \cite{karatzas1998methods}). For any real numbers $x$ and $y$,
\begin{equation*}
    x^2-2xy \ge -y^2,
\end{equation*}
with equality if and only if $x=y$. Since
\[
2W_T^*=\lambda_1^*+\lambda_2^*\phi_T,
\qquad
2M_t^*=\lambda_1^*+\lambda_2^*\phi_t,
\]
we have, pointwise,
\begin{align}
W_T^2-(\lambda_1^*+\lambda_2^*\phi_T)W_T
&\ge
(W_T^*)^2-(\lambda_1^*+\lambda_2^*\phi_T)W_T^*,\\
M_t^2-(\lambda_1^*+\lambda_2^*\phi_t)M_t
&\ge
(M_t^*)^2-(\lambda_1^*+\lambda_2^*\phi_t)M_t^* .
\end{align}
Multiplying the first inequality by ${}_T p_x$, the second by
${}_t p_x\mu_{x+t}$, integrating over $t\in[0,T]$, and taking expectations gives
\begin{equation}
\label{eq:verification-big}
\begin{aligned}
&\mathbb{E}\left[
{}_T p_x (W_T)^2
+ \int_0^T {}_t p_x \mu_{x+t} (M_t)^2 \, dt
\right] - \lambda_1^* \left(
{}_T p_x \mathbb{E}[W_T]
+ \int_0^T {}_t p_x \mu_{x+t} \mathbb{E}[M_t] \, dt
\right) \\
&\quad - \lambda_2^* \left(
{}_T p_x \mathbb{E}[\phi_T W_T]
+ \int_0^T {}_t p_x \mu_{x+t} \mathbb{E}[\phi_t M_t] \, dt
\right) \\
&\ge \mathbb{E}\left[
{}_T p_x (W_T^*)^2
+ \int_0^T {}_t p_x \mu_{x+t} (M_t^*)^2 \, dt
\right]  - \lambda_1^* \left(
{}_T p_x \mathbb{E}[W_T^*]
+ \int_0^T {}_t p_x \mu_{x+t} \mathbb{E}[M_t^*] \, dt
\right) \\
&\quad - \lambda_2^* \left(
{}_T p_x \mathbb{E}[\phi_T W_T^*]
+ \int_0^T {}_t p_x \mu_{x+t} \mathbb{E}[\phi_t M_t^*] \, dt
\right).
\end{aligned}
\end{equation}
Both $(W_T,M_t)$ and $(W_T^*,M_t^*)$ satisfy the expectation constraint and the static budget constraint. The $\lambda_1^*$ terms are therefore identical. For the $\lambda_2^*$ terms, the static budget constraint implies
\[
{}_T p_x\mathbb{E}[\phi_TW_T]
+\int_0^T {}_t p_x\mu_{x+t}\mathbb{E}[\phi_tM_t]\,dt
=
w_0+\mathbb{E}\left[\int_0^T {}_t p_x\phi_t C_t\,dt\right],
\]
which is also identical for both pairs. Hence all linear terms cancel, and
\[
\mathbb{E}[(Y_T)^2]\ge \mathbb{E}[(Y_T^*)^2].
\]
This completes the proof.
\end{proof}

\end{document}